\documentclass[11pt]{article}
\usepackage{times}
\usepackage{xspace,comment,calc}
\usepackage{amssymb,amsmath,amsthm,amsfonts,amstext}
\usepackage[margin=1cm]{caption}
\usepackage[shortlabels,inline]{enumitem}
\usepackage{multirow}
\usepackage{bbm}
\usepackage{comment}
\usepackage{url}
\usepackage{mdframed}
\usepackage{hyperref}
\usepackage{mathrsfs}

\usepackage{tikz}
\usepackage{pgfplots}

\usetikzlibrary{pgfplots.fillbetween}
\usetikzlibrary{patterns}
\pgfplotsset{compat=1.10}
\pgfdeclarelayer{ft}
\pgfdeclarelayer{bg}
\pgfsetlayers{bg,main,ft}

\makeatletter

\pgfdeclarepatternformonly[\LineSpace]{my north east lines}{\pgfqpoint{-1pt}{-1pt}}{\pgfqpoint{\LineSpace}{\LineSpace}}{\pgfqpoint{\LineSpace}{\LineSpace}}%
{
    \pgfsetcolor{\tikz@pattern@color}
    \pgfsetlinewidth{0.4pt}
    \pgfpathmoveto{\pgfqpoint{0pt}{0pt}}
    \pgfpathlineto{\pgfqpoint{\LineSpace + 0.1pt}{\LineSpace + 0.1pt}}
    \pgfusepath{stroke}
}

\pgfdeclarepatternformonly[\LineSpace]{my north west lines}{\pgfqpoint{-1pt}{-1pt}}{\pgfqpoint{\LineSpace}{\LineSpace}}{\pgfqpoint{\LineSpace}{\LineSpace}}%
{
    \pgfsetcolor{\tikz@pattern@color}
    \pgfsetlinewidth{0.4pt}
    \pgfpathmoveto{\pgfqpoint{0pt}{\LineSpace}}
    \pgfpathlineto{\pgfqpoint{\LineSpace + 0.1pt}{-0.1pt}}
    \pgfusepath{stroke}
}
\makeatother

\newdimen\LineSpace
\tikzset{
    line space/.code={\LineSpace=#1},
    line space=10pt
}

\hypersetup{colorlinks = true,linkcolor = blue, anchorcolor = blue, citecolor = blue, filecolor = red,urlcolor = red}
\usepackage[nameinlink]{cleveref}

\newcommand{\shP}{\#P\xspace}

\DeclareMathOperator*{\Exp}{E}
\newcommand{\prob}[2][{}]{\ensuremath{{\textstyle{\boldsymbol \Pr}_{#1}}\bigl[#2\bigr]}}
\newcommand{\E}[2][{}]{\ensuremath{{\textstyle{\boldsymbol \Exp}_{#1}}\bigl[#2\bigr]}}

\DeclareMathOperator{\supp}{supp}
\DeclareMathOperator{\conv}{conv}

\newcommand{\gap}{\enspace}
\newcommand{\gaptwo}{\gap \gap}

\newtheorem{theorem}{Theorem}[section]
\newtheorem{lemma}[theorem]{Lemma}
\newtheorem{claim}[theorem]{Claim}
\newtheorem{corollary}[theorem]{Corollary}

{\theoremstyle{definition} 
\newtheorem{definition}[theorem]{Definition}
}

{\theoremstyle{remark} \newtheorem{remark}{Remark}}

\newtheoremstyle{empty}
        {\topsep}{\topsep}              
        {\itshape}                      
        {}                              
        {\bfseries}                     
        {.}                             
        { }                             
        {\thmnote{\bfseries #3}}       
{\theoremstyle{empty} \newtheorem{duplicate}{noname}}

\newenvironment{proofof}[1]{\begin{proof}[Proof of #1]}{\end{proof}}

\newcommand{\bg}[1]{\medskip\noindent{\it #1}}

\newcommand{\paren}[1]{\ensuremath{\left(#1\right)}}
\newcommand{\curly}[1]{\ensuremath{\left\{ #1 \right\}}}

\newcommand{\mc}{\mathcal}
\newcommand{\R}{\ensuremath{\mathbb R}}
\newcommand{\Rp}{\ensuremath{\mathbb R_{\geq 0}}}

\newcommand{\Z}{\ensuremath{\mathbb Z}}
\newcommand{\Zp}{\ensuremath{\mathbb Z_{\geq 0}}}

\newcommand{\C}{\ensuremath{\mc{C}}}
\newcommand{\I}{\ensuremath{\mc I}}

\newcommand{\M}{\ensuremath{\mathcal M}}
\newcommand{\U}{\ensuremath{\mathcal U}}
\newcommand{\T}{\ensuremath{\mathcal T}}
\newcommand{\Q}{\ensuremath{\mc Q}}
\newcommand{\rank}{\ensuremath{\mathsf{r}}}

\newcommand{\OPT}{\ensuremath{\mathsf{OPT}}}
\newcommand{\iopt}{\ensuremath{\mathcal{O}^*}}

\newcommand{\val}{\ensuremath{\mathit{value}}}

\newcommand{\load}{\ensuremath{\mathsf{load}}}

\newcommand{\LB}{\ensuremath{\mathsf{LB}}}
\newcommand{\UB}{\ensuremath{\mathsf{UB}}}
\newcommand{\LP}[1]{\ensuremath{\mathsf{LP}}}

\newcommand{\bd}{\ensuremath{\mathsf{bd}}}
\newcommand{\low}{\ensuremath{\mathsf{low}}}
\newcommand{\hi}{\ensuremath{\mathsf{hi}}}
\newcommand{\excep}{\ensuremath{\mathsf{excep}}}

\newcommand{\sm}{\ensuremath{\setminus}} 

\newcommand{\sse}{\subseteq}

\newcommand{\ceil}[1]{\ensuremath{\left\lceil#1\right\rceil}}
\newcommand{\floor}[1]{\ensuremath{\left\lfloor#1\right\rfloor}}
\newcommand{\bigo}[1]{O\left( #1 \right)}

\newcommand{\tht}[1]{\Theta\left(#1\right)}

\newcommand{\poly}{\operatorname{poly}}

\newcommand{\Min}[1]{\min{\left\{#1\right\}}}

\newcommand{\POS}{\mathsf{POS}}
\newcommand{\POSl}{\POS^{<}}
\newcommand{\POSe}{\POS^{=}}

\newcommand{\norm}[1]{\left\lVert#1\right\rVert}

\newcommand{\ve}{\ensuremath{\varepsilon}} 
\newcommand{\gm}{\ensuremath{\gamma}} 
\newcommand{\ld}{\ensuremath{\lambda}} 
\newcommand{\kp}{\ensuremath{\kappa}}
\newcommand{\al}{\ensuremath{\alpha}}

\newcommand{\sg}{\ensuremath{\sigma}}

\newcommand{\bt}{\ensuremath{\overline t}}

\newcommand{\bY}{\ensuremath{\overline Y}}
\newcommand{\bz}{\ensuremath{\overline z}}

\newcommand{\qt}{\theta}
\newcommand{\eproxl}[1][\ell]{\ensuremath{\gm_{#1}}}
\newcommand{\problth}[1][\ell]{\ensuremath{\tau_{#1}}}
\newcommand{\bproblth}[1][\ell]{\ensuremath{t_{#1}}}

\newcommand{\bon}{\ensuremath{\mathbbm{1}}}
\newcommand{\down}{\ensuremath{\downarrow}}
\newcommand{\thresh}{\sg}
\newcommand{\scrv}{\ensuremath{Z}}
\newcommand{\vecrv}{\ensuremath{W}}
\newcommand{\sumrv}{\ensuremath{S}}
\newcommand{\treerv}{\ensuremath{Y}}

\newcommand{\tone}{A}
\newcommand{\budg}{B}
\newcommand{\bvec}{\ensuremath{\vec{b}}}
\newcommand{\bfun}{\ensuremath{b}}
\newcommand{\z}{\ensuremath{\eta}}

\makeatletter
\newcommand{\leqnomode}{\tagsleft@true\let\veqno\@@leqno}
\newcommand{\reqnomode}{\tagsleft@false\let\veqno\@@eqno}
\makeatother

\newcommand{\clbp}{\ensuremath{7}\xspace}
\newcommand{\cexpnorm}{\ensuremath{28}\xspace}
\newcommand{\ctauexpr}{\ensuremath{14}\xspace}

\newcommand{\pr}[1]{\ensuremath{\prob{#1}}} 

\newcommand{\eload}[1][i]{\xi_{\ensuremath{#1}}}
\newcommand{\beload}[1][i]{\overline{\xi}_{\ensuremath{#1}}}
\newcommand{\eloadl}[1][i,\ell]{\eload[i,\ell]} 
\newcommand{\topl}[1][\ell]{\ensuremath{\mathsf{Top}_{#1}}}
\newcommand{\Topl}[2][\ell]{\ensuremath{\topl[{#1}]\left(#2\right)}}
\newcommand{\taul}[1][\ell]{\ensuremath{\problth[#1]}} 
\newcommand{\Taul}[2][\ell]{\ensuremath{\problth[#1]\left(#2\right)}} 
\newcommand{\gaml}[1][\ell]{\ensuremath{\eproxl[#1]}} 
\newcommand{\Gaml}[2][\ell]{\ensuremath{\eproxl[#1]\left(#2\right)}} 
\newcommand{\erv}[2]{\ensuremath{#1^{\geq #2}}} 
\newcommand{\trv}[2]{\ensuremath{#1^{< #2}}} 
\newcommand{\eff}[2][\lambda]{\beta_{#1}\left(#2\right)}

\newcommand{\indset}[1][F]{\ensuremath{#1}}

\newcommand{\scalar}{\qt} 
\newcommand{\scal}{\scalar}

\newcommand{\viol}{\nu}
\newcommand{\sz}{\ensuremath{z^*}}

\newcommand{\ssg}{\ensuremath{\sg^*}}
\newcommand{\vt}{\ensuremath{\vec{t}}}
\newcommand{\vtp}{\ensuremath{\vec{t'}}}
\newcommand{\vtb}{\ensuremath{\vec{\bt}}}
\newcommand{\vts}{\ensuremath{\vec{t^*}}}
\newcommand{\LPv}[1][\vt]{\ensuremath{(\LP{}({#1}))}}
\newcommand{\LPTv}[1][\vt]{\ensuremath{(\tree{}({#1}))}}

\newcommand{\LPOPTv}[1][\vt]{\tree\text{-}\OPT\paren{#1}}

\newcommand{\rd}{\ensuremath{\mathsf{rd}}}
\newcommand{\ber}{\ensuremath{\mathsf{ber}}}
\newcommand{\lar}{\ensuremath{\mathsf{lg}}} 
\newcommand{\sml}{\ensuremath{\mathsf{sml}}} 
\newcommand{\mce}{\ensuremath{\mc{E}}}
\newcommand{\effm}[2][\lambda]{\beta'_{#1}\left(#2\right)}
\newcommand{\vol}{\ensuremath{\mathsf{vol}}} 
\newcommand{\Eb}[1]{\ensuremath{{\boldsymbol \Exp}\left[#1\right]}}

\newcommand{\tree}{\ensuremath{\mathsf{Tree}}}
\newcommand{\comp}{\ensuremath{\mathsf{comp}}}

\newcommand{\n}{\ensuremath{n-1}}
\newcommand{\jtoi}{{\ensuremath{j: j \mapsto i}}}

\newcommand{\sdleq}{\ensuremath{\preceq_{\,\mathrm{sd}}}}
\newcommand{\asgn}{\ensuremath{\mathsf{asgn}}}
\newcommand{\mst}{\ensuremath{\mathsf{tree}}}

\title{Approximation Algorithms for Stochastic Minimum Norm Combinatorial Optimization%
\thanks{An extended abstract is to appear in the Proceedings of the 61st FOCS, 2020.}} 
\author{
    Sharat Ibrahimpur\thanks{{\tt \{sharat.ibrahimpur,cswamy\}@uwaterloo.ca}.
    Dept. of Combinatorics and Optimization, Univ. Waterloo, Waterloo, ON N2L 3G1. 
    Supported in part by NSERC grant 327620-09 and an NSERC Discovery Accelerator
    Supplement Award.}
\and 
\addtocounter{footnote}{-1} 
    Chaitanya Swamy\footnotemark
}
\date{}

\begin{document}

\maketitle	
	
\begin{abstract}
Motivated by the need for, and growing interest in, modeling uncertainty in data, we
introduce and study {\em stochastic minimum-norm optimization}. 
We have an underlying combinatorial optimization problem where the costs involved
are {\em random variables} with given distributions; 
each feasible solution induces a random multidimensional cost vector, and given a   
certain objective function, 
the goal is to find a solution (that does not depend on the realizations of the costs)
that minimizes the expected objective value. 
For instance, in stochastic load balancing, jobs with random
processing times need to be assigned to machines, and the induced cost vector is the
machine-load vector. 
The choice of objective is typically the maximum- or sum- of the
entries of the cost vector, or in some cases some other $\ell_p$ norm of the cost vector. 
Recently, in the deterministic setting, Chakrabarty and
Swamy~\cite{ChakrabartyS19a} considered a much broader suite of objectives,
wherein we seek to minimize the $f$-norm of the cost vector under a given
{\em arbitrary monotone, symmetric norm} $f$.
In stochastic minimum-norm optimization, we work with this broad class of objectives,
and seek a solution that minimizes the {\em expected $f$-norm} of the induced cost vector.

The class of monotone, symmetric norms is versatile and 
includes $\ell_p$-norms, and $\topl$-norms (sum of $\ell$ largest coordinates in absolute
value), and enjoys various closure properties; in particular, 
it can be used to incorporate 
{\em multiple} norm budget constraints $f_\ell(x)\leq\budg_\ell$, $\ell=1,\ldots,k$. 

We give a general framework for devising algorithms for stochastic minimum-norm
combinatorial optimization, using which we obtain approximation algorithms for the
stochastic minimum-norm versions of the load balancing and spanning tree problems. We
obtain the following concrete results.

\begin{itemize}[noitemsep]
\item An $O(1)$-approximation for {\em stochastic minimum-norm load balancing on  unrelated machines} with:
(i) arbitrary monotone symmetric norms and job sizes that are Bernoulli random variables; and
(ii) $\topl$ norms and arbitrary job-size distributions.

\item An $\bigo{\log m/\log\log m}$-approximation for the general stochastic minimum-norm
load balancing problem, 
where $m$ is the number of machines.  

\item An $O(1)$-approximation for stochastic minimum-norm spanning tree with arbitrary
monotone symmetric norms and distributions; this guarantee extends to 
the stochastic minimum-norm matroid basis problem.
\end{itemize}

Two key technical contributions of this work are: 
(1) a structural result of independent interest 
connecting stochastic minimum-norm optimization to the simultaneous optimization of a
(\emph{small}) collection of expected $\topl$-norms; 
and 
(2) showing how to tackle expected $\topl$-norm minimization by leveraging techniques used
to deal with minimizing the expected maximum, circumventing the difficulties
posed by the non-separable nature of $\topl$ norms.
\end{abstract}

\section{Introduction} \label{sec:intro}
Uncertainty is a facet of many real-world decision environments, and
a thriving and growing area of optimization, {\em stochastic optimization}, 
deals with optimization under uncertainty. 
Stochastic load-balancing and scheduling problems, 
where we have uncertain job sizes (a.k.a processing times), constitute
a prominent and well-studied class of stochastic-optimization problems 
(see,
e.g.,~\cite{KleinbergRT00,GoelI99,GuptaKNS18,Pinedo04,MoehringSU99,ImMP15,GuptaMUX20}).
In {\em stochastic load balancing}, we are given job-size distributions
and we need to fix a job-to-machine assignment
{\em without knowing the actual processing-time realizations}. 
This assignment induces a {\em random} load vector, 
and its quality is assessed by evaluating the expected 
objective value of this load vector, which one seeks to minimize;
the objectives typically considered are: {\em makespan}---i.e., maximum
load-vector entry---which leads to the stochastic makespan minimization
problem~\cite{KleinbergRT00,GoelI99,GuptaKNS18},  
and the {\em $\ell_p$-norm of the load vector}~\cite{Molinaro19}. 
More generally, in a generic stochastic-optimization problem, we have
an underlying combinatorial optimization problem and the costs involved are 
described by random variables with given distributions. 
We need to take decisions, i.e., find a feasible solution, given only the distributional
information, and without knowing the realizations of the costs. 
({This is sometimes called {\em one-stage} stochastic optimization.})
Each feasible solution induces a random cost vector,  
we have an objective that seeks to quantitatively measure the quality of the solution
by aggregating the entries of the cost vector, and the goal is to find a feasible
solution that minimizes the expected objective value of the induced cost vector.
As another example in this setup, consider the {\em stochastic spanning tree} problem,
which is the following basic stochastic network-design problem: 
we have a graph with random edge costs and we seek a spanning tree of low expected
objective value, where the objective is applied to the cost vector that consists of the
costs of edges in the tree. 

The two most-commonly considered objectives in such settings (as also for deterministic
problems) are: (a) the {\em maximum} cost-vector entry   
(which adopts an egalitarian view); and (b) the {\em sum} of the cost-vector entries
(i.e., a utilitarian view that considers the total cost incurred).
These objectives give rise to various classical problems:
besides the makespan minimization 
problem in load balancing, 
other examples include (deterministic/stochastic) bottleneck spanning tree (max-
objective), minimum spanning tree (sum- objective),  
and the $k$-center (max- objective) and $k$-median (sum- objective) clustering problems.
Recognizing that the max- and sum- objectives tend to skew solutions in
different directions, 
other $\ell_p$-norms of the cost vector (e.g., $\ell_2$ norm) have also been considered in
certain settings as a means of interpolating between, or trading off, the max- and sum-
objectives.  
For instance, $\ell_p$-norms have been investigated for both deterministic and stochastic
load balancing~\cite{AwerbuchAGKKV95,AzarE05,MakarychevS18,Molinaro19} and for
deterministic $k$-clustering~\cite{GuptaT08}. 
Furthermore, very recently, Chakrabarty and Swamy~\cite{ChakrabartyS19a} introduced a rather general
model to unify these various problems (including minimizing $\ell_p$-norms), that they call 
{\em minimum-norm optimization}: 
given an \emph{arbitrary} monotone, symmetric norm $f$, 
find a solution that minimizes the $f$-norm of the induced cost vector.

\vspace*{-1ex}
\paragraph{Our contributions.}
In this work, we introduce and study {\em stochastic minimum-norm optimization}, 
which is the stochastic version of min-norm optimization: given a
stochastic-optimization problem and an arbitrary monotone, symmetric norm $f$, 
we seek a feasible solution that minimizes the 
{\em expected $f$-norm of the induced cost vector}.
As a model, this combines the versatility of
(deterministic) min-norm optimization with the more realistic setting of uncertain data,
thereby giving us a unified way of dealing with the various objective functions typically
considered for (deterministic and stochastic) optimization problems in the face of uncertain data.
We consider problems where there is a certain degree of independence in the underlying
costs, so that the components of the underlying cost vector are always {\em independent}
(and nonnegative) random variables.

{\em Our chief contribution is a framework that we develop for designing 
algorithms for stochastic minimum-norm combinatorial optimization problems, using which we  
devise approximation algorithms for the stochastic minimum-norm versions of load balancing 
(Theorems~\ref{thm:loadbalgen} and~\ref{thm:loadbalber}) and spanning trees
(Theorem~\ref{thm:tree}).} 
We only assume that we have a {\em value oracle} for the norm $f$; we remark that this is 
weaker than the optimization-oracle and first-order-oracle access required to $f$
in~\cite{ChakrabartyS19a} and~\cite{ChakrabartyS19b} respectively.

Stochastic minimum-norm optimization can be motivated from two distinct perspectives. 
The class of monotone, symmetric norms is quite rich and broad. In particular, it contains
all $\ell_p$-norms, as also another fundamental class of norms 
called {\em $\topl$-norms}: $\topl(x)$ is the sum of the $\ell$ largest coordinates of $x$
(in absolute value). Notice that $\topl$ norms provide another means of interpolating between
the min-max ($\topl[1]$) and min-sum ($\topl[m]$) problems (where $m$ is number of
coordinates). 
One of the motivating factors for~\cite{ChakrabartyS19a} to consider the (far-reaching)
generalization of arbitrary monotone, symmetric norms (in the deterministic setting) was
that it allows one to capture the optimization problems associated with these
various objectives under one umbrella, and thereby come up with a unified set of
techniques for handling these optimization problems. 
These same benefits also apply in the stochastic setting, 
making stochastic minimum-norm optimization an appealing model to study. 

Another motivation comes from the fact that the class of monotone, symmetric norms is
closed under various operations, including taking nonnegative combinations, and taking the
maximum of any finite collection of monotone, symmetric norms.  
A noteworthy and non-evident consequence of these closure properties is that they allow
us to incorporate budget constraints $f_\ell(x)\leq\budg_\ell$, $\ell=1,\ldots,k$ involving 
{\em multiple} monotone, symmetric norms $f_1,\ldots,f_k$ using the min-norm optimization
model: we can simply define another (monotone, symmetric) norm
$g(x):=\max\bigl\{{f_\ell(x)}/{\budg_\ell}: \ell=1,\ldots,k\bigr\}$, and the
(single) budget constraint $g(x)\leq 1$ can be captured by the problem of minimizing
$g(x)$. 
Multiple norm budget constraints may
arise, or be useful, for example, when 
no single norm may be a clear choice for assessing the solution quality.%
\footnote{Such considerations arise, for instance, when considering semi-supervised
learning on graphs using $\ell_p$-norm based Laplacian regularization, where different
choices of $p$ can lead to good solutions on different instances; see
e.g.,~\cite{AlaouiCRWJ16}.}
Moreover, such constraints, and, in particular, the above means of capturing them, can be
especially useful in stochastic settings, as they can provide us with more fine-grained
control of the underlying random cost vector, which can help offset the risk associated
with uncertainty; e.g., 
the constraint $\E{\max\bigl\{{\topl(Y)}/{\budg_\ell}: \ell=1,\ldots,k\bigr\}}\leq 1$, 
enforces a fair bit of control on the random cost vector $Y$ and provides safeguards   
against the costs being too high.

\medskip
To elaborate on our framework and results, it is useful to highlight and
appreciate two distinct high-level challenges that arise in dealing with stochastic
min-norm optimization. We delve into more details in Section~\ref{overview}.
First, how do we reason about the expectation of an {\em arbitrary} monotone, symmetric
norm? 
\cite{ChakrabartyS19a} prove the useful structural result that any monotone
symmetric norm $f$ can be expressed as the maximum of a collection of {\em ordered norms},
where an ordered norm is simply a nonnegative combination of $\topl$-norms
(Theorem~\ref{monnormthm}). 
While this does give a more concrete way of thinking about the expected $f$-norm, 
the challenge nevertheless in the stochastic setting is that one now needs to reason about
the expectation of the maximum of a collection of random variables, where each random
variable is the ordered norm of our cost vector $Y$. 
{\em One of our chief insights is that the expectation of the maximum of a collection of
ordered norms is within an $O(1)$-factor of the maximum of the expected ordered norms}
(Theorem~\ref{thm:expnorm}), 
i.e., interchanging the expectation and maximum operators only loses an $O(1)$ factor!
The crucial consequence is that this provides us with a ready 
means for reasoning
about $\E{f(Y)}$, namely by controlling $\E{\topl(Y)}$ for all indices $\ell$ (see
Theorem~\ref{stochmaj}). 
We believe that this structural result about the expectation of a monotone, symmetric norm
is of independent interest.

This brings us to the second challenge: how do we deal with a specific norm $f$,
such as the $\topl$ norm? To our knowledge, there is no prior work on any
stochastic $\topl$-norm minimization problem, and 
we need to control {\em all} expected $\topl$ norms.
Our approach is based on 
carefully identifying certain statistics of the random vector $Y$ that provide a
convenient handle on $\E{\topl(Y)}$ (see Section~\ref{proxy}).
(These statistics also play a role in establishing our above result on the expectation of
a monotone, symmetric norm.)
For a specific application (e.g., stochastic \{load balancing, spanning trees\}),
we formulate an LP encoding (loosely speaking) that the statistics of our random cost
vector match 
the statistics of the cost vector of an optimal solution. The main technical
component is then to devise a rounding algorithm that rounds the LP solution while losing
only a small factor in these statistics, and we utilize 
\nolinebreak\mbox{{\em iterative-rounding} ideas to achieve this.}

Combining the above ingredients leads to our approximation guarantees for stochastic
min-norm \{load balancing, spanning trees\}.
Our strongest and most-sophisticated results are for stochastic min-norm load balancing
with:
(i) arbitrary monotone symmetric norms and Bernoulli job sizes (Theorem~\ref{thm:loadbalber}); and
(ii) $\topl$ norms and arbitrary job-size distributions (Theorem~\ref{thm:loadbaltopl}); 
{\em in both cases, we obtain constant-factor approximations}. 
(We emphasize that we have
not attempted to optimize constants, and instead chosen to keep exposition simple and clean.)
We also obtain an $O(\log m/\log\log m)$-approximation for general stochastic min-norm
load balancing (Theorem~\ref{thm:loadbalgen}). 
We remark that dealing with Bernoulli distributions is often considered to be a stepping
stone towards handling general distributions (see, e.g.,~\cite{KleinbergRT00}), and so we
believe that our techniques will eventually lead to an constant-factor approximation for
(general) stochastic min-norm load balancing.
For stochastic spanning trees, wherein edge costs are random, we obtain an
$O(1)$-approximation for 
{\em arbitrary monotone symmetric norms and arbitrary distributions} (Theorem~\ref{thm:tree}).

\vspace*{-1ex}
\paragraph{Related work.}
As mentioned earlier, stochastic load balancing is a prominent combinatorial-optimization
problem that has been investigated in the stochastic setting under various $\ell_p$ norms.
Kleinberg, Rabani, and Tardos~\cite{KleinbergRT00} 
investigated {\em stochastic makesepan minimization} (i.e., minimize expected maximum load)
in the setting of identical machines (i.e., the processing time of a job is the same
across all machines), and were the first to obtain an $O(1)$-approximation for this
problem. We utilize the tools that they developed to reason about the expected makespan. 
Their results were improved for specific job-size distributions by~\cite{GoelI99}.  
Almost two decades later, 
Gupta et al.~\cite{GuptaKNS18} obtained the first the $\bigo{1}$-approximation for
stochastic makespan minimization on unrelated machines, and an 
$O\bigl(\frac{p}{\log p}\bigr)$-approximation for minimizing the expected $\ell_p$-norm of
the load vector.
Our $O(1)$-approximation result for all $\topl$-norms substantially generalizes the
makespan-minimization (i.e., $\topl[1]$-norm) result of~\cite{GuptaKNS18}. 
The latter guarantee was improved to a constant by Molinaro~\cite{Molinaro19} via a
careful and somewhat involved use of the so-called $L$-function method of~\cite{Latala97}. 

Our results and techniques are incomparable to those of~\cite{Molinaro19}. 
At a high level, Molinaro argues that $\E{\|Y\|_p}$, where $Y$ is the $m$-dimensional
machine-load vector, can be bounded by controlling the quantity 
$\bigl(\sum_{i\in[m]}\E{Y_i^p}\bigr)^{1/p}$, and uses a notion of effective size due to
Latala~\cite{Latala97} to obtain a handle on $\E{Y_i^p}$ in terms of the $X_{ij}$ random
variables of the jobs assigned to machine $i$. Essentially, ``effective size'' 
of a random variable maps the random variable to a deterministic quantity that one can
work with instead (but its definition and utility depends on a certain scale parameter).
Previously, for stochastic makespan minimization, Kleinberg et al.~\cite{KleinbergRT00}
and Gupta et al.~\cite{GuptaKNS18} utilizes a notion of 
effective size 
due to Hui~\cite{Hui88}, 
which helps in controlling tail bounds of the machine loads. 
Molinaro leverages the full power of Latala's notion of effective size 
by applying it at multiple scales, which allows him to obtain an $O(1)$-approximation
for $\ell_p$ norms.
A pertinent question that perhaps arises is: 
given the success yielded by the various notions of effective sizes for $\ell_\infty$ and
other $\ell_p$ norms, can one come up with a notion of effective size that one can use 
for a general monotone symmetric norm $f$?
This however seems unlikely. A concrete reason for this can be gleaned from the
modeling power of general monotone, symmetric norms that arises due to their closure
properties. Recall that one can encode multiple monotone, symmetric-norm budget
constraints 
via one monotone, symmetric norm $f$ (by taking a maximum of the scaled constituent norms). 
A notion of effective size for $f$ would (remarkably) translate to one deterministic
quantity that simultaneously yields some control for all the norms involved in the budget
constraints; this is unreasonable to expect, even when the constituent norms are $\ell_p$
norms. 

Examples of other well-known combinatorial optimization problems that have been
investigated in the stochastic setting include stochastic knapsack and bin
packing~\cite{KleinbergRT00,GoelI99,LiD19}, stochastic shortest paths~\cite{LiD19}. 
The works of~\cite{LiD19,LiY13,LiL16} consider expected-utility-maximization versions of
various combinatorial optimization problems.
In a sense, this can be viewed as a counterpart of stochastic min-norm optimization, where 
we have a {\em concave} utility function, and we seek to maximize the expected utility of
the underlying random value vector induced by our solution.
Their results are obtained by a clever discretization of the probability space; this
does not seem to apply to stochastic minimization problems.

$\topl$- and ordered-norms have been proposed in the location-theory literature, as a
means of interpolating between the $k$-center and $k$-median clustering problems, and
have been studied in the Operations Research literature~\cite{NickelP05,LaporteNG19}, but
largely from a modeling perspective.
Recently they have received much interest in the algorithms and optimization communities%
---partly, because $\topl$ norms yield an alternative (to $\ell_p$ norms) natural means of 
interpolating between the $\ell_\infty$- and $\ell_1$- objectives---%
and this work has led to strong algorithmic results for $\topl$-norm- and ordered-norm-
minimization in deterministic 
settings~\cite{AouadS19,AlamdariS17,ByrkaSS18,ChakrabartyS18,ChakrabartyS19a,ChakrabartyS19b}.

\section{Technical overview and organization} \label{sec:overview} \label{overview}
We discuss here the various challenges that arise in stochastic min-norm
optimization, and give an overview of the technical ideas we develop to overcome these
challenges  
{with pointers to the relevant sections for more details.}
We conclude in Section~\ref{concl} with a discussion of some open questions. 

As noted in~\cite{ChakrabartyS19a}, even for {\em deterministic} min-norm optimization,
the simple approach of minimizing $f(\vec{v})$, where $\vec{v}$ ranges over the cost of 
fractional solutions (e.g., convex combinations of integer solutions), fails badly since
this convex program often has large integrality gap. In the stochastic setting, a further
problem that arises with this approach is that the random variable $f(Y)$ will typically 
have exponential-size support (note that $Y$ follows a product distribution) 
making it computationally challenging to evaluate $\E{f(Y)}$. 

As noted earlier, we face two types of challenges 
in tackling stochastic min-norm
optimization. The first is posed by the generality of an arbitrary monotone, symmetric
norm. 
While the work of~\cite{ChakrabartyS19a} 
shows that $\topl$-norms are fundamental building blocks of monotone symmetric 
norms, and suggests a way forward for dealing with stochastic min-norm optimization,
as we elaborate below, the stochastic nature of the problem throws up various 
new issues, that create significant difficulties with leveraging the tools 
in~\cite{ChakrabartyS19a} to reason about 
{the {\em expectation} of a monotone symmetric norm.}

Second, stochastic min-norm optimization is complicated even for various {\em specific} 
norms. 
As mentioned earlier, 
$O(1)$-approximation algorithms were obtained only quite recently 
for stochastic min-norm load balancing under the $\ell_\infty$ norm (i.e., stochastic
makespan-minimization)~\cite{GuptaKNS18}, 
and other $\ell_p$ norms~\cite{Molinaro19}; 
there is no prior work on stochastic $\topl$-load balancing (or any other stochastic
$\topl$-norm optimization problem).  
The difficulties arise again due to the underlying stochasticity; even with the
$\ell_\infty$ norm, which is the most specialized of the 
aforementioned norms, one needs to bound the expectation of the maximum of a collection 
of random variables, a quantity that is not convenient to control. 
The importance of $\topl$ norms in the deterministic setting indicates that 
stochastic $\topl$-norm optimization is a key special case that one needs to understand,
but the non-separable nature of $\topl$ norms adds another layer of difficulty. 

\vspace*{-1ex}
\paragraph{Expectation of a monotone, symmetric norm (Section~\ref{expnorm}).}
Recall that Chakrabarty and Swamy~\cite{ChakrabartyS19a} show 
that for any monotone, symmetric norm $f:\R^m\mapsto\Rp$, there is a collection
$\C\sse\Rp^m$ of weight vectors with non-increasing coordinates
such that $f(x)=\max_{w\in\C} w^T x^\down$ for any $x\in\Rp^m$ (see
Theorem~\ref{monnormthm}), where 
$x^\down$ is the vector $x$ with coordinates sorted in non-increasing order.
The quantity $w^T x^\down$ 
is called an ($w$-) {\em ordered norm}, and can be expressed as 
a nonnegative combination $\sum_{\ell=1}^m(w_\ell-w_{\ell+1})\topl(x)$ of $\topl$
norms (where $w_{m+1}:=0$).
This structural result 
is quite useful 
for deterministic min-norm optimization, 
since it yields an easier-to-work-with objective, 
and more significantly, 
it {\em immediately} shows that controlling all $\topl$ norms 
suffices to control $f(x)$; 
these properties are leveraged by~\cite{ChakrabartyS19a} 
to devise approximation algorithms for deterministic min-norm optimization.

However, it is rather unclear how this structural result helps with the {\em stochastic}
problem. 
If $Y$ is the random cost vector of our solution (with independently-distributed
coordinates) 
one can now rewrite the objective function $\E{f(Y)}$ 
as $\E{\sup_{w\in\C}w^TY^\down}$ but dealing with the latter quantity
entails reasoning about the expectation of the maximum of a collection of (positively
correlated) random variables, which is a potentially onerous task.%
\footnote{The positive correlation between the random variables $w^TY^\down$ for $w\in\C$
renders the techniques used in stochastic makespan minimization inapplicable. In the
latter problem, one needs to bound the expectation of the maximum of a collection of {\em
independent} random variables
and the underlying techniques 
heavily exploit this independence.}  
Taking cue from the deterministic setting, 
a natural result to aim for in the stochastic setting is that, {\em bounding all  
expected $\topl$ norms  
enables one to bound $\E{f(Y)}$}. But unlike the deterministic setting, 
reformulating the objective as $\E{\sup_{w\in\C}w^TY^\down}$
does not yield any apparent dividends, and
it is not at all clear that such a result is actually true.
The issue again is that it is difficult to reason about $\E{\sup\ldots}$, and
interchanging the expectation and $\sup$ terms is not usually a viable option
(it is not hard to see that $\E{\max\{\scrv_1,\ldots,\scrv_k\}}$ may in general be
$\Omega(k)$ times $\max\{\E{\scrv_1},\ldots,\E{\scrv_k}\}$). 

One of our chief contributions is to prove that 
the analogue mentioned above for the stochastic setting {\em does indeed hold}, i.e.,
controlling $\E{\topl(Y)}$ for all $\ell\in[m]$ allows one to control $\E{f(Y)}$.
The key to this, and our main technical result here, is that, somewhat surprisingly and   
intriguingly, {\em $\E{f(Y)}=\E{\sup_{w\in\C}w^TY^\down}$ is at most a constant factor
larger than $\sup_{w\in\C}\E{w^TY^\down}$} (Theorem~\ref{thm:expnorm}). 
The quantity $\sup_{w\in\C}\E{w^TY^\down}=\sup_{w\in\C}w^T\E{Y^\down}$ has a nice
interpretation: it is simply $f\bigl(\E{Y^\down}\bigr)$, and so this can be restated
as $\E{f(Y)}=O(1)\cdot f\bigl(\E{Y^\down}\bigr)$.
This result 
provides us with the same mileage 
that~\cite{ChakrabartyS19a} obtain in the deterministic setting. 
Since 
$\E{w^TY^\down}$ is 
$\sum_{\ell=1}^m(w_\ell-w_{\ell+1})\E{\topl(Y)}$, as in the deterministic setting,
this immediately implies that controlling $\E{\topl(Y)}$ for all $\ell\in[m]$ allows
us to control $\E{f(Y)}$, thereby providing a foothold for reasoning about the   
fairly general stochastic min-norm optimization problem. 
In particular, we infer that 
if $\E{\topl(Y)}\leq\al\cdot\E{\topl(\vecrv)}$ for all $\ell\in[m]$ (and we can restrict to
$\ell$s that are powers of $2$ here), 
then $\E{f(Y)}=O(\al)\cdot\E{f(\vecrv)}$ (Theorem~\ref{stochmaj}). 
We believe that our structural result for the expectation of a monotone, symmetric
norm is of independent interest, and should
also find application in other stochastic settings
involving monotone, symmetric norms.

Our structural result showing that $\E{f(Y)}=O(1)\cdot f\bigl(\E{Y^\down}\bigr)$, 
is obtained by carefully exploiting the structure of monotone, symmetric norms. 
A key component of this is identifying suitable {\em statistics} of the random
vector $Y$, indexed by $\ell\in[m]$ 
such that: (a) the statistics for index $\ell$ lead to a convenient proxy function for 
estimating $\E{\topl(Y)}$ within $O(1)$ factors; 
(b) the statistics 
are related to the expectations of some random variables that are tightly concentrated
around their means; 
and (c) $\Pr\bigl[f(Y)>\thresh\cdot f\bigl(\E{Y^\down}\bigr)\bigr]$ is governed by the
probability that these random variables 
deviate from their means.  
Together these properties imply the desired bound on $\E{f(Y)}$.
Next, we elaborate on these statistics. 

\vspace*{-1ex}
\paragraph{Proxy functions and statistics (Section~\ref{proxy}).}
Since $f$ is a symmetric function, it is not hard to see that $f(Y)$ depends only on the  
{\em ``histogram''} $\bigl\{N^{>\qt}(Y)\bigr\}_{\qt\in\Rp}$, where $N^{>\qt}(Y)$ is the number of
coordinates of $Y$ larger than $\qt$. But this dependence is quite non-linear; its
precise form is given by the structural result for monotone, symmetric norms, and
by noting that we can write 
$\topl(Y)=\int_0^\infty\min\bigl\{\ell,N^{>\qt}(Y)\bigr\}d\qt$.
Despite these non-linearities, 
we show that the 
expected histogram curve $\bigl\{\E{N^{>\qt}(Y)}\bigr\}_{\qt\in\Rp}$ (see
Fig.~\ref{histogram} in Section~\ref{proxy}) controls $\E{f(Y)}$.

To show this, and also compress $\bigl\{\E{N^{>\qt}(Y)}\bigr\}_{\qt\in\Rp}$ 
to a finite, manageable number of statistics, we consider first the $\topl$-norm. While 
$\E{\topl(Y)}=\int_0^\infty\E{\min\{\ell,N^{>\qt}(Y)\}}d\qt$, 
interestingly, we prove that {\em this is only a constant-factor smaller than 
$\eproxl(Y):=\int_0^\infty\min\{\ell,\E{N^{>\qt}(Y)}\}d\qt$} (Theorem~\ref{eproxlthm});
i.e., interchanging $\Exp$ and $\min$ only leads to an $O(1)$-factor loss. 
Defining $\problth(Y)$ to be
the smallest $\qt$ such that $\E{N^{>\qt}(Y)}<\ell$, 
which can be viewed 
as an estimate of the $\ell$-th largest entry of $Y$, 
we can more compactly write 
$\eproxl(Y)=\ell\problth(Y)+\int_{\problth(Y)}^\infty\E{N^{>\qt}(Y)}d\qt$.
The statistics of interest to us are then the quantities $\problth=\problth(Y)$ (and 
implicitly $\E{N^{>\problth}(Y)}$) for $\ell=1,\ldots,m$. 
They enable us to bound the tail probability $\Pr[f(Y)>\thresh\cdot f(\E{Y^\down})]$ by
$\exp(-\Omega(\thresh))$ which implies that $\E{f(Y)}=O(1)\cdot f(\E{Y^\down})$.  
Roughly speaking, this follows because we show (by exploiting the proxy function
$\eproxl(Y)$) that $\Pr[f(Y)>\thresh\cdot f(\E{Y^\down})]$ is at most
the probability that $N^{>\problth}(Y)>\Omega(\thresh)\cdot\ell$ for some $\ell\in[m]$,
and 
Chernoff bounds show that $N^{>\problth}(Y)=\sum_{i=1}^m\Pr[Y_i>\problth]$ is tightly
concentrated around $\E{N^{>\problth}(Y)}<\ell$ (see Lemma~\ref{nbound}). 

\vspace*{-1ex}
\paragraph{Approximation algorithms for stochastic minimum-norm optimization:
load-balancing (Section~\ref{sec:loadbal}) and spanning tree 
(Section~\ref{sec:tree}).}
In Sections~\ref{sec:loadbal} and~\ref{sec:tree}, we apply our framework to design the
first approximation algorithms for the stochastic minimum-norm versions of load
balancing, and spanning tree (and matroid basis) problems respectively.
These sections can be read independently of each other.

Let $\iopt$ denote the random cost vector resulting from an optimal solution.
Applying our framework entails bounding $\E{\topl(Y)}$ in terms of
$\E{\topl(\iopt)}$ for all $\ell\in[m]$. For algorithmic tractability, we only
work with indices $\ell\in\POS=\POS_m:=\{2^i: i=0,1,\ldots,\floor{\log_2m}\}$:
since $\topl(\vec{v})\leq\topl[\ell'](\vec{v}) \leq 2\topl(\vec{v})$ holds for any
$\vec{v}\in\Rp^m$ and any $\ell\leq\ell'\leq 2\ell$, it is easy to see that with a 
factor-$2$ loss, this still yields a bound on $\E{\topl(Y)}/\E{\topl(\iopt)}$ for 
{\em all} $\ell\in[m]$. 
At a high level, 
we ``guess'' within, say a factor of $2$,
$\E{\topl(\iopt)}$ or certain associated quantities such as $\problth(\iopt)$, for all
$\ell\in\POS$. This guessing takes polynomial time since it involves enumerating a vector
with $O(\log m)$ {\em monotone} (i.e., non-increasing or non-decreasing) coordinates, each
of which lies in a logarithmically bounded range.
We then write an LP to obtain a fractional solution whose associated cost-vector
$\bY$ roughly speaking satisfies $\E{\topl(\bY)}\leq O\bigl(\E{\topl(\iopt)}\bigr)$ for
all $\ell\in\POS$. The chief technical ingredient is to devise a rounding procedure to
obtain an integer solution while preserving the $\E{\topl(.)}$-costs (up to some factor)
for all $\ell\in\POS$. We exploit {\em iterative rounding} to achieve this, by
capitalizing on the fact that (loosely speaking) our matrix of tight constraints has
bounded column sums (see Theorem~\ref{iterrndthm}).

For spanning trees (Section~\ref{sec:tree})---edge costs are random and we
seek to minimize the expected norm of the edge-cost vector of the spanning tree---%
the implementation of the above plan is quite direct. 
We guess $\problth^*:=\problth(\iopt)$ for all $\ell\in\POS$, and our LP imposes the
constraints $\E{N\bigl(\bY^{>\problth^*}\bigr)}\leq\ell$ for all $\ell\in\POS$. Since 
the coordinates of $\bY$ correspond to individual edge costs, and we know their
distributions, it is easy to impose the above constraint (see \ref{treelp}). Iterative
rounding works out quite nicely here since after normalizing the above constraints to
obtain unit right-hand-sides,
each column sum is $O(1)$. Thus, we obtain a 
{\em constant-factor approximation}; also, everything extends to the setting of
general matroids. 

\medskip
Our results for load balancing (Section~\ref{sec:loadbal}) are the most technically
sophisticated results in the paper. We obtain {\em constant-factor approximations} for:
(i) {\em arbitrary} norms with Bernoulli job processing times; and 
(ii) $\topl$-norm with {\em arbitrary} distributions. 
We also obtain an $O(\log m/\log\log m)$-approximation for the most general setting, where
both the monotone, symmetric norm and the job-processing-time distributions may be
arbitrary. 

The cost-vector $Y$ in load balancing corresponds to the (random) loads on the machines.
Each component $Y_i$ 
is thus an {\em aggregate} of some random variables: 
$Y_i=\sum_{j\text{ assigned to }i}X_{ij}$, where $X_{ij}$ is 
the (random) processing time of job $j$ on machine $i$. The complication that this creates  
is that we do not have direct access to (the distribution for) the random variable $Y_i$, 
making it difficult to calculate (or estimate) $\Pr[Y_i>\qt]$ (and hence
$\E{N(Y^{>\qt})}$). 
(In fact, this is \shP-hard, even for Bernoulli $X_{ij}$s (see~\cite{KleinbergRT00}), but
if we know the jobs assigned to $i$, then we can use a dynamic program to obtain a
$(1+\ve)$-approximation of $\Pr[Y_i>\qt]$.) 

We circumvent these difficulties by leveraging some tools from the work
of~\cite{KleinbergRT00,GuptaKNS18} on stochastic makespan minimization, 
in conjunction with an {\em alternate} proxy that we
develop for $\E{\topl(Y)}$. 
We show that, for a suitable choice of $\qt$, 
$\sum_{i=1}^m\E{Y_i^{\geq\qt}}$ is a good proxy for $\E{\topl(Y)}$ (see
Lemmas~\ref{lbproxupper} and~\ref{lbproxlower}), 
where $Y_i^{\geq\qt}$ is the random variable that is $0$ if $Y_i<\qt$ and $Y_i$ otherwise.  
Complementing this, the insight we glean from~\cite{KleinbergRT00,GuptaKNS18} is
that we can estimate $\E{Y_i^{\geq\qt}}$ within $O(1)$ factors using quantities  
that can be obtained from the distributions of the $X_{ij}$ random variables (see
Section~\ref{largevalbnd}). 
This involves 
analyzing the contribution of job $j$ (to $\E{Y_i^{\geq\qt}}$) differently based on whether
$X_{ij}$ is ``small'' (truncated) or ``large'' (exceptional), and utilizing the notion 
of {\em effective size} of a random variable~\cite{Hui88,KleinbergRT00} to bound the
contribution from small jobs. 
We guess the $\qt$ values---call
them $t^*_\ell$---corresponding to $\E{\topl(\iopt)}$ for all $\ell\in\POS$, and
write an LP for finding a fractional assignment where we enforce constraints 
encoding that these $t^*_\ell$ values are compatible with the correct guesses (see 
$\LPv$ in Section~\ref{sec:loadbalf}).

\section{Preliminaries} \label{sec:prelims} \label{prelims}
For an integer $m\geq 0$, we use $[m]$ to denote the set $\{1,\dots,m\}$. 
For any integer $m\geq 0$, we define $\POS_m=\{2^i:\, i\in\Zp, 2^i\leq m\}$; we drop the
subscript $m$ when it is clear from the context.
For $x \in \R$, define $x^+:=\max\{x,0\}$. 
For an event $E$ we use the indicator random variable $\mathbbm{1}_E$ to denote if event
$E$ happens. 
For any vector $x\in \Rp^m$ we use $x^\downarrow$ to denote the vector $x$
with its coordinates sorted in non-increasing order. 

Throughout, we use symbols $Y$ and $\vecrv$ to denote {\em random vectors}. 
The coordinates of these random vectors are always 
{\em independent}, nonnegative random variables. We denote this by
saying that the random vector 
follows a product distribution. 
For an event $\mce$, let $\mathbbm{1}_{\mce}$ be $1$ if $\mce$ happens, and $0$
otherwise. 
We reserve $\scrv$ to denote a scalar nonnegative random variable.
Given $\scrv$ and $\qt \in \Rp$, 
define 
the {\em truncated} random variable $\scrv^{<\qt}:=\scrv \cdot \bon_{\scrv<\qt}$, which
has support in $[0,\qt)$. Analogously, define the {\em exceptional} random variable 
$\scrv^{\geq\qt} := \scrv \cdot \bon_{\scrv\geq\qt}$, whose support lies in 
$\{0\} \cup [\qt,\infty)$. The following Chernoff bound will be useful.

\begin{lemma}[Chernoff bound] \label{chernoff}
Let $\scrv_1,\ldots,\scrv_k$ be independent $[0,1]$ random variables, and
$\mu\geq\sum_{j\in[k]}\E{\scrv_j}$. 
For any $\ve>0$, we have
$\prob{\sum_{j\in[k]}\scrv_j>(1+\ve)\mu}\leq\bigl(\frac{e^{\ve}}{(1+\ve)^{1+\ve}}\bigr)^{\mu}$. 
If $\ve>1$, then we also have the simpler bound 
$\prob{\sum_{j\in[k]}\scrv_j>(1+\ve)\mu}\leq e^{-\ve\mu/3}$.
\end{lemma}

\vspace*{-1ex}
\paragraph{Monotone, symmetric norms.}
A function $f:\R^m \mapsto \Rp$ is a \emph{norm} if it satisfies the following: 
(a) $f(x) = 0$ iff $x=0$; (b)
(homogeneity) $f(\lambda x) = |\lambda|f(x)$, for all $\lambda \in \R$, $x\in\R^m$; and 
(c) (triangle inequality) $f(x+y) \leq f(x) + f(y)$ for all $x,y\in\R^m$. 
Since our cost vectors are always nonnegative, we only consider nonnegative vectors in the 
sequel. 
A {\em monotone, symmetric} norm $f$ is a norm that satisfies: 
$f(x) \leq f(y)$ for all $0\leq x\leq y$ (monotonicity); and 
$f(x) = f(x^\down)$ for all $x\in\Rp^m$ (symmetry).
In the sequel, whenever we say norm, we always mean a monotone, symmetric norm.
We will often assume that $f$ is normalized, i.e., $f(1,0,\ldots,0) = 1$. 
We will use the following simple claim to obtain bounds on the optimal value.

\begin{claim} \label{prop:sandwich} \label{sandwich} \label{lem:sandwich}
For any $x \in \Rp^m$, we have $\max_{i \in [m]} x_i \leq f(x) \leq \sum_{i \in [m]} x_i$.
\end{claim}

\begin{proof}
For the lower bound observe that for any $i$, $f(x) \geq f(0,\dots,0,x_i,0,\dots,0) =
x_i$. Here we used that $f$ is monotone, symmetric, and homogeneous. Next, for the upper
bound we use triangle inequality: $f(x) \leq \sum_{i \in [m]} f(0,\dots,0,x_i,0,\dots,0) =
\sum_{i \in [m]} x_i$.  
\end{proof}

The following two types of monotone, symmetric norms will be especially important to us. 

\begin{definition} \label{defn:topl}
For any $\ell \in [m]$, the {\em $\topl$ norm} is defined as follows: 
for $x \in \Rp^m$, $\topl(x)$ is 
the sum of the $\ell$ largest coordinates of $x$, i.e., 
$\topl(x) = \sum_{i=1}^{\ell} x^\down_i$.

More generally, for any $w \in \Rp^m$ satisfying 
$w_1 \geq w_2 \geq \dots \geq w_m \geq 0$---we call such a $w$ a non-increasing vector---%
the {\em $w$-ordered norm} (or simply ordered norm) of a 
vector $x \in \Rp^m$ is defined as $\norm{x}_w:=w^T x^\down$. 
Observe that $\norm{x}_w=\sum_{\ell=1}^m(w_\ell-w_{\ell+1})\topl(x)$, where
$w_\ell:=0$ for $\ell>m$. 
\end{definition}

$\topl$-norm minimization yields a natural way of interpolating between the min-max
($\topl[1]$) and min-sum ($\topl[m]$) problems. 
The following result by
Chakrabarty and Swamy~\cite{ChakrabartyS19a}, which gives a foothold for working with an
arbitrary monotone, symmetric norm, further highlights their importance. 

\begin{theorem}[Structural result for monotone, symmetric norms~\cite{ChakrabartyS19a}] 
\label{monnormthm}
Let $f: \R^m \mapsto \Rp$ be a monotone, symmetric norm. 
\begin{enumerate}[(a), topsep=0.25ex, itemsep=0ex, leftmargin=*]
\item 
There is a collection $\C\sse\Rp^m$ of 
{non-increasing vectors 
such that $f(x)=\sup_{w\in\C} w^Tx^\down$ for all $x\in\Rp^m$.}
Hence, we have $\sup_{w\in\C}w_1=f(1,0,\ldots,0)$.

\item If $x,y\in\Rp^m$ are such that $\topl(x)\leq\al\topl(y)+\beta$ for all
$\ell\in[m]$, where $\al,\beta\geq 0$, then 
$f(x)\leq \al\cdot f(y)+\beta\cdot f(1,0,\ldots,0)$. 
\end{enumerate}
\end{theorem}

\begin{proof} 
Part (a) is quoted from~\cite{ChakrabartyS19a} (see Lemma 5.2 in~\cite{ChakrabartyS19a}).
Part (b) follows easily from part (a). 
Let $\C\sse\R^m$ be the collection of non-increasing weight vectors given by part (a) for
$f$. 
We have 
\begin{equation*}
\begin{split}
f(x) & =\sup_{w\in\C}w^Tx^\down=\sup_{w\in\C}\sum_{\ell\in[m]}(w_\ell-w_{\ell+1})\topl(x)
\\
& \leq \sup_{w\in\C}\sum_{\ell\in[m]}(w_\ell-w_{\ell+1})(\al\cdot\topl(y)+\beta) \\
& \leq \sup_{w\in\C}\sum_{\ell\in[m]}(w_\ell-w_{\ell+1})\al\cdot\topl(y)+
\sup_{w\in\C}\sum_{\ell\in[m]}(w_\ell-w_{\ell+1})\beta 
= \al\cdot f(y)+\sup_{w\in\C}w_1\cdot\beta \\
& = \al\cdot f(y)+\beta\cdot f(1,0,\ldots,0). \qedhere
\end{split}
\end{equation*}
\end{proof}

We will often need to enumerate vectors with monotone integer coordinates. 
We use the following standard result, lifted from~\cite{ChakrabartyS19a}, to obtain a
polynomial bound on the number of such vectors.

\begin{claim} \label{nondec} \label{monseq}
There are at most $(2e)^{\max\{M,k\}}$ non-increasing sequences of $k$ integers chosen
from $\{0,\ldots,M\}$. 
\end{claim}

\begin{proof} 
We reproduce the argument from~\cite{ChakrabartyS19a}.
A non-decreasing sequence $a_1\geq a_2\geq\ldots\geq a_k$, where $a_i\in\{0\}\cup [M]$ for all
$i\in[k]$, can be mapped bijectively to the set of $k+1$ integers 
$M-a_1,a_1-a_2,\ldots,a_{k-1}-a_k,a_k$ from $\{0\}\cup [M]$ that add up to $M$. The number
of such sequences of $k+1$ integers is equal to the coefficient of $x^M$ in the generating
function $(1+x+\ldots+x^M)^k$. This is equal to the coefficient of $x^M$ in $(1-x)^{-k}$,
which is $\binom{M+k-1}{M}$ using the binomial expansion. 
Let $U=\max\{N,k-1\}$. 
We have $\binom{M+k-1}{M}=\binom{M+k-1}{U}\leq\bigl(\frac{e(M+k-1)}{U}\bigr)^U\leq(2e)^U$.
\end{proof}

\paragraph{Iterative rounding.} 
Our algorithms are based on rounding fractional solutions to LP-relaxations that we
formulate for the stochastic min-norm versions of load balancing and spanning trees. 
The rounding algorithm needs to ensure that the various budget constraints
that we include in our LP to control quantities associated with expected $\topl$ norms (for
multiple indices $\ell$) are roughly preserved. The main technical tool involved in
achieving this is {\em iterative rounding}, as expressed by the following theorem, which
follows from a result in Linhares et al.~\cite{LinharesOSZ20}. 

\begin{theorem}[Follows from Corollary 11 in~\cite{LinharesOSZ20}] 
\label{thm:budgetedmatrounding} \label{iterrndthm}
Let $\M=(U,\I)$ be a matroid with rank function
$\rank$ (specified via a value oracle), 
and $\Q:=\{z\in\Rp^U:
z(U)=\rank(U),\ \ z(\indset)\leq\rank(\indset)\ \forall\indset\subseteq\U\}$ be its base
polytope. 
Let $\bz$ be a feasible solution to the following multi-budgeted matroid LP: 
\begin{equation} 
\tag{Budg-LP} \label{eq:iterlp}
\min \ \ c^T z \qquad \text{s.t.} \qquad 
A z \leq b, \quad z\in\Q.
\end{equation}
where $A \in \Rp^{k \times U}$. 
Suppose that for any $e \in \supp(\bz)$, we have $\sum_{i \in [k]} A_{i,e} \leq \viol$
for some parameter $\viol$. 
We can round $\bz$ in polynomial time to obtain a basis $B$ of $\M$ 
satisfying: 
(a) $c(B) \leq c^T \bz$; 
(b) $A\chi^B \leq b + \viol \mathbf{1}$, where $\mathbf{1}$ is the vector of all
$1$s; and
(c) $B$ is contained in the support of $\bz$.
\end{theorem}

\begin{proof}
We first consider a new instance where we move to the support of $\bz$, which will
automatically take care of (c). 
More precisely, let $J=\supp(\bz)$. For a vector $v\in\R^U$, let $v_J:=(v_e)_{e\in J}$
denote the restriction of $v$ to the coordinates in $J$. Let $\M_J=(J,\I_J)$ with rank
function $\rank_J$, be the restriction of the matroid $\M$ to $J$. 
Let $A_J$ be $A$ restricted to the columns corresponding to $J$. 
Note that $\rank_J(J)=\rank(J)\geq\bz(J)=\bz(U)=\rank(U)$; so we have 
$\rank(J)=\bz(J)=\rank(U)$, and therefore a basis of $\M_J$ is also a basis of $\M$. 
It is easy to see now that $\bz_J$ is a feasible solution to \eqref{eq:iterlp} where we
replace $c$ and $A$ by $c_J$ and $A_J$ respectively, and replace $\Q$ by the base polytope
of $\M_J$.
It suffices to show how to round $\bz_J$ to obtain a basis $B$ of $\M_J$ satisfying
properties (a) and  (b) (i.e., $c_J(B)\leq c_J^T\bz_J$ and 
$A_J\chi^B\leq b+\viol\mathbf{1}$), since, by construction, we have $B\sse J$. 
Note that each column of $A_J$ sums to at most $\viol$. 

We now describe how Corollary 11 in~\cite{LinharesOSZ20} yields the desired rounding of $\bz_J$.
This result pertains to a more general setting, where we have a fractional point in the
the base polytope of one matroid that satisfies matroid-independence constraints for some
other matroids, and some knapsack constraints. Translating Corollary 11 to our setting
above, where we have only one matroid, yields the following:

\begin{enumerate}[label={}]
\item {\bf Corollary 11 in~\cite{LinharesOSZ20} in our setting:}\ 
Let $A'$ be obtained from $A_J$ by scaling each row so that $\max_{e\in J}A'_{ie}=1$ for
all $i\in[k]$. Let $p_1,\ldots,p_k\geq 0$ be such that
$\sum_{i\in[k]}\frac{A'_{ie}}{p_i}\leq 1$ for all $e\in J$.%
\footnote{In~\cite{LinharesOSZ20}, the $p_i$s are stated to be positive integers,
but the proof of Corollary 11 shows that this is not actually needed.}
Then, we can round $\bz_J$ to
obtain a basis $B$ of $\M_J$ such that $c_J(B)\leq c_J^T\bz_J$, and 
$\sum_{e\in B}A_{ie}\leq b_i+p_i\max_{e\in J}A_{ie}$ for all $i\in[k]$.
\end{enumerate}

Setting $p_i=\frac{\viol}{\max_{e\in J}A_{ie}}$ for all $i\in[k]$ satisfies the conditions
above, since $\sum_{i\in[k]}\frac{A'_{ie}}{p_i}=\sum_{i\in[k]}A_{ie}/\viol\leq 1$ for all
$e\in J$, and yields the desired rounding of $\bz_J$.
\end{proof}

\subsection{Bounding \boldmath $\E{\sumrv^{\geq\qt}}$ when $\sumrv$ is the sum of
  independent random variables}
\label{subsec:expectedmax} \label{largevalbnd} 
In stochastic load balancing, each coordinate of the random cost vector is a ``composite''
random variable that is the 
sum of independent random variables. In such settings, where we 
do not have direct access to the distributions of individual coordinates, we 
develop a proxy function for estimating $\E{\topl(Y)}$ 
that involves computing $\E{Y_i^{>\qt}}$ for $i\in[m]$. We discuss here to compute
$\E{\sumrv^{>\qt}}$ for a composite random variable $\sumrv=\sum_{j\in[k]}\scrv_j$, where
$\scrv_1,\ldots,\scrv_k\geq 0$ are independent random variables. 
By gleaning suitable insights from~\cite{KleinbergRT00,GuptaKNS18}, we show how
to estimate this quantity given access to the distributions of the $\scrv_j$ random
variables.   
First, observe that for any $j \in [k]$, if $\scrv_j\geq\qt$ then we also have
$\sumrv\geq\qt$. So after some simple manipulations, we can show that 
\begin{equation}
\E{\sumrv^{\geq\qt}} = 
\Theta\biggl(\sum_{j \in [k]} \E{\scrv_j^{\geq\qt}} + 
\Exp\Bigl[\bigl(\sum_{j \in [k]} Z_j^{<\qt}\bigr)^{\geq\qt}\Bigr]\biggr).
\label{expgeqbnd}
\end{equation}
The first term above is easily computed from the $\scrv_j$-distributions. 
So to get a
handle on $\E{\sumrv^{\geq\qt}}$ it suffices 
control the second term, which
measures the contribution from the sum of the \emph{truncated} random variables
$\scrv_j^{<\qt}$ to $\E{\sumrv^{\geq\qt}}$. This requires a nuanced 
notion called \emph{effective size}, a concept that originated in queuing
theory~\cite{Hui88}. 

\begin{definition} \label{defn:effectivesize}
For a nonnegative random variable $\scrv$ and a parameter $\lambda > 1$, the
$\lambda$-effective size $\beta_\lambda(\scrv)$ of $\scrv$ is defined as 
$\log_\lambda\E{\lambda^{\scrv}}$. We also define $\beta_1(\scrv) := \E{\scrv}$. 
\end{definition}

The usefulness of effective sizes follows from Lemmas~\ref{smallbeta} and \ref{largebeta},
which indicate that $\E{\bigl(\sum_j\scrv_j^{<\qt}\bigr)^{\geq\qt}}$ 
can be estimated
by controlling the effective sizes of 
{some random variables related to the $\scrv_j^{<\qt}$
random variables.}  

\begin{lemma} \label{lem:smallbetasmalltail} \label{smallbeta}
Let $\scrv$ be a nonnegative random variable and $\lambda \geq 1$. 
If $\beta_\lambda(\scrv) \leq b$, then for any $c \geq 0$, we have 
$\prob{\scrv \geq b+c} \leq \lambda^{-c}$. 
Furthermore, if $\lambda \geq 2$, then $\E{\scrv^{\geq \beta_\lambda(\scrv)+1}}
\leq \frac{\beta_\lambda(\scrv)+3}{\lambda}$. 
\end{lemma}

\begin{proof}
The first part of the lemma trivially holds for $\lambda = 1$ so assume that $\lambda >
1$. From the above definition, $\lambda^{\beta_\lambda(\scrv)} = \E{\lambda^{\scrv}}$. By Markov's
inequality, 
\[
\prob{\scrv \geq b+c} = \prob{\lambda^{\scrv} \geq \lambda^{b+c}} 
\leq \frac{\E{\lambda^{\scrv}}}{\lambda^{b+c}} \leq \lambda^{-c}\,.
\]
For the second part, we have 
$\E{\scrv^{\geq\beta_\ld(\scrv)+1}}=(\beta_\ld(\scrv)+1)\Pr[\scrv\geq\beta_\ld(\scrv)+1]
+\int_{\beta_\ld(\scrv)+1}^\infty\Pr[\scrv>\sg]d\sg$.
Using the first part, the first term is at most $\frac{\beta_\ld(\scrv)+1}{\ld}$ 
and the second term is
$\int_1^\infty\Pr[\scrv>\beta_\ld(\scrv)+c]dc\leq\int_1^\infty\ld^{-c}dc
\leq\frac{1}{\ld\ln\ld}\leq\frac{2}{\ld}$, where 
the last inequality is because $\ld\geq 2$.
\end{proof}

Noting that $\beta_\ld(\scrv+\scrv')=\beta_\ld(\scrv)+\beta_\ld(\scrv')$ for independent
random variables $\scrv$ and $\scrv'$, we can infer
from Lemma~\ref{smallbeta} that if 
$\sum_{j\in[k]}\beta_\ld(\scrv_j^{<\qt}/\qt)=O(1)$, then
$\E{\bigl(\sum_{j\in[k]}\scrv_j^{<\qt}/\qt)^{>\Omega(1)}}\leq O(1)/\ld$, or equivalently
$\E{\bigl(\sum_{j\in[k]}\scrv_j^{<\qt})^{>\Omega(\qt)}}\leq O(\qt)/\ld$.

A key contribution of Kleinberg et al.~\cite{KleinbergRT00} is encapsulated 
by Lemma~\ref{largebeta} below, which is obtained by combining various
results from their work, and complements the above upper bound.
It lower bounds $\E{\sumrv^{\geq\qt}}$, 
where $\sumrv$ is the sum of independent, {\em bounded} random variables (e.g., truncated
random variables), in terms of the sum of the $\ld$-effective sizes
of these random variables.
The proof of this lemma is somewhat long and technical and is deferred to
Section~\ref{append-largebeta}. 

\begin{lemma} \label{lem:largebetalargetail} \label{largebeta}
Let $\sumrv = \sum_{j \in [k]} \scrv_j$, where the $\scrv_j$s are independent
$[0,\qt]$-bounded random variables. Let $\ld\geq 1$ be an integer. 
Then, 
\[
\E{\sumrv^{\geq \qt}} \geq \qt \cdot \frac{\sum_{j \in [k]}
\beta_\lambda\paren{\scrv_j/4\qt} - 6}{4 \ld} 
\]
\end{lemma}

\section{Proxy functions and statistics for controlling \boldmath $\E{\topl(Y)}$ and $\E{f(Y)}$} 
\label{sec:proxy} \label{proxy}
In this section, we identify various statistics of a random vector $Y$ following a
product distribution on $\R^m$, which lead to two proxy functions that provide a 
convenient handle on $\E{\topl(Y)}$ (within $O(1)$ factors). 

The first proxy function uses $\sum_{i \in [m]} \E{Y_i^{\geq\qt}}$ as a
means to control $\E{\topl(Y)}$. Roughly speaking, if $\qt$ is such that 
$\sum_{i\in[m]}\E{Y_i^{\geq\qt}}=\ell\qt$, then we argue that $\ell\qt$ is a good proxy for
$\E{\topl(Y)}$; 
Lemmas~\ref{lem:upperboundtopl} and~\ref{lem:lowerboundtopl} make this statement precise. 
This proxy is helpful in settings where each $Y_i$ is a sum of independent random
variables, 
such as stochastic min-norm load balancing (Section~\ref{sec:loadbal}), wherein $Y_i$
denotes the load on machine $i$. 
The second proxy function is based on a statistic that we denote $\problth(Y)$, 
which aims to capture
the (expected) $\ell$-th largest entry of $Y$. 
This statistic is defined using the {\em expected histogram curve} 
$\{\E{N^{>\qt}(Y)}\}_{\qt\in\Rp}$ (see Fig.~\ref{histogram}), where $N^{>\qt}(Y)$ is the
number of coordinates of $Y$ that are larger than $\qt$, and we show that this 
leads to an effective proxy $\eproxl(Y)$ for $\E{\topl(Y)}$ (Theorem~\ref{eproxlthm}). 
Collectively, these statistics and the $\eproxl(Y)$ proxies (for all $\ell\in[m]$) 
are 
instrumental in showing that the expected histogram curve 
controls $\E{f(Y)}$, and hence 
that $\E{f(Y)}$ can be bounded in terms of the expected $\topl$ norms of $Y$. 
This follows because the $N^{>\problth(Y)}(Y)$ random variables enjoy nice concentration
properties, 
and their deviations from the mean govern the tail probability of $f(Y)$ (see
Lemma~\ref{nbound}). 
Also, these statistics are quite convenient to work with in designing algorithms for problems   
where the $Y_i$s are ``atomic'' random variables with known distributions, such as
stochastic min-norm spanning trees (see Section~\ref{sec:tree}).  

\paragraph{Proxy function based on \boldmath $\E{Y_i^{\geq\qt}}$ statistics.}

\begin{lemma} \label{lem:upperboundtopl} \label{lbproxupper}
For any $\qt\geq 0$ if $\sum_{i \in [m]} \E{Y_i^{\geq\qt}} \leq \ell \qt$ holds, then $\E{\topl(Y)} \leq 2 \ell \qt$. 
\end{lemma}

\begin{proof}
We have $Y_i=Y_i^{<\qt}+Y_i^{\geq\qt}$ for all $i\in[m]$. So
\[
\E{\topl(Y)} \leq \E{\topl(Y_1^{<\qt},\dots,Y_m^{<\qt})} +
\E{\topl(Y_1^{\geq\qt},\dots,Y_m^{\geq\qt})} 
\leq \ell\qt + \sum_{i \in [m]} \E{Y_i^{\geq\qt}} \leq 2 \ell\qt \ .
\]
The first inequality is due to the triangle inequality; 
\nolinebreak\mbox{the second is 
because each $Y_i^{<\qt}$ is at most $\qt$.} 
\end{proof}

\begin{lemma} \label{lem:lowerboundtopl} \label{lbproxlower}
For any $\qt\geq 0$ if $\sum_{i \in [m]} \E{Y_i^{\geq\qt}}>\ell\qt$ holds, then
$\E{\topl(Y)}>\ell\qt/2$. 
\end{lemma}

\begin{proof}
The proof is by induction on $\ell+m$. The base case is when $\ell = m = 1$, where 
clearly  
$\E{\topl[1](Y_1)} = \E{Y_1} \geq \E{Y_1^{\geq\qt}}>\qt$.
Another base case that we consider is when $m\leq\ell$: here the $\topl$-norm is simply
the sum of all the $Y_i$s and thus, 
\[
\E{\topl\paren{Y_1,\dots,Y_m}} = \E{Y_1+\dots+Y_m} 
\geq \sum_{i \in [m]} \E{Y_i^{\geq\qt}} > \ell\qt \boldsymbol{>} \ell\qt/2 \ .
\]
Now consider the general case with $m \geq \ell+1$. 
Our induction strategy will be the following.
If $Y_m$ is at least $\qt$, then we include $Y_m$'s contribution
towards the $\topl$-norm and collect the expected $\topl[\ell-1]$-norm from the remaining
coordinates. Otherwise, we simply collect the expected $\topl$-norm from the remaining
coordinates. We use $Y_{-m}$ to denote the vector $(Y_1,\ldots,Y_{m-1})$. 

We handle some easy cases separately.
The case $\E{Y_m^{\geq\qt}}>\ell\qt$ 
causes a hindrance to applying the induction hypothesis. 
But this case is quite easy: we have
\[
\E{\topl(Y_1,\dots,Y_m)} \geq \E{Y_m} \geq \E{Y_m^{\geq\qt}} > \ell\qt \boldsymbol{>} \ell\qt/2 \ .
\]
So assume that $\E{Y_m^{\geq\qt}} \leq\ell\qt$. Let $q := \prob{Y_m \geq \qt}$. 
Another easy case that we handle separately is when $q=0$. Here, we have 
$\sum_{i\in[m-1]}\E{Y_i^{\geq\qt}}=\sum_{i\in[m]}\E{Y_i^{\geq\qt}}$, and we have
$\E{\topl(Y_1,\ldots,Y_m)}\geq\E{\topl(Y_{-m})}>\ell\qt/2$, where the last inequality is
due to the induction hypothesis.

So we are left with the case where $m\geq\ell+1$, $\E{Y_m^{\geq\qt}} \leq\ell\qt$ and 
$q=\prob{Y_m \geq \qt}>0$.  
Let $s :=\E{Y_m | Y_m \geq \qt}$, which is well defined, and is at least $\qt$. 
Note that $qs=\E{Y_m^{\geq\qt}}$.
We define two thresholds $\qt_1,\qt_2 \in (0,\qt]$
to apply the induction hypothesis to smaller cases. 
\[
\qt_1 := \frac{\ell\qt - \E{Y_m^{\geq\qt}}}{\ell} \gaptwo \text{ and } \gaptwo \qt_2 := \Min{\qt,\frac{\ell\qt-\E{Y_m^{\geq\qt}}}{\ell-1}}
\]
Noting that $\E{Y_i^{\geq t}}$ is a non-increasing function of $t$, observe that:
\begin{enumerate}[label=(C\arabic*), topsep=0.25ex, itemsep=0ex, leftmargin=*]
\item \label{indn1} $\sum_{i \in [m-1]} \E{Y_i^{\geq \qt_2}} > (\ell-1) \qt_2$: 
since $\qt_2 \leq \qt$, we have 
$\sum_{i \in [m-1]} \E{Y_i^{\geq \qt_2}} \geq 
\sum_{i\in [m-1]} \E{Y_i^{\geq\qt}} > \paren{\ell\qt - \E{Y_m^{\geq\qt}}}
\geq(\ell-1)\qt_2$.

\item \label{indn2} $\sum_{i \in [m-1]} \E{Y_i^{\geq \qt_1}} > \ell\qt_1$:
since $\qt_1 \leq \qt$, we have 
$\sum_{i \in [m-1]} \E{Y_i^{\geq \qt_1}} \geq 
\sum_{i \in [m-1]} \E{Y_i^{\geq\qt}} > \ell\qt - \E{Y_m^{\geq\qt}} = \ell\qt_1$.
\end{enumerate}  
We now have the following chain of inequalities.
\begin{alignat*}{2}
\E{\topl\paren{Y_1,\dots,Y_m}} & \geq 
q \paren{s + \E{\topl[\ell-1]\paren{Y_1,\dots,Y_{m-1}}}} + (1&&-q) \E{\topl(Y_1,\dots,Y_{m-1})} \\ 
& \boldsymbol{>} q \paren{ s + \frac{(\ell-1)\qt_2}{2} } + \frac{(1-q) \ell\qt_1}{2} 
\tag{induction hypothesis, using~\ref{indn1} and~\ref{indn2}} \\
& = qs + \frac{q}{2} \cdot \Min{(\ell-1)\qt,\ell\qt-\E{Y_m^{\geq\qt}}} && 
+ \frac{1-q}{2} \paren{\ell\qt - \E{Y_m^{\geq\qt}}} \\
& = \frac{\ell\qt}{2} + \frac{qs}{2}+\frac{q}{2}\cdot\Min{qs-\qt,0}
\qquad \qquad && \text{(since $qs = \E{Y_m^{\geq\qt}}$)} \\
& \geq \frac{\ell\qt}{2} + \frac{qs}{2} - \frac{q\qt}{2} \\
& \geq \ell\qt / 2. 
&& \text{(since $s\geq\qt$)} \qedhere
\end{alignat*}
\end{proof}

\paragraph{Proxy function modeling the $\ell$th largest coordinate.}

Our second proxy function is derived by the simple, but quite useful, observation that
for any $x \in \Rp^m$, we have $\topl(x) = \int_0^{\infty}\min\bigl\{\ell,N^{>\qt}(x)\bigr\} d\qt$,
where $N^{>\qt}(x)$ is the number of coordinates of $x$ that are greater than $\qt$.  
Noting that $\E{\min\{\ell,N^{>\qt}(Y)\}}\leq\min\bigl\{\ell,\E{N^{>\qt}(Y)}\bigr\}$, we
therefore obtain that
$\E{\topl(Y)}\leq\eproxl(Y):=\int_0^\infty\min\{\ell,\E{N^{>\qt}(Y)}\}d\qt$ (see Fig.~\ref{histogram}).
Interestingly, we show that $\eproxl(Y)=O(1)\cdot\E{\topl(Y)}$. 

\begin{figure}[ht!]
\scalebox{0.75}{
\begin{tikzpicture}

\draw[line width=0mm,white,name path=dummyx1] (0,0) -- (7.5,0);
\draw[line width=0mm,white,name path=dummyx2] (7.5,0) -- (9.5,0);
\draw[line width=0mm,white,name path=dummyx3] (0,0) -- (9.5,0);

\draw[->,line width=0.4mm] (0,0) -- (10,0); 
\draw[->,line width=0.4mm] (0,0) -- (0,6); 

\node(xaxis) at (9.8,-0.3) {\scalebox{1.2}{$\qt$}};
\node(yaxis) at (-1.2,4.2) {\scalebox{1.2}{$\E{N^{>\theta}(Y)}$}};
\node(origin) at (-0.2,-0.2) {\scalebox{1.2}{$\mathbf{0}$}};
\node(ym) at (-0.4,5) {\scalebox{1.2}{$m$}};

\node(gammal) at (7,5.25) {\scalebox{1.2}{$\gaml$}};
\node(gammalm1) at (7.15,4.5) {\scalebox{1.2}{$\gaml[\ell-1]$}};
\node(excep) at (8.1,3.85) {$\sum_i \E{\bigl(Y_i-\taul[1]\bigr)^+}$};

\draw[pattern=north east lines, pattern color=black] (6,5) rectangle (6.5,5.5);
\draw[pattern=north west lines, pattern color=black] (6,4.3) rectangle (6.5,4.8);
\draw[pattern=dots, pattern color=black] (6,3.6) rectangle (6.5,4.1);

\draw[domain=0:9.5,smooth,variable=\x,white,line width=0.1mm,name path=histo2] plot ({\x},{min(2.36,5*(exp(-\x/4)))});
\draw[domain=0:9.5,smooth,variable=\x,white,line width=0.1mm,name path=histo3] plot ({\x},{min(1.84,5*(exp(-\x/4)))});
\draw[domain=7.5:9.5,smooth,variable=\x,white,line width=0.1mm,name path=histo4] plot ({\x},{min(0.767,5*(exp(-\x/4)))});

\draw[domain=0:9.5,smooth,variable=\x,blue,line width=0.4mm,name path=histo1] plot ({\x},{5*(exp(-\x/4))});

\node(x1) at (7.5,-0.3) {\scalebox{1.2}{$\taul[1]$}};
\draw[domain=0:1.5,smooth,variable=\y,brown,line width=0.4mm] plot ({7.5},{\y});
\draw[domain=0:8.3,smooth,variable=\x,brown,line width=0.4mm] plot ({\x},{0.767});
\node(y1) at (-0.3,0.767) {\scalebox{1.2}{$1$}};

 
\node(xellm1) at (4.2,-0.3) {\scalebox{1.2}{$\taul[\ell-1]$}};
\draw[domain=0:2.5,smooth,variable=\y,red,line width=0.4mm] plot ({4},{\y});
\draw[domain=0:5.1,smooth,variable=\x,red,line width=0.4mm] plot ({\x},{1.84});
\node(yellm1) at (-0.5,1.84) {\scalebox{1.2}{$\ell-1$}};
 

\node(xell) at (3,-0.3) {\scalebox{1.2}{$\taul$}};
\draw[domain=0:3,smooth,variable=\y,green,line width=0.4mm] plot ({3},{\y});
\draw[domain=0:3.7,smooth,variable=\x,green,line width=0.4mm] plot ({\x},{2.36});
\node(yell) at (-0.3,2.36) {\scalebox{1.2}{$\ell$}};

\node(extra) at (9.9,0.3) {\scalebox{1.2}{$\dots$}};

\tikzfillbetween[of=dummyx3 and histo2]{pattern=my north east lines};
\tikzfillbetween[of=dummyx3 and histo3]{pattern=my north west lines};
\tikzfillbetween[of=dummyx2 and histo4]{pattern=dots};
\end{tikzpicture}
}
\caption{The expected histogram curve, $\{\E{N^{>\qt}(Y)}\}_{\qt\in\Rp}$} \label{histogram}
\end{figure}

\begin{theorem} \label{lem:gammaproxytopl} \label{eproxlthm}
For any $\ell\in[m]$, we have $\E{\topl(Y)}\leq\eproxl(Y)\leq 4\cdot\E{\topl(Y)}$, where
$\eproxl(Y):=\int_0^\infty\min\{\ell,\E{N^{>\qt}(Y)}\}d\qt$.
\end{theorem}

The proof of the second inequality above relies on a rephrasing of $\eproxl(Y)$
that makes it amenable to relate it to Lemma~\ref{lbproxlower}.
For any $\ell \in [m]$, define 
$\problth(Y) := \inf \bigl\{\qt\in \Rp : \E{N^{>\qt}(Y)} < \ell\bigr\}$. 
This infimum is attained because $\prob{Y_i\leq\qt}$ is a right-continuous
function of $\qt$,\footnote{%
If $\qt_1,\qt_2,\ldots$ is a decreasing sequence converging to $\qt$, then
$\{Y_i\leq\qt\}=\bigcap_{n=1}^{\infty}\{Y_i\leq\qt_n\}$, 
so $\prob{Y_i\leq\qt}=\lim_{n\to\infty}\prob{Y_i\leq\qt_n}$.}
and $\E{N^{>\qt}(Y)}=m-\sum_{i\in[m]}\prob{Y_i\leq\qt}$.
Also, define $\problth[0](Y) := \infty$ and $\problth(Y): = 0$ for $\ell > m$. 
Since $\E{N^{>\qt}(Y)}$ is a non-increasing function of $\qt$, we then have 
$\eproxl(Y)=\ell\problth(Y)+\int_{\problth(Y)}^\infty\E{N^{>\qt}}d\qt$.
We can further rewrite this in a way that quite nicely brings out the similarities between
$\eproxl(Y)$ and the (exact) proxy function for $\topl(x)$ used by~\cite{ChakrabartyS19a}
in the deterministic setting. Claim~\ref{enintl} gives a convenient way of casting the
integral in the second term, which leads to expression for $\eproxl(Y)$ stated in
Lemma~\ref{eproxlexpr}. 

\begin{claim} \label{enintl}
For any $t\in\Rp$, we have 
$\int_t^\infty\E{N^{>\qt}(Y)}d\qt=\sum_{i\in[m]}\E{(Y_i-t)^+}$.
\end{claim}

\begin{proof}
We have 
$$
\E{(Y_i-t)^+}=\int_0^\infty\Pr[(Y_i-t)^+>\qt]d\qt
=\int_0^\infty\Pr[Y_i>t+\qt]d\qt=\int_{t}^\infty\Pr[Y_i>\qt]d\qt.
$$
Therefore, $\sum_{i=1}^m\E{(Y_i-t)^+}
=\int_{t}^\infty\sum_{i=1}^m\Pr[Y_i>\qt]d\qt$, which is precisely 
$\int_{t}^\infty\E{N^{>\qt}(Y)}d\qt$.
\end{proof}  

\begin{lemma} \label{prop:gammatau} \label{eproxlexpr}
Consider any $\ell\in[m]$.
We have 
$\eproxl(Y) = \ell \problth(Y) + \sum_{i \in [m]} \E{\paren{Y_i-\problth(Y)}^+}$
and $\eproxl(Y)=\min_{t\in\Rp}\bigl(\ell t + \sum_{i \in [m]} \E{\paren{Y_i-t}^+}\bigr)$. 
\end{lemma}

\begin{proof}
Figure~\ref{histogram} yields a nice pictorial proof.
Algebraically, the first equality follows from the expression 
$\eproxl(Y)=\ell\problth(Y)+\int_{\problth(Y)}^\infty\E{N^{>\qt}}d\qt$ (see
Fig.~\ref{histogram}) and Claim~\ref{enintl}.
The second equality follows again from Claim~\ref{enintl} because we have 
$\eproxl(Y)\leq\ell t+\int_{t}^\infty\E{N^{>\qt}(Y)}d\qt$ for any $t\in\Rp$.
\end{proof}

\begin{remark}
The expression in Lemma~\ref{eproxlexpr} and $\problth(Y)$ nicely mirror similar
quantities in the deterministic setting. When $Y$ is deterministic: 
(i) $\problth(Y)$ is {\em precisely} $Y^{\down}_\ell$; 
the $\ell$-th largest coordinate of $Y$; 
and \linebreak
(ii) Lemma~\ref{eproxlexpr} translates to
$\eproxl(Y)=\min_{t\geq 0}\bigl(\ell t+\sum_{i\in[m]}(Y_i-t)^+\bigr)$;
the RHS is exactly $\topl(Y)$~\cite{ChakrabartyS19a} and is minimized at $Y^\down_\ell$. 
\end{remark}

\begin{proofof}{Theorem~\ref{eproxlthm}}
Let $\rho := \inf \{\qt : \sum_{i \in [m]} \E{Y_i^{\geq\qt}} \leq \ell\qt\}$. 
Note that $\sum_{i\in[m]}\E{Y_i^{>\rho}}\leq\ell\rho$.
By Lemma~\ref{prop:gammatau}, we have 
$\eproxl(Y) \leq \ell\rho + \sum_{i \in [m]} \E{(Y_i-\rho)^+} 
\leq \ell\rho + \sum_{i \in [m]} \E{Y_i^{>\rho}} \leq 2 \ell \rho$.
For any $\qt<\rho$,
Lemma~\ref{lem:lowerboundtopl} implies that $\E{\topl(Y)} > \ell\qt/2$. Thus, 
$\E{\topl(Y)}\geq \ell\rho/2$. 
\end{proofof}

The following result will be useful in Section~\ref{expnorm}.

\begin{lemma} \label{clm:gammadifflower} \label{eproxconsec}
For any $\ell \in \{2,\dots,m\}$, we have 
$\problth[\ell-1](Y)\geq\eproxl(Y)-\eproxl[\ell-1](Y) \geq \problth(Y)$.
\end{lemma}

\begin{proof}
Both inequalities are easily inferred from Fig.~\ref{histogram}.
Algebraically, we have 
$$
\eproxl(Y)-\eproxl[\ell-1](Y)
=\int_{0}^\infty\bigl(\min\{\ell,\E{N^{>\qt}(Y)}\}-\min\{\ell-1,\E{N^{>\qt}(Y)}\}\bigr)d\qt.
$$
For $\qt\geq\problth[\ell-1](Y)$, the integrand is $0$ since $\E{N^{>\qt}(Y)}<\ell-1$;
for $\qt<\problth(Y)$, the integrand is $1$ since $\E{N^{>\qt}(Y)}\geq\ell$; and
for $\problth(Y)\leq\qt<\problth[ell-1](Y)$, the integrand is at most $1$ since
$\ell-1\leq\E{N^{>\qt}(Y)}<\ell$.  
It follows that $\problth[\ell-1](Y)\leq\eproxl(Y)-\eproxl[\ell-1](Y)\geq\problth$.
\end{proof}

\section{Expectation of a monotone, symmetric norm} \label{sec:expectedfnorm} \label{expnorm} 

Let $Y$ follow a product distribution on $\Rp^m$, and $f:\R^m\mapsto\Rp$ be an arbitrary
monotone, symmetric norm. By Theorem~\ref{monnormthm} (a), there is a collection
$\C\sse\Rp^m$ of weight vectors with non-increasing coordinates such 
that $f(x)=\sup_{w\in\C} w^Tx^\down$ for all $x\in\Rp^m$. 
For $w\in\C$, recall that we define $w_\ell:=0$ whenever $\ell > m$. 
We now prove one of our main technical results: {\em the expectation   
of a supremum of ordered norms is within a constant factor of the supremum of the
expectation of ordered norms}. 
Formally, it is clear that 
$\E{f(Y)} = \E{\sup_{w\in \C} w^T Y^\down}
\geq\sup_{w \in \C} \E{w^T Y^\down}=\sup_{w \in \C} w^T\E{Y^\down}=f(\E{Y^\down})$; 
we show that an inequality in the opposite direction also holds.
Note that $Y^\down$ is the vector of order statistics of $Y$, from $\max$ to $\min$: i.e.,
$Y^\down_1$ is the maximum entry of $Y$, $Y^\down_2$ is the second-maximum entry, and so on. 

\begin{theorem} \label{thm:expectedfnorm} \label{thm:expnorm}
Let $Y$ follow a product distribution on $\Rp^m$, and $f:\R^m\to\Rp$ be a
monotone, symmetric norm. 
Then $f(\E{Y^\down})\leq \E{f(Y)}\leq \cexpnorm\cdot f(\E{Y^\down})$.
\end{theorem}

Theorem~\ref{thm:expnorm} is the backbone of 
our framework for stochastic minimum-norm optimization.
Since $f(\E{Y^\down})=\sup_{w\in\C}\sum_{\ell\in[m]}(w_\ell-w_{\ell+1})\E{\topl(Y)}$, 
Theorem~\ref{thm:expnorm} gives us a concrete way of bounding $\E{f(Y)}$, namely, by 
bounding all expected $\topl$ norms. In particular, we obtain the following 
{\em corollary}, that we call {\em approximate stochastic majorization}. Recall that 
$\POS_m=\{2^i:\, i\in\Zp, 2^i\leq m\}$.

\begin{theorem}[{\bf Approximate stochastic majorization}] \label{thm:apxstochmajor} \label{stochmaj}
Let $Y$ and $\vecrv$ follow product distributions on $\Rp^m$. Let $f$ be a monotone, symmetric
norm on $\R^m$.
\begin{enumerate}[(a), topsep=0.25ex, itemsep=0ex, leftmargin=*]
\item If $\E{\topl(Y)}\leq\al\cdot\E{\topl(\vecrv)}$ for all $\ell\in[m]$, then
$\E{f(Y)}\leq\cexpnorm\al\cdot\E{f(\vecrv)}$.

\item If $\E{\topl(Y)}\leq\al\cdot\E{\topl(\vecrv)}$ for all $\ell\in\POS_m$, then
$\E{f(Y)}\leq 2\cdot\cexpnorm\al\cdot\E{f(\vecrv)}$.
\end{enumerate}
\end{theorem}

\begin{proof}
Let $\C\sse\R^m$ be such that $f(x)=\sup_{w\in\C}w^Tx^\down$ for $x\in\Rp^m$.
Part (a) follows directly from Theorem~\ref{thm:expectedfnorm}, since for each random
variable $R\in\{Y,\vecrv\}$,
$f(\E{R^\down})$ depends only on the $\E{\topl(R)}$ quantities. We have 
\begin{equation*}
\begin{split}
\E{f(Y)} & \leq\cexpnorm\cdot f(\E{Y^\down})
=\cexpnorm\cdot\sup_{w\in\C}\sum_{\ell\in[m]}(w_\ell-w_{\ell+1})\E{\topl(Y)} \\
& \leq\cexpnorm\al\cdot\sup_{w\in\C}\sum_{\ell\in[m]}(w_\ell-w_{\ell+1})\E{\topl(\vecrv)}
= \cexpnorm\al\cdot f(\E{\vecrv^\down})\leq\cexpnorm\al\cdot\E{f(\vecrv)}.
\end{split}
\end{equation*}

Part (b) follows from part (a). For any $\ell\in[m]$, let 
$\ell'$ be the largest index in $\POS$ that is at most $\ell$. 
Then, $\ell'\leq\ell\leq 2\ell'$, and so 
So $\topl(Y)\leq 2\topl[\ell'](Y)\leq 2\al\topl[\ell'](\vecrv)\leq 2\al\topl(\vecrv)$. 
It follows that $\E{\topl(Y)}\leq 2\al\E{\topl(\vecrv)}$ for all $\ell\in[m]$.
\end{proof}

\begin{remark}
As should be evident from the proof above, the upper bounds on the ratio
$\E{f(Y)}/f\bigl(\E{Y^\down}\bigr)$ in Theorem~\ref{thm:expnorm}, and the 
ratio $\E{f(Y)}/\bigl(\al\cdot\E{f(\vecrv)}\bigr)$ in Theorem~\ref{stochmaj} are closely
related. 
Let $\kp_f$ be the supremum of $\E{f(Y)}/f\bigl(\E{Y^\down}\bigr)$ over random vectors $Y$
that follow a product distribution on $\Rp^m$.
Let $\zeta_f$ be the supremum of $\E{f(\vecrv^{(1)})}/\E{f(\vecrv^{(2)})}$ over random vectors
$\vecrv^{(1)},\vecrv^{(2)}$ that follow product distributions on $\Rp^m$ and satisfy
$\E{\topl(\vecrv^{(1)})}\leq\E{\topl(\vecrv^{(2)})}$ for all $\ell\in[m]$.
The above proof shows that $\zeta_f\leq\kp_f$. It is worth noting that $\kp_f\leq\zeta_f$
as well. To see this, take $\vecrv^{(1)}=Y$, and $\vecrv^{(2)}$ to be the 
{\em deterministic} vector $\E{Y^\down}$, which trivially follows a product
distribution. By definition, we have 
$\E{\topl(Y))}=\topl(\vecrv^{(2)})$ for all
$\ell\in[m]$. So we have that $\E{f(Y)}\leq\zeta_f\cdot f\bigl(\E{Y^\down}\bigr)$.  
\end{remark}

Theorem~\ref{stochmaj} (b) has immediate applicability 
for a given stochastic minimum-norm optimization problem such as stochastic min-norm load
balancing.  
In our algorithms, we enumerate (say in powers of $2$) all possible sequences
$\{\budg_\ell\}_{\ell\in\POS_m}$  
of estimates of $\bigl\{\E{\topl(\iopt)}\bigr\}_{\ell\in\POS_m}$, 
where $\iopt$ is the cost vector arising from an optimal solution,
and find a solution (if one exists) whose cost vector satisfies (roughly speaking) these
expected-$\topl$-norm estimates. 
A final step 
is to identify which of the solutions so obtained is a
near-optimal solution. While a probabilistic guarantee follows easily since one can
provide a randomized oracle for evaluating the objective $\E{f(Y)}$ 
(as $f(Y)$ enjoys good concentration properties), 
one can do better and {\em deterministically} identify a good solution.
To this end, we show below (Corollary~\ref{detestimate}) that if
$\{\budg_\ell\}_{\ell\in\POS_m}$ 
is a sequence 
that term-by-term well estimates $\bigl\{\E{\topl(Y)}\bigr\}_{\ell\in\POS_m}$, either
from below or or from above, 
then we can define a deterministic vector $\bvec\in\Rp^m$ such that $f(\bvec)$
well-estimates $\E{f(Y)}$, from below or above respectively.
We will apply parts (a) and (b) of Corollary~\ref{detestimate} with the cost
vectors arising from our solution and an optimal solution respectively, to argue the
near-optimality of our solution. 

Corollary~\ref{detestimate} utilizes a slightly technical result stated in
Lemma~\ref{bvecest}. 
We need the following definition.
Let $K:=2^{\floor{\log_2 m}}$ be the largest index in $\POS_m$.
Given a non-decreasing, nonnegative sequence $\{\budg_\ell\}_{\ell\in\POS_m}$, its
{\em upper envelope curve} $\bfun:[0,m]\to\Rp$ 
is defined by $\bfun(x):=\max\,\bigl\{y: (x,y)\in\conv(S)\bigr\}$, where 
$S=\{(\ell,\budg_\ell): \ell\in\POS_m\}\cup\{(0,0),(m,\budg_K)\}$ and $\conv(S)$ denotes
the convex hull of $S$. 

\begin{lemma} \label{bvecest}
Let $f$ be a monotone, symmetric norm on $\Rp^m$, and $y\in\Rp^m$ be a non-increasing
vector. 
Let $\{\budg_\ell\}_{\ell\in\POS_m}$ be a non-decreasing, nonnegative sequence such that
$\budg_{\ell}\leq 2\budg_{\ell/2}$ for all $\ell\in\POS_m, \ell>1$.  
Let $\bfun:[0,m]\mapsto\Rp$ be the {upper envelope curve} of 
$\{\budg_\ell\}_{\ell\in\POS_m}$.
Define $\bvec:=\bigl(\bfun(i)-\bfun(i-1)\bigr)_{i\in[m]}$.
\begin{enumerate}[(a), topsep=0.25ex, itemsep=0ex, leftmargin=*]
\item 
If $\Topl{y} \leq \budg_\ell$ for all $\ell \in \POS_m$, then $f(y) \leq 2 f(\bvec)$.

\item 
If $\Topl{y} \geq \budg_\ell$ for all $\ell \in \POS_m$, then $f(y) \geq f(\bvec)/3$.
\end{enumerate}
\end{lemma}

\begin{corollary} \label{detestimate}
Let $Y$ follow a product distribution on $\Rp^m$, and $f$ be a 
monotone, symmetric norm on $\R^m$.
Let $\{\budg_\ell\}_{\ell\in\POS_m}$ be a non-decreasing sequence 
such that $\budg_{\ell}\leq 2\budg_{\ell/2}$ for all $\ell\in\POS_m, \ell>1$. 
Let $\bvec\in\Rp^m$ be the vector given by Lemma~\ref{bvecest} for
the sequence $\{\budg_\ell\}_{\ell\in\POS_m}$. 
\begin{enumerate}[(a), topsep=0.75ex, itemsep=0ex, leftmargin=*]
\item If $\E{\topl(Y)}\leq\al\budg_\ell$ for all $\ell\in\POS_m$ (where $\al>0$), then  
$\E{f(Y)}\leq 2\cdot\cexpnorm\al\cdot f(\bvec)$.

\item If $\budg_\ell\leq 2\cdot\E{\topl(Y)}$ for all $\ell\in\POS_m$, 
then $f(\bvec)\leq 6\cdot\E{f(Y)}$.
\end{enumerate}
\end{corollary}

To avoid delaying the proof of Theorem~\ref{thm:expnorm}, we 
defer the proofs of Lemma~\ref{bvecest} and Corollary~\ref{detestimate} to
Section~\ref{proofs-expnorm}. 
The proof of Theorem~\ref{thm:expnorm} will also show that $\E{f(Y)}$ can be
controlled by controlling the $\problth(Y)$ statistics. This is particularly useful in
settings where the $Y_i$s are ``atomic'' random variables, and we have direct access to
their distributions. We discuss this in Section~\ref{efboundtau}.

\subsection*{Proof of Theorem~\ref{thm:expnorm}}
We now delve into the proof of the second (and main) inequality of
Theorem~\ref{thm:expnorm}.  
Since $Y$ and $f$ are fixed throughout, we drop the dependence on these in most items of  
notation in the sequel. 
We may assume without loss of generality that $f$ is normalized (i.e.,
$f(1,0,\ldots,0)=1$) as scaling $f$ to normalize it scales both $\E{f(Y)}$ and $f(\E{Y^\down})$
by the same factor.
It will be easier to work with the proxy function $\eproxl(Y)$ for $\E{\topl(Y)}$
defined in Section~\ref{proxy}. 
Recall that 
$\eproxl=\ell\problth+\int_{\problth}^\infty\E{N^{>\qt}(Y)}d\qt$, 
where $\problth$ is the smallest $\qt$ such that $\E{N^{>\qt}(Y)}<\ell$. 
Define $\problth[0]:=\infty$ and $\eproxl[0]:=0$ for notational convenience.
Define 
$\LB':=\sup_{w\in\C}\sum_{\ell\in[m]}(w_\ell-w_{\ell+1})\eproxl$. Given
Theorem~\ref{eproxlthm}, it suffices to show that $\E{f(Y)}\leq\clbp\cdot\LB'$.

The intuition and the roadmap of the proof are as follows. We have 
$f(Y)=\sup_{w\in\C}\sum_{\ell\in[m]}(w_\ell-w_{\ell+1})\topl(Y)$. Plugging in
$\topl(Y)=\int_0^\infty\min\bigl\{\ell,N^{>\qt}(Y)\bigr\}d\qt\leq\ell\problth+\int_{\problth}^\infty
N^{>\qt}(Y)d\qt$, 
we obtain that
\begin{equation}
f(Y)\leq\sup_{w\in\C}\biggl[\sum_{\ell\in[m]}(w_\ell-w_{\ell+1})\Bigl(\ell\problth
+\int_{\problth}^\infty N^{>\qt}(Y)d\qt\Bigr)\biggr].
\label{fyexpr}
\end{equation}
Comparing \eqref{fyexpr} and $\LB'$ (after expanding out the $\eproxl$ terms),
syntactically, the only difference is that the $N^{>\qt}(Y)$ terms appearing in
\eqref{fyexpr}  
are replaced by their expectations in $\LB'$; however the dependence of $f(Y)$ on the
$N^{>\qt}(Y)$ terms is quite non-linear, due to the $\sup$ operator.  
The chief insight that the $N^{>\qt}(Y)$ terms that really matter are those for 
$\qt=\problth$ for $\ell\in[m]$, and that the $N^{>\problth}(Y)$ quantities are 
{\em tightly concentrated} around their expectations. 
This allows us to, in essence, replace the $N^{>\qt}(Y)$ terms for
$\qt=\problth,\,\ell\in[m]$ with their expectations (roughly speaking) when we consider
$\E{f(Y)}$, incurring a constant-factor loss, and thereby argue that $\E{f(Y)}=O(\LB')$.

To elaborate, since $N^{>\qt}(Y)$ is non-increasing in $\qt$, we can upper bound 
each $\int_{\problth}^\infty N^{>\qt}(Y)d\qt$ expression in terms of $N^{>\problth[i]}(Y)$,
for $i=2,\ldots,\ell$, and $\int_{\problth[1]}^\infty N^{>\qt}(Y)d\qt$.
Consequently, the RHS of \eqref{fyexpr} can be upper bounded in
terms of the $N^{>\problth}(Y)$ quantities for $\ell=2,\ldots,m$, and
a term that depends on $\int_{\problth[1]}^\infty N^{>\qt}(Y)d\qt=\sum_{i\in[m]}(Y_i-\problth[1])^+$.
It is easy to charge the expectation of the latter term directly to $\LB'$.
We use $f_{-1}(Y)$ to denote (an upper bound on) the contribution from the
$N^{>\problth}(Y)$ quantities, $\ell=2,\ldots,m$ (see Lemma~\ref{efsplit}).
We bound $\E{f_{-1}(Y)}$ (and hence $\E{f(Y)}$) by $O(\LB')$ by proving that the tail
probability $\Pr[f_{-1}(Y)>\thresh\cdot\LB']$ decays exponentially with $\thresh$.
The {\em crucial observation} here is that 
$\Pr[f_{-1}(Y)>\thresh\cdot\LB']$ is at most
the probability that $N^{>\problth}(Y)>\Omega(\thresh)\cdot\ell$ for some index 
$\ell=2,\ldots,m$ (see Lemma~\ref{nbound} (a)). Each $N^{>\problth(Y)}$ is tightly
concentrated around its expectation, which is at most $\ell$, and so this probability 
is $O\bigl(e^{-\Omega(\thresh)}\bigr)$ (see Lemma~\ref{nconc} (b)).

We begin by proving convenient lower and upper bounds on $\LB'$ that will also 
relate $\LB'$ to the $f$-norm of the deterministic vector 
$\vec{\problth[]}:=(\problth)_{\ell\in[m]}$.
Define $\vec{\problth[]}_{-1}:=(\problth)_{\ell\in\{2,\ldots,m\}}$.
Recall that $f$ is normalized, and so $\sup_{w\in\C}w_1=1$.

\begin{lemma} \label{lbpbounds}
\begin{enumerate*}[(a)]
\item $\LB'\geq f\bigl(\sum_{i\in[m]}\E{(Y_i-\problth[1])^+}+\problth[1],\vec{\problth[]}_{-1}\bigr)$; \quad
\item $\LB'\leq f\bigl(\sum_{i\in[m]}\E{(Y_i-\problth[1])^+}+2\problth[1],\vec{\problth[]}_{-1}\bigr)$.
\end{enumerate*}
\end{lemma}

\begin{proof}
We have
$\LB'=\sup_{w\in\C}\sum_{\ell\in[m]}(w_\ell-w_{\ell+1})\eproxl=\sum_{\ell\in[m]}w_\ell(\eproxl-\eproxl[\ell-1])$.
Let $\vec{v}$ be the vector $(\eproxl-\eproxl[\ell-1])_{\ell\in[m]}$. (Recall that
$\eproxl[0]:=0$, so $v_1=\eproxl[1]$.)
Lemma~\ref{eproxconsec} combined with the fact that $\eproxl[1]\geq\problth[1]$, shows
that $\vec{v}$ has non-increasing coordinates, and so we have $\LB'=f(\vec{v})$.
The bounds on $\LB'$ now easily follow from the bounds on $\eproxl-\eproxl[\ell-1]$
in Lemma~\ref{eproxconsec}. 

Part (a) follows immediately from the lower bound on $\eproxl-\eproxl[\ell-1]$ in
Lemma~\ref{eproxconsec}, since it shows that $\vec{v}$ is coordinate-wise larger than the
vector  
$(\eproxl[1],\problth[2],\ldots,\problth[m])=
\bigl(\sum_{i\in[m]}\E{(Y_i-\problth[1])^+}+\problth[1],\vec{\problth[]}_{-1}\bigr)$.

For part (b), the upper bound on $\eproxl-\eproxl[\ell-1]$ from Lemma~\ref{eproxconsec}
shows that $\vec{v}$ is coordinate-wise smaller than the vector
$\vec{u}:=(\eproxl[1],\problth[1],\problth[2],\ldots,\problth[m-1])$. Therefore, 
$\LB'=f(\vec{v})\leq f(\vec{u})$.
Let $\vec{r}$ denote the vector $(\eproxl[1]+\problth[1],\vec{\problth[]}_{-1})$.
Finally, due to Theorem~\ref{monnormthm} (b), note that $f(\vec{u})\leq f(\vec{r})$ 
since it is easy to see that $\topl(\vec{u})\leq\topl(\vec{r})$ for all $\ell\in[m]$.
\end{proof}

For a given $w\in\C$, the term $\sum_{\ell\in[m]}(w_\ell-w_{\ell+1})\ell\problth$ will
figure frequently in our expressions, so to prevent cumbersome notation, we denote this by
$\tone(w)$ in the sequel. 
Define the quantity \linebreak
$f_{-1}(Y):=\sup_{w\in\C}\bigl(\tone(w)+\sum_{\ell=2}^mw_\ell(\problth[\ell-1]-\problth)N^{>\problth}(Y)\bigr)$. 

\begin{lemma} \label{efsplit}
We have $f(Y)\leq\sum_{i\in[m]}(Y_i-\problth[1])^++f_{-1}(Y)$.
\end{lemma}

\begin{proof}
We have $f(Y)=\sup_{w\in\C}\sum_{\ell\in[m]}(w_\ell-w_{\ell+1})\topl(Y)$.
Recall that $\problth[0]:=\infty$.
Plugging in the inequality 
$\topl(Y)\leq\ell\problth+\int_{\problth}^\infty N^{>\qt}(Y)d\qt$ for all $\ell\in[m]$, we
obtain that
\begin{equation*}
\begin{split}
f(Y) & \leq\sup_{w\in\C}\Bigl(\tone(w)+
\sum_{\ell\in[m]}(w_\ell-w_{\ell+1})\int_{\problth}^\infty N^{>\qt}(Y)d\qt\Bigr) \\ 
& =\sup_{w\in\C}\Bigl(\tone(w)+\sum_{\ell\in[m]}w_\ell
\Bigl(\int_{\problth}^\infty N^{>\qt}(Y)d\qt-\int_{\problth[\ell-1]}^\infty N^{>\qt}(Y)d\qt\Bigr)\Bigr) \\
& =\sup_{w\in\C}\Bigl(\tone(w)+w_1\int_{\problth[1]}^\infty N^{>\qt}(Y)d\qt+
\sum_{\ell=2}^m w_\ell\int_{\problth}^{\problth[\ell-1]} N^{>\qt}(Y)d\qt\Bigr) \\
& \leq \Bigl(\sup_{w\in\C}w_1\Bigr)\sum_{i\in[m]}(Y_i-\problth[1])^++
\sup_{w\in\C}\Bigl(\tone(w)+\sum_{\ell=2}^m w_\ell(\problth[\ell-1]-\problth)N^{>\problth}(Y)\Bigr) \\
& =\sum_{i\in[m]}(Y_i-\problth[1])^++f_{-1}(Y).
\end{split}
\end{equation*}
The final inequality follows since 
$\int_{t}^\infty N^{>\qt}(Y)d\qt=\sum_{i\in[m]}(Y_i-t)^+$ for any $t\geq 0$, and since
$N^{>\qt}(Y)$ is a non-increasing function of $\qt$.
\end{proof}

\begin{lemma} \label{nbound} \label{nconc}
Let $\thresh\geq 3$.
(a) $\Pr[f_{-1}(Y)>\thresh\cdot\LB']$ is at most 
$\Pr[\exists\ell\in\{2,\ldots,m\}\ \text{s.t.}\ N^{>\problth}(Y)>\thresh(\ell-1)]$; and
(b) the latter probability is at most $2.55\cdot e^{-\thresh/3}$.
\end{lemma}

\begin{proof}
For part (a), we show that if $N^{>\problth}(Y)\leq\thresh(\ell-1)$ for all
$\ell=2,\ldots,m$, then $f_{-1}(Y)\leq\thresh f(\vec{\problth[]})\leq\thresh\cdot\LB'$.
The upper bounds on $N^{>\problth}(Y)$ (and since $\tone(w)\geq 0$ for all $w\in\C$) imply that 
$f_{-1}(Y)$ is at most $\thresh\cdot
\sup_{w\in\C}\bigl(\tone(w)+\sum_{\ell=2}^m w_\ell(\problth[\ell-1]-\problth)(\ell-1)\bigr)$.
Expanding $\tone(w)=\sum_{\ell\in[m]}(w_\ell-w_{\ell+1})\ell\problth$, this bound simplifies to 
$\thresh\cdot\sup_{w\in\C}\sum_{\ell\in[m]}w_\ell\problth=\thresh\cdot f(\vec{\problth[]})$.

For part (b), we note that $\E{N^{>\problth}(Y)}<\ell$ by the choice of $\problth$, and
$N^{>\problth}(Y)=\sum_{i\in[m]}\bon_{Y_i>\problth}$, where the $\bon_{Y_i>\problth}$ are
{\em independent} random variables. Noting that $\ell-1\geq\ell/2$ (as $\ell\geq 2$),
by Chernoff bounds (Lemma~\ref{chernoff}), 
we have $\Pr[N^{>\problth}(Y)>\thresh(\ell-1)]\leq e^{-\thresh\ell/6}$.
So 
\[
\Pr[\exists\ell\in\{2,\ldots,m\}\ \text{s.t.}\ N^{>\problth}(Y)>\thresh(\ell-1)]
\leq\sum_{\ell=2}^me^{-\thresh\ell/6}
\leq\frac{e^{-\thresh/3}}{1-e^{-\thresh/6}}\leq 2.55\cdot e^{-\thresh/3}.
\]
The final inequality follows because $\thresh\geq 3$.
\end{proof}

\begin{proof}[{\bf Finishing up the proof of Theorem~\ref{thm:expnorm}}]
We have 
\begin{equation*}
\begin{split}
\E{f_{-1}(Y)}=\int_0^\infty\Pr[f_{-1}(Y)>\qt]d\qt
& \leq 3\cdot\LB'+\int_{3\LB'}^\infty\Pr[f_{-1}(Y)>\qt]d\qt \\
& \leq 3\cdot\LB'+\LB'\cdot\int_3^\infty\prob{f_{-1}(Y)>\thresh\cdot\LB'}d\thresh
\end{split}
\end{equation*}
where the final inequality follows by the variable-change $\thresh=\qt/\LB'$.
Using Lemma~\ref{nconc}, the factor multiplying $\LB'$ in the second term is at most
$\int_3^\infty 2.55e^{-\thresh/3}d\thresh=3\cdot 2.55\cdot e^{-1}\leq 3$.
Combining this with Lemmas~\ref{efsplit} and~\ref{lbpbounds}, we obtain that
$\E{f(Y)}\leq\clbp\cdot\LB'$. Hence, by Theorem~\ref{eproxlthm}, 
$\E{f(Y)}\leq\cexpnorm\cdot f(\E{Y^\down})$.
\end{proof}

\subsection{Proofs of Lemma~\ref{bvecest} and Corollary~\ref{detestimate}} 
\label{proofs-expnorm}

\begin{duplicate}[Lemma~\ref{bvecest} (restated)]
Let $f$ be a monotone, symmetric norm on $\Rp^m$, and $y\in\Rp^m$ be a non-increasing
vector. 
Let $\{\budg_\ell\}_{\ell\in\POS_m}$ be a non-decreasing, nonnegative sequence such that
$\budg_{\ell}\leq 2\budg_{\ell/2}$ for all $\ell\in\POS_m, \ell>1$.  
Let $\bfun:[0,m]\mapsto\Rp$ be the {upper envelope curve} of the sequence 
$\{\budg_\ell\}_{\ell\in\POS_m}$.
Define $\bvec:=\bigl(\bfun(i)-\bfun(i-1)\bigr)_{i\in[m]}$.
\begin{enumerate}[(a), topsep=0.25ex, itemsep=0ex, leftmargin=*]
\item 
If $\Topl{y} \leq \budg_\ell$ for all $\ell \in \POS_m$, then $f(y) \leq 2 f(\bvec)$.

\item 
If $\Topl{y} \geq \budg_\ell$ for all $\ell \in \POS_m$, then $f(y) \geq f(\bvec)/3$.
\end{enumerate}
\end{duplicate}

\begin{proof} 
Recall that the upper envelope curve $\bfun:[0,m]\to\Rp$ 
of $\{\budg_\ell\}_{\ell\in\POS_m}$ 
is given by $\bfun(x):=\max\,\bigl\{y: (x,y)\in\conv(S)\bigr\}$, where 
$S=\{(\ell,\budg_\ell): \ell\in\POS_m\}\cup\{(0,0),(m,\budg_K)\}$ and $\conv(S)$ denotes
the convex hull of $S$. 
Define $\budg_0:=0$, and $\budg_\ell=\budg_{K}$ for all $\ell>K$, for notational
convenience. 

We may assume that $\budg_1>0$, as otherwise $\bfun(x)=0$ for all $x\in[0,m+1]$; so
$\bvec=\vec{0}$, and in part (a), we have $y=\vec{0}$, so both parts hold trivially.
By Theorem~\ref{monnormthm} (b), it suffices to show that $\topl[i](y)\leq 2\topl[i](\bvec)$
for all $i\in[m]$ for part (a), and to show that $\topl[i](y)\geq\topl[i](\bvec)/3$ for all
$i\in[m]$ for part (b). 

It is easy to see that $\bfun$ is a {\em non-decreasing, concave function}: consider 
$0\leq x<x'\leq m$. By Carath\'eodory's theorem, there at most two points
$(\ell_1,\budg_{\ell_1})$, $(\ell_2,\budg_{\ell_2})$ in $S$ such that $(x,\bfun(x))$ lies
in the convex hull of these two points, which is a line segment $L$. Without loss of
generality, suppose $\ell_1\leq x\leq \ell_2$, and so
$\budg_{\ell_1}\leq\bfun(x)\leq\budg_{\ell_2}$. If $x'\leq\ell_2$, then the point 
$(x',\cdot)$ on $L$ has $y$-coordinate at least $\bfun(x)$. If $x'>\ell_2$, then
$\bfun(x')$ is at least the $y$-coordinate of the point $(x',\cdot)$ lying on the line
segment joining $(\ell_2,\budg_{\ell_2})$ and $(m,\budg_K)$; this value is at least
$\budg_{\ell_2}\geq\bfun(x)$.

Concavity follows simply from the fact that for any $x_1,x_2\in[0,m+1]$ and $\ld\in[0,1]$, 
since the points $(x_1,\bfun(x_1))$ and $(x_2,\bfun(x_2))$ lie in $\conv(S)$, the point
$\bigl(\ld x_1+(1-\ld)x_2,\ld\bfun(x_1)+(1-\ld)\bfun(x_2)\bigr)$ is also in
$\conv(S)$. Hence, $\bfun(\ld x_1+(1-\ld)x_2)\geq\ld\bfun(x_1)+(1-\ld)\bfun(x_2)$.

It follows that $\bvec$ is a non-increasing vector, and
$\topl[i](\bvec)=\sum_{j=1}^i\bvec_j=\bfun(i)$ for all $i\in[m]$.

Part (a) now follows easily. 
Consider any $i\in[m]$, and let $\ell$ be the largest index in $\POS_m$ that is at
most $i$. We have 
$$
\topl[i](\bvec)=\bfun(i)\geq\bfun(\ell)\geq\topl(y)\geq\topl[i](y)/2.
$$
The first inequality is due to the monotonicity of $\bfun$; the second follows from our
assumptions, and the last inequality is because $i\leq 2\ell$.

For part (b), consider any $i\in[m]$, and again let $\ell$ be the largest index in
$\POS_m$ that is at most $i$. If $(i,\bfun(i))$ is an {\em extreme point} of $\conv(S)$,
then it must be that $i\in\POS_m\cup\{m\}$, and $\bfun(i)=\budg_i$. Also, we have
$\budg_i\leq\budg_\ell\leq\topl(y)\leq\topl[i](y)$.
So suppose $(i,\bfun(i))$ is not an extreme point of $\conv(S)$.
Then there are two points $(\ell_1,\budg_{\ell_1})$, $(\ell_2,\budg_{\ell_2})$ in $S$, which
are extreme points of $\conv(S)$, such that $(i,\bfun(i))$ lies on the line segment $L$
joining these two points. Let $\ell_1\leq i\leq\ell_2$.
The slope of $L$ is $s:=\frac{\budg_{\ell_2}-\budg_{\ell_1}}{\ell_2-\ell_1}$.
Noting that $\bfun(x)=\budg_x$ for $x\in\{\ell_1,\ell_2\}$, we have
$\bfun(i)=\budg_{\ell_1}+(i-\ell_1)s$. 
We argue that $s\leq\frac{\budg_\ell}{\ell}$, which implies that 
$$
\bfun(i)\leq\budg_{\ell_1}+(i-\ell_1)\cdot\frac{\budg_\ell}{\ell}\leq 
3\budg_\ell\leq 3\topl(y)\leq 3\topl[i](y).
$$
Since $\budg_{\ell'}\leq 2\budg_{\ell'/2}$ for all $\ell'\in\POS_m, \ell'>1$, it
follows that $\bigl\{\frac{\budg_{\ell'}}{\ell'}\bigr\}_{\ell'\in\POS_m}$ is a
non-increasing sequence. It follows that
$\budg_{\ell_1}\geq\ell_1\cdot\frac{\budg_{\ell_2}}{\ell_2}$, and therefore 
$s\leq\frac{\budg_{\ell_2}}{\ell_2}\leq\frac{\budg_\ell}{\ell}$. 
\end{proof}

\begin{duplicate}[Corollary~\ref{detestimate} (restated)]
Let $Y$ follow a product distribution on $\Rp^m$, and $f$ be a 
monotone, symmetric norm on $\R^m$.
Let $\{\budg_\ell\}_{\ell\in\POS_m}$ be a non-decreasing sequence 
such that $\budg_{\ell}\leq 2\budg_{\ell/2}$ for all $\ell\in\POS_m, \ell>1$. 
Let $\bvec\in\Rp^m$ be the vector given by Lemma~\ref{bvecest} for
the sequence $\{\budg_\ell\}_{\ell\in\POS_m}$. 
\begin{enumerate}[(a), topsep=0.5ex, itemsep=0ex, leftmargin=*]
\item If $\E{\topl(Y)}\leq\al\budg_\ell$ for all $\ell\in\POS_m$ (where $\al>0$), then  
$\E{f(Y)}\leq 2\cdot\cexpnorm\al\cdot f(\bvec)$.

\item If $\budg_\ell\leq 2\cdot\E{\topl(Y)}$ for all $\ell\in\POS_m$, 
then $f(\bvec)\leq 6\cdot\E{f(Y)}$.
\end{enumerate}
\end{duplicate}

\begin{proof} 
The proof follows by combining Theorem~\ref{thm:expnorm} and Lemma~\ref{bvecest}.
Consider the vector $y=\E{Y^\down}$.
By Theorem~\ref{thm:expnorm}, we have $\E{f(Y)}\leq\cexpnorm f(y)$. 
For part (a), we apply Lemma~\ref{bvecest} (a) with the vector $y$ and the sequence
$\{\al\budg_\ell\}_{\ell\in\POS_m}$. Note that the vector given by Lemma~\ref{bvecest} for
$\{\al\budg_\ell\}_{\ell\in\POS_m}$ is exactly $\al\bvec$. Thus, we obtain that 
$f(y)\leq 2f(\al\bvec)=2\al\cdot f(\bvec)$.

For part (b), we have $\E{f(Y)}\geq f(y)$. 
Applying Lemma~\ref{bvecest} (b) with the vector $y$ and the sequence
$\{\budg_\ell/2\}_{\ell\in\POS_m}$, we obtain that $f(y)\geq f(\bvec/2)/3=f(\bvec)/6$.
\end{proof}

\subsection{Controlling $\E{f(Y)}$ using the $\problth(Y)$ statistics} \label{efboundtau}
It is implicit from the proof of Theorem~\ref{thm:expnorm} that the $\problth$ statistics
control $\E{f(Y)}$. We now make this explicit, so that we can utilize this in settings 
where 
we have direct access to the distributions of the $Y_i$ random variables. 
As before, 
we drop the dependence on $Y$ and $f$ in most items of notation. 
The proof of Theorem~\ref{thm:expnorm} coupled with Lemma~\ref{lbpbounds} and the fact
that $\LB'=\Theta\bigl(f(\E{Y^\down})\bigr)$, yields the following useful 
expression involving $\problth$s for estimating $\E{f(Y)}$.
As before, $\vec{\problth[]}(Y)$ is the vector $\bigl(\problth(Y)\bigr)_{\ell\in[m]}$.

\begin{lemma} \label{tauexpr}
Let $Y$ follow a product distribution on $\Rp^m$. Let $f$ be a normalized, 
monotone, symmetric norm on $\R^m$. 
We have
$\frac{1}{\ctauexpr}\cdot\E{f(Y)}
\leq\sum_{i\in[m]}\E{\bigl(Y_i-\problth[1](Y)\bigr)^+}+f\bigl(\vec{\problth[]}(Y)\bigr)
\leq 8\cdot\E{f(Y)}$.
\end{lemma}

\begin{proof} 
We abbreviate $\problth(Y)$ to $\problth$ for all $\ell\in[m]$, and $\vec{\problth[]}(Y)$
to $\vec{\problth[]}$.
We have $\E{f(Y)}\geq f(\E{Y^\down})\geq\LB'/4\geq\frac{1}{8}\cdot
\bigl(\sum_{i\in[m]}\E{(Y_i-\problth[1])^+}+f(\vec{\problth[]})\bigr)$. The second
inequality is due to Theorem~\ref{eproxlthm}, and the third follows from 
Lemma~\ref{lbpbounds} (a). 

We also have 
$\E{f(Y)}\leq\clbp\cdot\LB'\leq
2\cdot\clbp\bigl(\sum_{i\in[m]}\E{(Y_i-\problth[1])^+}+f(\vec{\problth[]})\bigr)$.
The first inequality is what we show in the proof of Theorem~\ref{thm:expnorm}; the second
inequality is follows from Lemma~\ref{lbpbounds} (b). 
\end{proof}

In utilizing Lemma~\ref{tauexpr} in an application, 
we work with estimates $\{\bproblth\}_{\ell\in\POS_m}$ of the $\problth^*$ values,
for $\ell\in\POS_m$, of the cost vector arising from an optimal solution,  
and seek a solution whose cost vector $Y$ minimizes
$\sum_{i\in[m]}\E{(Y_i-\bproblth[1])^+}$ subject to the constraint that (roughly speaking)
the $\problth(Y)$s are bounded by the corresponding $\bproblth$ values, for all
$\ell\in\POS_m$. 
Lemma~\ref{tauexpr} then indicates that if the $\bproblth$s are good estimates of the
$\problth^*$ values, then our solution is a near-optimal solution.
We state a more-robust such result below that incorporates various approximation factors, 
which will be particularly convenient to apply since
such approximation factors inevitably arise in the process of finding $Y$.
We need the following notation: given a non-increasing sequence
$\{v_\ell\}_{\ell\in\POS_m}$, we define its {\em expansion} to be the
vector $v'\in\R^m$ given by $v'_i:=v_{2^{\floor{\log_2 i}}}$ for all $i\in[m]$.

\begin{theorem} \label{apxstatsthm}
Let $Y$ follow a product distributions on $\Rp^m$. Let $f$ be a normalized, 
monotone, symmetric norm on $\R^m$. 
Let $\{\bproblth\}_{\ell\in\POS_m}$ be a nonnegative, non-increasing sequence, and
$\vtp$ be its expansion. 
\begin{enumerate}[(a), topsep=0.5ex, itemsep=0ex, leftmargin=*]
\item Suppose that $\problth[\beta\ell](Y)\leq\al\bproblth$ for all $\ell\in\POS_m$, where 
$\al,\beta\geq 1$. Then, we have that \linebreak
$\E{f(Y)}\leq 2\cdot\clbp\cdot(\al+2)\beta\cdot
\bigl(\sum_{i\in[m]}\E{(Y_i-\al\bproblth[1])^+}+f(\vtp)\bigr)$.

\item Suppose that $\problth(Y)\leq\bproblth\leq 2\problth(Y)+\kp$
for all $\ell\in\POS_m$, where $\kp\geq 0$. 
Then, we have that \linebreak
$\sum_{i\in[m]}\E{(Y_i-\bproblth[1])^+}+f(\vtp)\leq 32\cdot\E{f(Y)}+m\kp$.
\end{enumerate}
\end{theorem}

\begin{proof}
We abbreviate $\POS_m$ to $\POS$.
We abbreviate $\problth(Y)$ to $\problth$ for all $\ell\in[m]$.
Define $\problth[0]:=\infty$ 
for notational convenience.
Let $K:=2^{\floor{\log_2 m}}$ be the largest index in $\POS$.
Let $\vec{\problth[]}=(\problth[1],\problth[2],\ldots,\problth[m])$.

\vspace*{-1ex}
\paragraph{Part (a).}
We begin with the bound 
$\E{f(Y)}\leq 2\cdot\clbp\cdot\bigl(\sum_{i\in[m]}\E{(Y_i-\problth[1])^+}+f(\vec{\problth[]})\bigr)$
given by Lemma~\ref{tauexpr}:
We proceed to bound the RHS expression in terms of the quantities in the stated
upper bound.

By Lemma~\ref{eproxlexpr}, we have
\begin{equation}
\problth[1]+\sum_{i\in[m]}\E{(Y_i-\problth[1])^+}=\eproxl[1](Y)\leq
\al\bproblth[1]+\sum_{i\in[m]}\E{(Y_i-\al\bproblth[1])^+}.
\label{apxstatsineq1}
\end{equation}
To compare $f(\vec{\problth[]})$ and $f(\vtp)$, by Theorem~\ref{monnormthm}
(b), it suffices to compare the $\topl$-norms of $\vec{\problth[]}$ and
$\vtp$
Consider $i\in[m]$ with $i\geq\beta$. Note that $2^{\floor{\log_2(\floor{i/\beta})}}\in\POS$.
We have
$$
\problth[i]\leq\problth[\beta\cdot\floor{i/\beta}]
\leq\problth[\beta\cdot 2^{\floor{\log_2(\floor{i/\beta})}}]
\leq\al\bproblth[2^{\floor{\log_2(\floor{i/\beta})}}].
$$
Therefore, for any $\ell\in[m]$, we have that
$$
\topl(\vec{\problth[]})=\sum_{i=1}^\ell\problth[i]
\leq\beta\problth[1]+\sum_{i=\ceil{\beta}}^\ell\problth[i]
\leq\beta\problth[1]+\sum_{i=\ceil{\beta}}^\ell\al\bproblth[2^{\floor{\log_2(\floor{i/\beta})}}].
$$
The $\bproblth[j]$ terms that appear in the second summation are for indices $j\in\POS$,
$j\leq\ell/\beta$, and each such $\bproblth[j]$ term appears at most $\beta j$ times. 
Therefore, we have 
$\topl(\vec{\problth[]})\leq\beta\problth[1]+\beta\sum_{j\in\POS: j\leq\ell/\beta}j\bproblth[j]$.
Observe that in $\vtp$, each $t'_j$ term for $j\in\POS\sm\{K\}$ appears
$j$ times, 
and $\sum_{j\in\POS: j\leq\ell/\beta}j\leq 2\ell/\beta\leq 2\ell$. 
This implies that $\sum_{j\in\POS: j\leq\ell/\beta}j\bproblth[j]$ is at most
$2\cdot\topl(\vtp)$, and we obtain that  
$\topl(\vec{\problth[]})\leq\beta\problth[1]+2\beta\cdot\topl(\vtp)$.

By Theorem~\ref{monnormthm} (b) (and since $f$ is normalized), it follows that 
$f(\vec{\problth[]})\leq\beta\problth[1]+2\beta\cdot f(\vtp)$.
Adding this to $\beta$ times \eqref{apxstatsineq1} (recall that $\beta\geq 1$), noting
that $\bproblth[1]\leq f(\vtp)$, and
simplifying gives
$\bigl(\sum_{i\in[m]}\E{(Y_i-\problth[1])^+}+f(\vec{\problth[]})\bigr)
\leq(\al+2)\beta\bigl(\sum_{i\in[m]}\E{(Y_i-\al\bproblth[1])^+}+f(\vtp)\bigr)$.
Combining this with the upper bound on $\E{f(Y)}$ mentioned at the beginning proves part 
(a). 

\vspace*{-1ex}
\paragraph{Part (b).}
We now start with the bound 
$\sum_{i\in[m]}\E{(Y_i-\problth[1])^+}+f(\vec{\problth[]})\leq 8\cdot\E{f(Y)}$
given by Lemma~\ref{tauexpr}, and
proceed to upper bound $\sum_{i\in[m]}\E{(Y_i-\bproblth[1])^+}+f(\vtp)$ in
terms of the LHS of this inequality.

Since $\bproblth[1]\geq\problth[1]$, it is immediate that
$\sum_{i\in[m]}\E{(Y_i-\bproblth[1])^+}\leq\sum_{i\in[m]}\E{(Y_i-\problth[1])^+}$.
As should be routine by now, we compare $f(\vtp)$ and $f(\vec{\problth[]})$
by comparing the $\topl$-norms of $\vtp$ and $\vec{\problth[]}$.
First, note that for any $i\in[m]$, we have
$$ 
t'_i=\bproblth[2^{\floor{\log_2 i}}]
\leq 2\problth[2^{\floor{\log_2 i}}]+\kp\leq 2\problth[\ceil{i/2}]+\kp
$$
where the last inequality follows since $2^{\floor{\log_2 i}}>i/2$.
Now consider any $\ell\in[m]$. We have 
\begin{equation*}
\begin{split}
\topl(\vtp) & = \sum_{i=1}^\ell t'_i
\leq\sum_{i=1}^\ell\bigl(2\problth[\ceil{i/2}]+\kp\bigr) \\
& \leq 4\sum_{j=1}^{\ceil{\ell/2}}\problth[j]+m\kp\leq 4\topl(\vec{\problth[]})+m\kp.
\end{split}
\end{equation*}

It follows that $f(\vtp)\leq 4 \cdot f(\vec{\problth[]})+m\kp$, and hence we have
\begin{equation*}
\sum_{i\in[m]}\E{(Y_i-\bproblth[1])^+}+f(\vtp)
\leq 4\bigl(\sum_{i\in[m]}\E{(Y_i-\problth[1])^+}+f(\vec{\problth[]})\bigr)+m\kp
\leq 32\cdot\E{f(Y)}+m\kp. \qedhere
\end{equation*}
\end{proof}

\section{Load Balancing} \label{sec:loadbal}
We now apply our framework to devise approximation algorithms for  
{\em stochastic minimum norm load balancing on unrelated machines}.  
We are given $n$ \emph{stochastic} jobs that need to be assigned to $m$ unrelated
machines. Throughout, we use $[n]$ and $[m]$ to denote the set of jobs and machines
respectively; we use $j$ to index jobs, and $i$ to index machines.
For each $i \in [m],j \in [n]$, we have a nonnegative r.v. $X_{ij}$ that denotes the
processing time of job $j$ on machine $i$, whose distribution is specified in the input.  
Jobs are independent, so $X_{ij}$ and $X_{i'j'}$ are independent whenever $j \neq j'$; 
however, $X_{ij}$ and $X_{i'j}$ could be correlated.
A feasible solution is an assignment $\sg : [n] \to [m]$ of jobs to machines. 
This induces a random load vector $\load_\sg$ where $\load_\sg(i) := \sum_{j : \sg(j) = i}
X_{ij}$ for each $i \in [m]$; note that $\load_\sg$ follows a product distribution on
$\Rp^m$. 
The goal is to find an assignment $\sg$ that minimizes $\E{f(\load_\sg)}$ for a given
monotone symmetric norm $f$. 

We often use $j \mapsto i$ as a shorthand for denoting $\sg(j) = i$, when $\sg$ is clear
from the context. 
We use $\ssg$ to denote an optimal solution, and 
$\OPT :=\E{f(\load_{\ssg})}$ to denote the optimal value. 
Let $\POS=\POS_m := \{ 1,2,4, \dots, 2^{\floor{\log_2 m}}\}$.  

\paragraph{Overview.}
Recall that Theorem~\ref{stochmaj} underlying our framework 
shows that in order to obtain
an $\bigo{\alpha}$-approximation for stochastic $f$-norm load balancing, it suffices to
find an assignment $\sg$ that satisfies \linebreak $\E{\Topl{\load_\sg}} \leq \alpha \,
\E{\topl(\load_{\ssg})}$ for all $\ell \in \POS$. 
First, in Section~\ref{sec:loadbaltopl}, we consider the simpler problem where we only
have one expected-$\topl$ budget constraint, or equivalently, where $f$ 
is a $\topl$-norm, and obtain an $O(1)$-approximation algorithm in this case. 

\begin{theorem} \label{thm:loadbaltopl}
There is a constant-factor approximation algorithm for stochastic $\topl$-norm load balancing on unrelated machines with arbitrary job size distributions.
\end{theorem}

Section~\ref{sec:loadbaltopl} will introduce many of the techniques that we build upon and
refine in 
Section~\ref{sec:loadbalf}, where we deal with an arbitrary monotone, symmetric norm.  
In Section~\ref{sec:loadbalber}, we devise a constant-factor approximation when job sizes
are {\em Bernoulli random variables}. Observe that since a deterministic job size can also
be viewed as a trivial Bernoulli random variable, modulo constant factors, this result
strictly generalizes the $O(1)$-approximation for deterministic min-norm load
balancing obtained by~\cite{ChakrabartyS19a,ChakrabartyS19b}.
In Section~\ref{sec:loadbalgen}, we consider the most general setting, (i.e., arbitrary
norm and arbitrary distributions), and give an $O(\log m/\log\log m)$-approximation. 

\begin{theorem} \label{thm:loadbalber}
There is a constant-factor approximation algorithm for stochastic minimum-norm load balancing
on unrelated machines with an arbitrary monotone, symmetric norm and Bernoulli job-size
\linebreak distributions.
\end{theorem}

\begin{theorem} \label{thm:loadbalgen}
There is an $\bigo{\log m/\log \log m}$-approximation algorithm for the general stochastic
minimum-norm load balancing problem on unrelated machines, where the underlying monotone,
symmetric norm and job-size distributions are arbitrary.
\end{theorem}

\subsection{Stochastic \boldmath{$\topl$}-norm load balancing} \label{sec:loadbaltopl}
In this section, we prove Theorem~\ref{thm:loadbaltopl} and devise an $O(1)$-approximation
for stochastic $\topl$-norm load balancing.
The key to our approach are Lemmas~\ref{lbproxupper} and~\ref{lbproxlower}, which
together imply that, for $t$ chosen suitably, $\sum_i\E{{\load_{\sg}(i)}^{\geq t}}$ acts as
a convenient proxy for $\E{\topl(\load_\sg)}$: in particular, we have
\linebreak $\E{\topl(\load_\sg)}=O(\ell t)$ if and only if 
$\sum_i\E{\load_{\sg}(i)^{\geq t}}=O(\ell t)$. 
Lemma~\ref{lbproxlower} shows
that if $t\geq\frac{2 \, \OPT}{\ell}$, then $\sum_i \E{{\load_{\sg*}(i)}^{\geq t}}\leq\ell t$. 
We write a linear program, $\LPv[\ell,t]$, to find such an assignment (roughly speaking), and round its solution to obtain an assignment $\sg$ such that $\sum_i\E{{\load_{\sg}(i)}^{\geq t}} = O(\ell t)$; by
Lemma~\ref{lbproxupper}, this implies that $\E{\topl(\load_{\sg})}=O(\ell t)$. 
Hence, if we work with $t=O\bigl(\frac{\OPT}{\ell}\bigr)$ such that $\LPv[\ell,t]$ is
feasible---which we can find via binary search---then we obtain an $O(1)$-approximation.

\paragraph{LP relaxation.}
Let $t\geq 0$ be a given parameter.
Our LP seeks a fractional assignment satisfying \linebreak
$\sum_i\E{\load(i)^{\geq t}}=O(\ell t)$. 
As usual, we have $z_{ij}$ variables indicating if job $j$
is assigned to machine $i$, so $z$ belongs to the assignment polytope 
$\Q^\asgn:= \{ z\in \Rp^{m \times n}:\ \sum_{i \in [m]} z_{ij} = 1\quad \forall i \in [m]\}$.

As alluded to in Section~\ref{largevalbnd}, $\E{\erv{\load(i)}{t}}$ can be controlled by separately handling the contribution from exceptional jobs $\erv{X_{ij}}{t}$ and truncated jobs $\trv{X_{ij}}{t}$. 
Our LP enforces that both these contributions (across all machines) are at most $\ell t$, thereby ensuring that $\sum_i\E{\load(i)^{\geq t}}=O(\ell t)$ (due to \eqref{expgeqbnd}). 
Constraint \eqref{eq:excepineqtopl} directly encodes that the total contribution from
exceptional jobs is at most $\ell t$.
Handling the contribution from truncated jobs is more complicated.
We utilize Lemma~\ref{largebeta} here, which uses the notion of effective sizes.
For each machine $i$, let $L_i := \sum_\jtoi \trv{X_{ij}}{t}/t$ denote
the scaled load on machine $i$ due to the truncated jobs assigned to it. 
We use an auxiliary variable $\eload$ to model $\E{\erv{L_i}{1}}$, so that $t \eload$ models
$\E{(\sum_\jtoi X_{ij}^{<t})^{\geq t}}$. 
Since $L_i$ is a sum of independent $[0,1]$-random variables, Lemma~\ref{largebeta} yields various
lower bounds on $\E{L_i^{\geq 1}}$; these are incorporated by constraints
\eqref{eq:betaineqtopl}. 
A priori, we do not know which is the \emph{right} choice of $\ld$ in Lemma~\ref{largebeta}, 
so we simply include constraints for a sufficiently large collection of $\ld$ values so that one of them is close enough to the right choice.
Finally, constraint~\eqref{eq:truncineqtopl} ensures that 
$\sum_i\E{(\sum_\jtoi X_{ij}^{<t})^{\geq t}}\leq\ell t$, thereby bounding the
contribution from the truncated jobs. 
We obtain the following feasibility program.

\begin{minipage}{0.2\textwidth}
\leqnomode
\begin{equation}
\tag*{\LPv[\ell,t]} \label{topllp}
\end{equation}
\reqnomode
\end{minipage}
\ 
\begin{minipage}{0.75\textwidth}
\begin{alignat}{1}
\sum_{i \in [m],j \in [n]} \E{\erv{X_{ij}}{t}} z_{ij} & \leq \ell t \label{eq:excepineqtopl} \\
\frac{\sum_{j \in [n]} \eff{\trv{X_{ij}}{t}/4t} z_{ij} - 6}{4\ld} & \leq \eload \qquad 
\forall i \in [m], \forall \ld\in\{1,\ldots,100m\} \label{eq:betaineqtopl} \\
\sum_{i \in [m]} \eload & \leq \ell  \label{eq:truncineqtopl} \\
\eload[] & \geq 0, \gaptwo z \in \Q^\asgn.  
\end{alignat}
\end{minipage}

\begin{claim} \label{clm:LBfeastopl}
$\LPv[\ell,t]$ is feasible for any $t$ satisfying $\E{\Topl{\load_{\ssg}}}\leq\ell t/2$.
\end{claim}

\begin{proof}
Consider the solution $(z^*,\eload[]^*)$ induced by the optimal solution $\ssg$, where 
$z^*$ is the indicator vector of $\ssg$, and
$\eload^*:= \E{\erv{L_i}{1}}$ for any $i \in [m]$, with
$L_i := \sum_{j: \ssg(j)=i}\trv{X_{ij}}{t}/t$.
The choice of $t$ implies that by Lemma~\ref{lbproxlower}, we have
$\sum_i\E{\load_{\ssg}(i)^{\geq t}}\leq\ell t$.
Notice that 
$\load_{\ssg}(i)^{\geq t}\geq\sum_\jtoi X_{ij}^{\geq t}+\bigl(\sum_\jtoi X_{ij}^{<t}\bigr)^{\geq t}$.
Therefore, $\ell t\geq \sum_i \E{ \erv{\load(i)}{t} } \geq \sum_i \sum_\jtoi \E{\erv{X_{ij}}{t}}$, 
showing that \eqref{eq:excepineqtopl} holds.
Next Lemma~\ref{lem:largebetalargetail} applied to the composite r.v. $L_i$, which is a
sum of independent $[0,1]$-bounded r.v.s, 
shows that constraint \eqref{eq:betaineqtopl} holds for any integral $\ld \geq 1$.
Finally, $t\eload^* = \E{\bigl(\sum_\jtoi X_{ij}^{<t}\bigr)^{\geq t}}$ for each $i$, and the upper
bound of $\ell t$ on $\sum_i\E{\load_{\ssg}(i)^{\geq t}}$ implies that $\sum_i \eload^* \leq \ell$.
\end{proof}

\paragraph{Rounding algorithm.}
We assume now that 
$\LPv[\ell,t]$ is feasible.
Let $(\bz,\beload[])$ be a feasible fractional solution to this LP. 
For each machine $i$, we carefully choose a suitable budget constraint from
among the constraints \eqref{eq:betaineqtopl}, and round the fractional assignment $\bz$ to
obtain an assignment $\sg$ such that:
(i) these budget constraints for the machines are satisfied approximately; and 
(ii) the total contribution from the exceptional jobs (across all machines) remains at
most $\ell t$.
The rounding step amounts to rounding a fractional solution to an instance of the
{\em generalized assignment problem} (GAP), for which we can utilize the algorithm
of~\cite{ShmoysT93}, or use the iterative-rounding result from Theorem~\ref{iterrndthm}.

The budget constraint that we include for a machine is tailored to ensure that the
total $\beta_\ld\bigl(\trv{X_{ij}}{t}/4t\bigr)$-effective load on a machine under the
assignment $\sg$ is 
not too large; via Lemma~\ref{lem:smallbetasmalltail}, this will imply a suitable bound on
$\E{L_i^{\geq\Omega(1)}}$, where $L_i=\sum_\jtoi X_{ij}^{<t}/4t$.
Ideally, for each machine $i$ we would like to choose constraint \eqref{eq:betaineqtopl}
for $\ld_i = 1/\beload$.
This yields $\sum_{j} \beta_{\ld_i}\bigl(\trv{X_{ij}}{t}/4t\bigr) \bz_{ij} \leq 4 \ld_i\beload+6 =10$.
So if this budget constraint is approximately satisfied in the rounded solution, say with
RHS equal to some {\em constant} $b$, then 
Lemma~\ref{lem:smallbetasmalltail} roughly gives us 
$\E{\erv{L_i}{b+1}}\leq\paren{b+3}/\ld_i=(b+3)\beload$.
This in turn implies that 
$$
\sum_i{\boldsymbol\Exp}\Bigl[\erv{\Bigl(\sum_\jtoi \trv{X_{ij}}{t}\Bigr)}{4(b+1)t}\Bigr]=4t\cdot\sum_i\E{L_i^{\geq b+1}}\leq 
4t(b+3)\sum_i\beload\leq 4(b+3)\ell t
$$
where the last inequality follows due to \eqref{eq:truncineqtopl}.
The upshot is that $\sum_i\E{(\sum_\jtoi X_{ij}^{<t})^{\geq\Omega(t)}}=O(\ell t)$;
coupled with the fact that $\sum_j \E{X_{\sg(j),j}^{\geq t}}\leq \ell t$, we obtain
that $\sum_i\E{\erv{\load_\sg(i)}{\Omega(t)}}=O(\ell t)$, and hence
$\E{\topl(\load_\sg)}=O(\ell t)$.
The slight complication is that $1/\beload$ need not be an integer in 
$[100m]$, 
so we modify the choice of $\ld_i$s appropriately to deal with this. 

We remark that, whereas we work with a more general norm than $\ell_\infty$, our
entire approach---polynomial-size LP-relaxation, rounding algorithm, and analysis---is
in fact simpler and cleaner than the one used in~\cite{GuptaKNS18} for the special
case of $\topl[1]$ norm. 
Our savings can be traced to
the fact that we leverage the notion of effective size in a more powerful way, by
utilizing it at multiple scales to obtain lower bounds on 
$\E{\bigl(\sum_{j\mapsto i}X_{ij}^{<t}/t\bigr)^{\geq 1}}$ 
(Lemma~\ref{largebeta}). 
Our rounding algorithm is summarized below.
\begin{enumerate}[label=T\arabic*., ref=T\arabic*, topsep=1ex, itemsep=0.5ex, leftmargin=*]
\item \label{lddefs}
Define $\M_\low := \{ i \in [m] : \beload < 1/2 \}$, 
and let $\M_\hi := [m] \setminus \M_\low$. 
For every $i\in[m]$, define
$\ld_i:=\min(100m,\floor{1/\beload}) \in \{2,\dots,100m\}$ if $i\in\M_{\low}$,  
and $\ld_i:=1$ (so
$\beta_{\ld_i}\bigl(\trv{X_{ij}}{t}/4t\bigr)=\E{\trv{X_{ij}}{t}/4t}$) otherwise.

\item \label{budglp}
Consider the following LP that has one budget constraint for each machine
$i$ corresponding to constraint \eqref{eq:betaineqtopl} (slightly simplified) for
parameter $\ld_i$.
\begin{alignat}{3}
\min & \quad & \sum_{i \in [m],j \in [n]} \E{\erv{X_{ij}}{t}}\z_{ij} & \qquad && 
\tag{B-LP} \label{LBtoplobj} \\
\text{s.t.} && \sum_{j \in [n]} \eff[\ld_i]{\trv{X_{ij}}{t}/4t}\z_{ij} & \leq 10 \qquad && \forall \, i \in \M_\low \label{LBtoplbudgetlow} \\
&& \sum_{j \in [n]} \E{\trv{X_{ij}}{t}/4t}\z_{ij} & \leq 4 \beload + 6 \qquad && \forall \, i \in \M_\hi \label{LBtoplbudgethi} \\
&& \z & \in \Q^\asgn \ . \notag 
\end{alignat}
Clearly, $\bz$ is a feasible solution to \eqref{LBtoplobj}. 
Observe that $\Q^\asgn$ is the base polytope of the partition matroid encoding that each job is
assigned to at most one machine. 
We round $\bz$ to obtain an integral assignment $\sg$, either by using GAP rounding, or by
invoking Theorem~\ref{iterrndthm}.
\end{enumerate}

\paragraph{Analysis.} \label{analysistopl}
We now show that $\E{\Topl{\load_\sg}}=O(\ell t)$.
We first note that Theorem~\ref{iterrndthm} directly shows that $\sg$ satisfies constraints
\eqref{LBtoplbudgetlow} and \eqref{LBtoplbudgethi} with an additive violation of at most
$1$, and the total contribution from exceptional jobs is at most $\ell t$.

\begin{claim} \label{clm:proproundtopl}
The assignment $\sigma$ satisfies:
\begin{enumerate*}[(a)]
\item $\sum_j \E{\erv{X_{\sg(j),j}}{t}} \leq \ell t$; \ \
\item $\sum_\jtoi \eff[\ld_i]{\trv{X_{ij}}{t}/4t} \leq 11$ for all
  $i\in\M_\low$; and
\item $\sum_\jtoi \E{\trv{X_{ij}}{t}/4t} \leq 4 \beload + 7$ for all $i\in\M_\hi$.
\end{enumerate*}
\end{claim}

\begin{proof}
This follows directly from Theorem~\ref{iterrndthm} by noting that the parameter $\nu$,
denoting an upper bound on the column sum of a variable, is at most $1$ (in fact at most
$1/4$), and since $\bz$ is a feasible solution to \eqref{LBtoplobj} of objective value at
most $\ell t$. 
\end{proof}

Next, we bound $\E{\topl(\load_\sg)}$ by bounding the expected $\topl$-norm of the load induced by
three different sources: exceptional jobs, truncated jobs in $\M_\low$, and truncated jobs
in $\M_\hi$. 
Observe that \linebreak $\load_\sg = Y^\low + Y^\hi + Y^\excep$, where 
\begin{equation*}
Y^\low_i :=
\begin{cases}
\sum_\jtoi \trv{X_{ij}}{t}; & \text{if } i \in \M_\low \\
0; & \text{otherwise}
\end{cases}
\qquad
Y^\hi_i :=
\begin{cases}
\sum_\jtoi \trv{X_{ij}}{t}; & \text{if } i \in \M_\hi \\
0; & \text{otherwise}
\end{cases}
\qquad
Y^\excep_i := \sum_\jtoi \erv{X_{ij}}{t} \ .
\end{equation*}
All three random vectors follow a product distribution on $\Rp^m$.
By the triangle inequality, it suffices to bound the expected $\topl$ norm of each vector by
$\bigo{\ell t}$. 
It is easy to bound even the expected $\topl[m]$ norms of $Y^\hi$ and $Y^\excep$
(Lemma~\ref{lem:boundedtoplexcephi}); to bound $\E{\topl(Y^\low)}$
(Lemma~\ref{lem:boundedtopllow}), we utilize properties of effective sizes.

\begin{lemma} \label{lem:boundedtoplexcephi}
We have (i) $\E{\Topl{Y^\excep}} \leq \ell t$, and 
(ii) $\E{\Topl{Y^\hi}} \leq 72 \, \ell t$.
\end{lemma}

\begin{proof}
Part (i) follows immediately from Claim~\ref{clm:proproundtopl} (a) since 
$\E{\topl(Y^\excep)}\leq\E{\topl[m](Y^\excep)}=\sum_j\E{\erv{X_{\sg(j),j}}{t}}$.

For part (ii), we utilize Claim~\ref{clm:proproundtopl} (c) which gives $\sum_\jtoi \E{\trv{X_{ij}}{t}/4t}\leq 4\beload+7\leq 18\beload$ for every $i\in\M_\hi$, where the last inequality follows since $\beload\geq 1/2$ as $i\in\M_\hi$. 
It follows that $\E{Y_i^\hi}=\sum_\jtoi \E{\trv{X_{ij}}{t}}\leq 72 \, t \, \beload$.
Therefore, 
$\E{\topl(Y^\hi)}\leq\E{\topl[m](Y^\hi)}=72 \, t \, \sum_{i\in\M_\hi}\beload\leq 72 \, \ell t$.
\end{proof}

\begin{lemma} \label{lem:boundedtopllow}
$\E{\Topl{Y^\low}} \leq 232 \, \ell t$.
\end{lemma}

\begin{proof}
Let $\vecrv = Y^\low/4t$. 
It suffices to show that 
$\sum_i\E{\erv{\vecrv_i}{29}}=\sum_{i \in \M_\low} \E{\erv{\vecrv_i}{29}} \leq 29\ell$,
since then by Lemma~\ref{lem:upperboundtopl}, we have $\E{\topl(\vecrv)}\leq 58\,\ell$, or 
equivalently, $\E{\topl(Y^\low)}\leq 232\,\ell t$.
By Claim~\ref{clm:proproundtopl} (c), we have $\eff[\ld_i]{\vecrv_i} \leq 11$, where 
$\ld_i = \min(100m,\floor{1/\beload})$.  
Using Lemma~\ref{lem:smallbetasmalltail},
\[
\Eb{\erv{\vecrv_i}{12}} \leq 14/\ld_i = 14 \, \max(1/100m,1/\floor{1/\beload}) \leq 28\beload + 14/100m \ ,
\]
where the last inequality is because $\beload<1/2$.
Summing over all machines in $\M_\low$ gives $\sum_{i \in \M_\low} \E{\vecrv_i^{\geq 12}} \leq 29 \ell$, since $\sum_i \beload \leq \ell$. 
\end{proof}

Combining the two lemmas above yields the following result.

\begin{theorem} \label{thm:boundedtopl}
The assignment $\sg$ satisfies $\E{\Topl{\load_\sg}} \leq 305 \, \ell t$.
\end{theorem}

\begin{proof}[{\bf Finishing up the proof of Theorem~\ref{thm:loadbaltopl}}]
Given Theorem~\ref{thm:boundedtopl} and Claim~\ref{clm:LBfeastopl}, it is clear that if we
work with $t=O\bigl(\frac{\OPT}{\ell}\bigr)$ such that $\LPv[\ell,t]$ is feasible and run
our algorithm, 
then we obtain an $O(1)$-approximation. As is standard, we can find such a $t$, within a 
$(1+\ve)$-factor, via binary search. To perform this binary search, we show that we can 
come up with an upper bound $\UB$ such that $\frac{\UB}{m}\leq\OPT\leq\UB$. We show this
even in the general setting where we have an arbitrary monotone, symmetric norm $f$.

\begin{lemma} \label{lem:optrangeloadbal} \label{binsearch}
Let $f$ be a normalized, monotone, symmetric norm, and let $\OPT_f$ be the optimal value
of the stochastic $f$-norm load balancing problem.
Define $\UB:=\sum_j\min_i\E{X_{ij}}$, which can be easily computed from the input data.
We have $\frac{\UB}{m}\leq\OPT_f\leq\UB$.
\end{lemma}

\begin{proof}
Notice that $\UB$ is the optimal value of stochastic $\topl[m]$-norm load-balancing, i.e.,
it is the objective value of the assignment that minimizes the sum of the expected machine
loads. 
For any assignment $\sg':[n]\to[m]$, 
by Lemma~\ref{lem:sandwich}, we have
$$
\Eb{\tfrac{\topl[m](\load_{\sg'})}{m}}\leq\E{\topl[1](\load_{\sg'})}\leq\E{f(\load_{\sg'})}
\leq\E{\topl[m](\load_{\sg'})}.
$$
Taking the minimum over all assignments, and plugging in that 
$\UB=\leq\min_{\sg':[n]\to[m]}\E{\topl[m](\load_{\sg'})}$, as noted above, 
it follows that
\begin{equation*}
\tfrac{\UB}{m} 
\leq\min_{\sg':[n]\to[m]}\E{f(\load_{\sg'})}=\OPT_f
\leq 
\UB. \qedhere
\end{equation*}
\end{proof}

Thus, if we binary search in the interval $\bigl[0,\frac{2\,\UB}{\ell}\bigr]$, for any 
$\ve>0$, we can find in $\poly\bigl(\log\frac{m}{\ve}\bigr)$ iterations
$t\leq\frac{2\,\OPT}{\ell}+\frac{\ve\,\UB}{m^2}\leq(2+\ve)\cdot\frac{\OPT}{\ell}$ such that
$\LPv[\ell,t]$ is feasible. 
By Theorem~\ref{thm:boundedtopl}, we obtain an assignment whose expected
$\topl$ norm is at most $305\,\ell t\leq\bigl(610+O(\ve)\bigr)\OPT$.
\end{proof}

\subsection{Stochastic \boldmath{$f$}-norm load balancing} \label{sec:loadbalf}
We now focus on stochastic load balancing when $f$ is a general monotone symmetric norm.
Recall that $\ssg$ denotes an optimal solution, and $\OPT=\OPT_f$ denotes the optimal
value.

\paragraph{Overview.}
As noted earlier, Theorem~\ref{thm:apxstochmajor} guides our strategy: we seek an
assignment $\sg$ that simultaneously satisfies 
$\E{\Topl{\load_{\ssg}}} = \bigo{\alpha \,\E{\Topl{\load_{\ssg}}}}$ for all 
$\ell \in \POS_m$ 
for some small factor $\alpha$. 
It is natural to therefore leverage the insights gained in Section~\ref{sec:loadbaltopl} 
from the study of the $\topl$-norm problem. 

Since we need to simultaneously work with all $\topl$ norms, we now work with a guess
$t_\ell$ of the quantity $2\,\E{\topl(\load_{\ssg})}/\ell$, 
for every $\ell\in\POS$.
For each $\vec{t}=(t_\ell)_{\ell\in\POS}$ vector, we write an LP-relaxation $\LPv$ that
generalizes \ref{topllp}, and if it is feasible, we round its feasible solution to obtain an
assignment of jobs to machines. We argue that one can limit the number of $\vec{t}$ vectors to
consider to a polynomial-size set, so this
yields a polynomial number of candidate solutions. 
We remark that, interestingly, $\LPv$, the rounding algorithm, and the resulting set of
solutions generated, are {\em independent} of the norm $f$: they only depend only on the 
underlying $\vec{t}$ vector. The norm $f$ is used only in the final step 
to select one of the candidate solutions as the desired near-optimal solution, utilizing 
Corollary~\ref{detestimate}. 

The LP-relaxation we work with is an easy generalization of \ref{topllp}.
We have the usual $z_{ij}$ variables encoding a fractional assignment.
For each $\ell\in\POS$, there is a different definition of truncated
r.v. $\trv{X_{ij}}{t_\ell}$ and exceptional r.v. $\erv{X_{ij}}{t_\ell}$. 
Correspondingly, for each index $\ell\in\POS$, we have a separate set of constraints 
\eqref{eq:excepineqtopl}--\eqref{eq:truncineqtopl} involving (the $z_{ij}$s), a variable
$\eload[i,\ell]$ (that represents $\eload$ for the index $\ell$, i.e., 
$\E{\erv{(\sum_\jtoi \trv{X_{ij}}{t_\ell}/t_\ell)}{1}}$), and the guess $t_\ell$.      
For technical reasons that will become clear when we discuss the rounding algorithm (see
Claim~\ref{clm:lowcolsumexcep}), we 
include additional constraints \eqref{eq:largeexcep}, which enforce that a job $j$ cannot be 
assigned to a machine $i$ if $\E{\erv{X_{ij}}{t_1}}>t_1$; observe that this is valid for
the optimal integral solution whenever $t_1 \geq 2\,\E{\Topl[1]{\load_{\ssg}}}$.
This yields the following LP relaxation.

\noindent
\begin{minipage}{0.1\textwidth}
\leqnomode
\begin{equation}
\tag*{\LPv} \label{genflp}
\end{equation}
\reqnomode
\end{minipage}
\ 
\begin{minipage}{0.9\textwidth}
\begin{alignat}{2}
\sum_{i \in [m],j \in [n]} \E{\erv{X_{ij}}{t_\ell}} z_{ij} & \leq \ell t_\ell \qquad &&
\forall \ell\in\POS \label{eq:excepineqgen} \\
\frac{\sum_{j \in [n]} \eff{\trv{X_{ij}}{t_\ell}/4t_\ell} z_{ij} - 6}{4\ld} & \leq \eload[i,\ell] \qquad 
&& \forall i \in [m], \forall \ld\in\{1,\ldots,100m\},\,\forall \ell\in\POS \label{eq:betaineqgen} \\
\sum_{i \in [m]} \eload[i,\ell] & \leq \ell && \forall \ell\in\POS \label{eq:truncineqgen} \\
z_{ij} & = 0  && \forall i \in [m], j \in [n] \text{ with } \E{\erv{X_{ij}}{t_1}} > t_1 \label{eq:largeexcep} \\
\eload[] \geq 0, \gaptwo z & \in \Q^\asgn.  \notag 
\end{alignat}
\end{minipage}

\bigskip

Claim~\ref{clm:LBfeastopl} easily generalizes to the following.

\begin{claim} \label{clm:LBfeasgen}
Let $\vt$ be such that $t_\ell \geq 2\E{\Topl{\load_{\ssg}}}/\ell$ for all $\ell \in \POS$. 
Then, $\LPv$ is feasible.
\end{claim}

Designing an LP-rounding algorithm is substantially more challenging now. In the
$\topl$-norm case, we set up an auxiliary LP \eqref{LBtoplobj} by extracting a single
budget constraint for each machine that served to bound the contribution from the
truncated jobs on that machine. 
This LP was quite easy to round (e.g., using Theorem~\ref{iterrndthm}) because each
$z_{ij}$ variable participated in exactly one constraint (thereby trivially yielding an
$O(1)$ bound on the column sum for each variable). 
Since we now have to simultaneously control multiple $\topl$-norms, for each machine $i$,
we will now need to include 
a budget constraint {\em for every index $\ell\in\POS$} so as to
bound the contribution from the truncated jobs for index $\ell$ (i.e., 
$\E{\erv{\bigl(\sum_\jtoi \trv{X_{ij}}{t_\ell}\bigr)}{\Omega(t_\ell)}}$).
Additionally, unlike \eqref{LBtoplobj}, wherein 
the contribution from the exceptional jobs was bounded by incorporating it in the
objective,  
we will now need, for each $\ell\in\POS$, a separate constraint to bound the total
contribution from the exceptional jobs for index $\ell$.

Thus, while we can set up an auxiliary LP similar to \eqref{LBtoplobj} containing these
various budget constraints (see \eqref{eq:budgetlow} and \eqref{eq:budgethi}), rounding a
fractional solution to this LP to obtain an assignment that approximately satisfies these
various budget constraints presents a significant technical hurdle. As alluded to above,
every $z_{ij}$ variable now participates in {\em multiple} budget constraints, which makes
it difficult to argue a bounded-column-sum property for this variable and thereby leverage
Theorem~\ref{iterrndthm}. 
(The multiple budget constraints included for each machine $i$ to bound the
contribution from the truncated jobs for each $\ell\in\POS$ present the main difficulty; 
one can show that the column sum if we only consider the constraints bounding the
contribution from the exceptional jobs is $O(1)$ (see Claim~\ref{clm:lowcolsumexcep}).)

In Sections~\ref{sec:loadbalber} and~\ref{sec:loadbalgen}, we show how to partly overcome
this obstacle. 
Section~\ref{sec:loadbalber} considers the setting where job sizes are
Bernoulli random variables. Here, we show 
that an auxiliary LP as above (see \eqref{berauxlp}) yields a constraint matrix with
$O(1)$-bounded column sums.   
Consequently, this auxiliary LP can be rounded with an $O(1)$ violation of all budget
constraints (using Theorem~\ref{iterrndthm}), which then (with the right choice of $\vt$)
leads to an $O(1)$-approximation algorithm (Theorem~\ref{thm:loadbalber}).  
In Section~\ref{sec:loadbalgen}, we consider the general case. We argue that we can
set up an auxiliary LP that imposes a {\em weaker} form of budget constraints involving
expected truncated job sizes, and does have $O(1)$ column sums.  
Via a suitable use of Chernoff bounds, this then leads to an 
$O(\log m/\log\log m)$-approximation for general stochastic $f$-norm load balancing
(Theorem~\ref{thm:loadbalgen}). 

We do not know how to overcome the impediment discussed above in setting up the auxiliary LP
for the general setting, 
Specifically, we do not know how to set up an auxiliary LP with a suitable choice of
budget constraints for each machine $i$ and $\ell\in\POS$ that: (a) imposes the desired
bound on the contribution from the truncated jobs for index $\ell$ within $O(1)$ factors;
and (b) yields a constraint matrix with $O(1)$-bounded column sums. 
We leave the question of determining the integrality gap of \ref{genflp}, as 
also the technical question of setting up a suitable auxiliary LP, 
in the general case, as intriguing open problems.  
We believe however that our LP-relaxation \ref{genflp} is in fact better than what we
have accounted for. 
In particular, we remark that if the Koml\'os conjecture in discrepancy theory is
true, then it is quite likely that the auxiliary LP \eqref{berauxlp} that we formulate in the
Bernoulli case can be rounded with {\em at most} an $O(\sqrt{\log m})$-violation of the
budget constraints, which would yield an $O(\sqrt{\log m})$ upper bound on the the
integrality gap of \ref{genflp}. 

The final step involved is choosing the ``right'' $\vt$ vector, 
and this is the only place where we use the (value oracle for) the norm $f$.
Clearly, binary search is not an option now. Noting that the vector $\vec{t^*}$
corresponding to the optimal solution (so $t^*_\ell$ is supposed to be a guess of 
$2\cdot\E{\topl(\load_{\ssg})}/\ell$) is a vector with non-increasing coordinates, and the
bounds in Lemma~\ref{binsearch} yield an $O(\log m)$ search space for each $t^*_\ell$ (by
considering powers of $2$), 
there are only a polynomial number of $\vt$
vectors to consider (see Claim~\ref{clm:needleinapolyhaystack}). We can give a randomized
value oracle for evaluating 
$\E{f(Y)}$ for the cost vector $Y$ 
resulting from each $\vt$ vector (for which $\LPv$ is feasible), and therefore arrive at
the best solution computed; this yields a randomized approximation guarantee.%
\footnote{We would obtain an approximation guarantee that holds with high probability, or
for the expected cost of the solution, where the randomness is now over the random choices
involved in the randomized value oracle.}
We can do better and obtain a {\em deterministic} (algorithm and) approximation guarantee
by utilizing Corollary~\ref{detestimate}: this shows that for each $\vt$ (for which
$\LPv$ is feasible), we can define a corresponding vector $\bvec=\bvec(\vt)\in\Rp^m$
such that  $f\bigl(\bvec(\vt)\bigr)$ acts as a good proxy for 
the objective value $\E{f(Y)}$ of the solution computed for $\vt$, provided that 
$\E{\topl(Y)}=O\bigl(\topl(\bvec)\bigr)$ for all $\ell\in\POS$.
It follows that the solution corresponding to the smallest $f\bigl(\bvec(\vt)\bigr)$ 
is a near-optimal solution.

\subsubsection{Rounding algorithm and analysis: Bernoulli jobs} 
\label{sec:loadbalber} 
Assume that $\vt\in\Rp^\POS$ is such that $\LPv$ is feasible.
We will further assume that 
$t_\ell$ is a power of $2$ (possibly smaller than $1$) for all $\ell \in \POS$,
and $t_\ell / t_{2\ell} \in \{1,2\}$ for all $\ell$ such that $\ell,2\ell \in\POS$.
The rationale here is that these properties are satisfied when $\vt_\ell$ is the smallest power of $2$ that is at least $2\,\E{\topl(\load_{\ssg})}/\ell$.
Let $(\bz,\beload[])$ be a feasible fractional solution to this LP.
We show how to round $\bz$ to obtain an assignment $\sg$ satisfying
$\E{\topl(\load_\sg)}=O(1)\cdot\ell t_\ell$ for all $\ell\in\POS$. 

Our strategy is similar to that in Section~\ref{sec:loadbaltopl}. As discussed earlier, we 
set up an auxiliary LP \eqref{berauxlp} 
where we include a budget constraint for: 
(i) each machine $i$ and index $\ell\in\POS$ from \eqref{eq:betaineqgen}, 
to bound the contribution from the truncated jobs for index $\ell$ on machine $i$;
(ii) each $\ell\in\POS$ to bound the total contribution from the exceptional jobs for
index $\ell$.
We actually refine the above set of constraints to drop various redundant budget
constraints (i.e., constraints implied by other included constraints), so as to enable us
to prove an $O(1)$ column sum for the resulting constraint matrix. 
The fractional assignment $\bz$ yields a feasible solution to \eqref{berauxlp}. 
Now we really need to utilize the full power of Theorem~\ref{iterrndthm} to round 
$\bz$
to obtain an assignment that approximately satisfies these constraints. 
In the analysis, we show that approximately satisfying the budget
constraints of \eqref{berauxlp} ensures that $\E{\topl(\load_\sg)}=O(\ell t_\ell)$. This
argument is slightly more complicated now, due to the fact that we drop some redundant
constraints, 
but it is along the same lines as that in
Section~\ref{sec:loadbaltopl}. We now delve into the details.

\begin{enumerate}[label=B\arabic*., ref=T\arabic*, topsep=0.5ex, itemsep=0.5ex, leftmargin=*]
\item 
For each $\ell \in \POS$, define the following quantities.
Define $\M^\ell_\low := \{ i \in [m] : \beload[i,\ell] < 1/2 \}$, 
and let $\M^\ell_\hi := [m] \setminus \M^\ell_\low$. 
For every $i \in [m]$, set 
$\ld_{i,\ell} := \Min{100m,\floor{1/\beload[i,\ell]}}\in\{2,\dots,100m\}$ if $i\in\M^\ell_\low$,
and $\ld_{i,\ell}:=1$ otherwise.

Define $\POSl := \{\ell \in \POS : \ell = 1 \text{ or } t_\ell = \frac{t_{\ell/2}}{2} \}$,
and $\POSe := \{ \ell \in \POS: 2\ell \notin \POS \text{ or } t_\ell = t_{2\ell} \}$.

\item The auxiliary LP will enforce constraints
\eqref{eq:excepineqgen}, and constraint \eqref{eq:betaineqgen} for parameter $\ld_{i,\ell}$,
for every machine $i$ and $\ell\in\POS$, but we remove various redundant constraints.
Notice that if $\ell\notin\POSl$ (so $t_\ell=t_{\ell/2}$), then
constraints \eqref{eq:betaineqgen} for $\ell$ and $\ell/2$ have the same LHS. We may
therefore assume that $\beload[i,\ell/2]=\beload[i,\ell]$ for every machine $i$, and hence
constraint \eqref{eq:truncineqgen} for index $\ell$ is implied by \eqref{eq:truncineqgen}
for index  $\ell/2$. 
Similarly, \eqref{eq:excepineqgen} for index $\ell$ is implied by 
\eqref{eq:excepineqgen} for index $\ell/2$. Also, if $\ell\notin\POSe$ (so 
$t_{\ell}=2t_{2\ell}$), then constraint \eqref{eq:excepineqgen} for index $\ell$ is implied by this constraint for index $2\ell$. 
The auxiliary LP is therefore as follows.

\begin{minipage}{0.15\textwidth}
\leqnomode
\begin{equation}
\tag{Ber-LP} \label{berauxlp}
\end{equation}
\reqnomode
\end{minipage}
\
\begin{minipage}{0.8\textwidth}
\begin{alignat}{2}
\sum_{i \in [m],j \in [n]}\Bigl(\E{\erv{X_{ij}}{t_\ell}}/(\ell t_\ell)\Bigr)\z_{ij} & \leq 1
\qquad && \forall \ell \in \POSl\cap\POSe \label{eq:excepl} \\
\sum_{j \in [n]} \eff[\ld_{i,\ell}]{\trv{X_{ij}}{t_\ell}/4t_\ell}\z_{ij} & \leq 10 
\qquad && \forall \, \ell \in \POSl, \, \forall\,i \in \M^\ell_\low \label{eq:budgetlow} \\
\sum_{j \in [n]} \E{\trv{X_{ij}}{t_\ell}/4t_\ell}\z_{ij} & \leq 4 \, \beload[i,\ell] + 6 
\qquad && \forall \, \ell \in \POSl, \, \forall\,i \in \M^\ell_\hi \label{eq:budgethi} \\
\z & \in \Q^\asgn. \notag 
\end{alignat}
\end{minipage}

We have scaled constraints \eqref{eq:excepineqgen} for the exceptional jobs
to reflect the fact that we can afford to incur an $O(\ell t_\ell)$ error in the
constraint for index $\ell$. 

\item The fractional assignment $\bz$ is a feasible solution to \eqref{berauxlp}. We apply
Theorem~\ref{iterrndthm} with the above system to round $\bz$ and obtain an integral
assignment $\sg$.
\end{enumerate} 

\paragraph{Analysis.}
We first show that the constraint
matrix of \eqref{berauxlp} has $O(1)$-bounded column sums for the support of $\bz$. 
Claim~\ref{clm:lowcolsumexcep} shows this for the budget constraints for the exceptional
jobs; 
this bound holds for any distribution. 

\begin{claim} \label{clm:lowcolsumexcep}
Let $i,j$ be such that $\bz_{ij} > 0$. 
Then $\sum_{\ell \in \POSl\,\cap\,\POSe} \bigl(\E{\erv{X_{ij}}{t_\ell}}/\ell t_\ell\bigr) \leq 4$.
\end{claim} 

\begin{proof}
Crucially, we first note that $\E{X_{ij}^{\geq t_1}} \leq t_1$ due to \eqref{eq:largeexcep}. 
So, $\E{X_{ij}^{\geq t_\ell}} \leq \E{X_{ij}^{\geq t_1}} + \E{X_{ij}^{< t_1}} \leq 2 t_1$. 
Observe that if $\ell \in \POSl\cap\POSe$, then $2\ell \not\in \POSl$.
Also, $t_\ell$ drops by a factor of exactly $2$ as $\ell$ steps over the indices in $\POSl$.
So, $\ell t_\ell$ increases by a factor of at least $2$ as $\ell$ steps over the indices
in $\POSl\cap\POSe$. Thus,
\begin{equation*}
\sum_{\ell \in \POSl\cap\POSe} \frac{\E{X_{ij}^{\geq t_\ell}}}{\ell t_\ell} \leq 2 t_1
\paren{ \frac{1}{t_1} + \frac{1}{2t_1} + \frac{1}{4t_1} + \dots  } \leq 4 \ . \qedhere
\end{equation*}
\end{proof}

Claim~\ref{clm:lowcolsumbudget} shows the $O(1)$ column sum for
constraints \eqref{eq:budgetlow} and \eqref{eq:budgethi}.
The $O(1)$ column sum for constraints \eqref{eq:budgethi} in fact holds for any
distribution (see Claim~\ref{clm:lowcolsumexpct}); the $O(1)$ column sum for constraints
\eqref{eq:budgetlow} relies on the fact that the job sizes are Bernoulli random random
variables. 
Let $a^\ell_{i,j}$ denote the coefficient of $\bz_{ij}$ in constraints
\eqref{eq:betaineqgen} and \eqref{eq:truncineqgen}.

\begin{claim} \label{clm:lowcolsumbudget}
When the $X_{ij}$s are Bernoulli random variables, we have 
$\sum_{\ell \in \POSl} a^\ell_{ij} \leq 1$ for any $i,j$. 
\end{claim}

\begin{proof}
Fix some machine $i$ and job $j$.
For any $[0,\scal]$-bounded r.v. $\scrv$ and a parameter $\ld \in \R_{\geq 1}$, we have
$\eff{\scrv} \leq \scal$ (see Definition~\ref{defn:effectivesize}). 
Suppose that $X_{ij}$ is a Bernoulli 
random variable that takes value $s$ with probability $q$, and value $0$ with probability
$1-q$. 
For any $\ell \in \POSl$, we have 
$a^\ell_{ij} \leq \frac{s \cdot \mathbbm{1}_{\paren{s<t_\ell}}}{4 t_\ell}$; recall that 
$\mathbbm{1}_{\mce} \in \{0,1\}$ is $1$ if and only if the event $\mce$ happens. 
By definition of $\POSl$, the $t_\ell$s decrease geometrically over $\ell \in \POSl$, so
\begin{equation*}
\sum_{\ell \in \POSl} a^\ell_{ij} \leq \sum_{\ell \in \POSl} \frac{s \cdot \mathbbm{1}_{\paren{s < t_\ell}}}{4 t_\ell} = \frac14 \, \sum_{\ell \in \POSl} \frac{s}{t_\ell} \, \mathbbm{1}_{\paren{s/t_\ell < 1}} \leq \frac{ 1 + 1/2 + 1/4 + \dots }{4} \leq \frac12 \ . 
\qedhere
\end{equation*}
\end{proof}

As before, we analyze 
$\E{\topl(\load_\sg)}$ by bounding the expected $\topl$ norm of the exceptional jobs for
index $\ell$, the truncated jobs for index $\ell$ in $\M^\ell_\low$, and the truncated
jobs for index $\ell$ in $\M^\ell_\hi$. 
Define $Y^{\excep,\ell}_i := \sum_\jtoi \erv{X_{ij}}{t_\ell}$, and
\begin{equation*}
Y^{\low,\ell}_i :=
\begin{cases}
\sum_\jtoi \trv{X_{ij}}{t_\ell}; & \text{if } i \in \M^\ell_\low \\
0; & \text{otherwise}
\end{cases}
\qquad \qquad
Y^{\hi,\ell}_i :=
\begin{cases}
\sum_\jtoi \trv{X_{ij}}{t_\ell}; & \text{if } i \in \M^\ell_\hi \\
0; & \text{otherwise}
\end{cases}
\end{equation*}
We have $\load_\sg=Y^{\low,\ell}+Y^{\hi,\ell}+Y^{\excep,\ell}$.
Given the above claims, it follows from Theorem~\ref{iterrndthm} (b) that the assignment $\sg$
satisfies the constraints of \eqref{berauxlp} with an additive violation of at most
$5$. 
This will enable us to bound 
$\E{\topl(Z)}$, where $Z\in\{Y^{\low,\ell},Y^{\hi,\ell},Y^{\excep,\ell}\}$, by
$O(\ell t_\ell)$ (Lemmas~\ref{lem:boundedgenexcep}--\ref{lem:boundedgenlow}), and hence show that
$\E{\topl(\load_\sg)}=O(\ell t_\ell)$ (Lemma~\ref{lem:boundedgen}).
The proofs follow the same template as those in Section~\ref{sec:loadbaltopl}, but we
need to also account for the redundant constraints that were dropped.
Recall that \linebreak
$\POSl := \{\ell \in \POS : \ell = 1 \text{ or } t_\ell = \frac{t_{\ell/2}}{2} \}$,
and $\POSe := \{ \ell \in \POS: 2\ell \notin \POS \text{ or } t_\ell = t_{2\ell} \}$.

\begin{lemma} \label{clm:boundedgenexcep} \label{lem:boundedgenexcep}
We have $\E{\Topl{Y^{\excep,\ell}}}\leq 6\,\ell t_\ell$ for all $\ell\in\POS$.
\end{lemma}

\begin{proof}
Fix $\ell\in\POS$.
We upper bound $\E{\topl[m](Y^{\excep,\ell})}=\sum_j \E{\erv{X_{\sg(j),j}}{t_\ell}}$ by $6\ell t_\ell$.
If $\ell \in \POSl\cap\POSe$, then by Theorem~\ref{thm:budgetedmatrounding} (b), we have
$\sum_j \E{\erv{X_{\sg(j),j}}{t_\ell}} \leq 6 \ell t_\ell$ 
since the column sum for each $\z_{ij}$ variable with $\bz_{ij}>0$ is at most $5$ by
Claims~\ref{clm:lowcolsumexcep} and~\ref{clm:lowcolsumbudget}.

Next, suppose that $\ell \in \POS \setminus \POSl$, so $t_\ell = t_{\ell/2}$. 
Let $\ell'$ be the largest index in $\POSl$ that is at most $\ell$. Then,
$\ell'\leq\ell/2$ and $t_\ell=t_{\ell'}$. So $Y^{\excep,\ell} = Y^{\excep,\ell'}$. 
Also, $\ell'\in\POSe$, since $t_\ell\leq t_{2\ell'}\leq t_{\ell'}=t_\ell$.
Therefore,
$\E{\topl[m](Y^{\excep,\ell})}=\E{\topl[m](Y^{\excep,\ell'})}\leq 6\ell't_{\ell'}\leq 6 \ell t_\ell$.

Finally, suppose $\ell\in\POS\sm\POSe$, so $t_\ell=2t_{2\ell}$. Now, let $\ell'$ be the
smallest index in $\POSe$ that is at least $\ell$. Then, $\ell'\geq 2\ell$, and
$t_{\ell'/2}>t_{\ell'}$ (since $\ell'/2\notin\POSe$), so $\ell'\in\POSl$.
We claim that $\ell't_{\ell'}=\ell t_\ell$. 
This is because for every $\ell''\in\POS$, $\ell\leq\ell''<\ell'$, we have
$t_{\ell''}=2t_{2\ell''}$, and so $\ell''t_{\ell''}=2\ell''t_{2\ell''}$.
Hence, we have  
$\E{\topl[m](Y^{\excep,\ell})}\leq\E{\topl[m](Y^{\excep,\ell'})}\leq 6\ell't_{\ell'}=6\ell
t_\ell$.
\end{proof}

\begin{lemma} \label{clm:boundedgenhi} \label{lem:boundedgenhi}
We have $\E{\Topl{Y^{\hi,\ell}}}\leq 104\,\ell t_\ell$ for all $\ell\in\POS$. 
\end{lemma}

\begin{proof}
Fix $\ell\in\POS$.
We again show that in fact $\E{\topl[m](Y^{\hi,\ell})}$ is at most $104\,\ell t_\ell$.
Note that we only need to show this for $\ell\in\POSl$: if $\ell\in\POS\sm\POSl$, 
then if we consider the largest index $\ell'\in\POSl$ that is at
most $\ell$, we have $t_{\ell'}=t_\ell$, and so $Y^{\hi,\ell'}=Y^{\hi,\ell}$; thus, the
bound on $\E{\topl[m](Y^{\hi,\ell})}$ follows from that on $\E{\topl[m](Y^{\hi,\ell'})}$.
So suppose $\ell\in\POSl$.
We have 
$$
\E{Y^{\hi,\ell}_i}=\sum_\jtoi \E{\trv{X_{ij}}{t_\ell}}
\leq 4t_\ell(4\beload[i,\ell]+11) \leq 4t_\ell\cdot 26\beload[i,\ell].
$$
The first inequality follows from Theorem~\ref{iterrndthm} (b)
Summing over all $i$, since $\sum_i\beload[i,\ell]\leq\ell$, we obtain that 
$\E{\topl[m](Y^{\hi,\ell})}\leq 104\, \ell t_\ell$.
\end{proof}

\begin{lemma} \label{clm:boundedgenlow} \label{lem:boundedgenlow}
We have $\E{\Topl{Y^{\low,\ell}}}\leq 296\,\ell t_\ell$ for all $\ell\in\POS$. 
\end{lemma}

\begin{proof}
As with Lemma~\ref{lem:boundedgenhi}, we only need to consider $\ell\in\POSl$.
So fix $\ell \in \POSl$. 
Let $\vecrv=Y^{\low,\ell}/4t_\ell$.
Mimicking the proof of Lemma~\ref{lem:boundedtopllow}, it suffices to show that 
$\sum_{i \in \M^\ell_\low} \E{\erv{\vecrv_i}{37}} \leq 37\ell$, as 
Lemma~\ref{lem:upperboundtopl} then implies that $\E{\Topl{\vecrv}}\leq 74\,\ell$, or
equivalently $\E{\Topl{Y^{\low,\ell}}} \leq 296 \, \ell t_\ell$.

By Theorem~\ref{thm:budgetedmatrounding} (b), for any $i \in \M^\ell_\low$, we have 
$\eff[\ld_{i,\ell}]{\vecrv_i}=\sum_\jtoi \beta_{\ld_{i,\ell}}\bigl(\trv{X_{ij}}{t_\ell}/4t_\ell\bigr)\leq 15$.
So by Lemma~\ref{lem:smallbetasmalltail}, we have
$\E{\erv{\vecrv_i}{16}}\leq 18/\ld_{i,\ell}\leq 
18\,\max(\frac{1}{100m},\frac{1}{\floor{1/\beload[i,\ell]}}) \leq 36\beload[i,\ell]+18/100m$.
This implies that $\sum_{i \in \M^\ell_\low} \E{\erv{\vecrv_i}{16}} \leq 37\ell$.
\end{proof}

\begin{lemma} \label{lem:boundedgen}
The assignment $\sg$ satisfies 
$\E{\Topl{\load_\sg}} \leq 406 \, \ell t_\ell$.
\end{lemma}

\subsubsection*{Finishing up the proof of Theorem~\ref{thm:loadbalber}}
For $\ell \in \POS$, let $t^*_\ell$ be the smallest number of the form $2^k$, where $k\in\Z$ (and could be negative) such that $t^*_\ell\geq 2 \, \E{\topl(\load_{\ssg})}/{\ell}$.
Given Claim~\ref{clm:LBfeasgen} and Lemma~\ref{lem:boundedgen}, we know that 
the solution $\sg$ computed for $\vec{t^*}$ will be an $O(1)$-approximate
solution. We argue below that we can identify a $\poly(m)$-size set $\T\sse\Rp^\POS$
containing $\vec{t^*}$, and $\vec{t^*}$ satisfies the assumptions made at the beginning of Section~\ref{sec:loadbalber}, namely $t^*_\ell/t^*_{2\ell}\in\{1,2\}$ for all $\ell$ such
that $\ell,2\ell\in\POS$. 
While we cannot necessarily identify $\vts$, we use Corollary~\ref{detestimate} to infer
that 
if $\LPv$ is feasible, then we can use $\vt$ to come up with a good estimate of the
expected $f$-norm of the solution computed for $\vt$. This will imply that the ``best''
vector $\vt\in\T$ yields an $O(1)$-approximate solution.

By Lemma~\ref{binsearch}, we have a bound on $\OPT$ and $t^*_1$:
$\frac{\UB}{m} \leq \E{\Topl[1]{\load_{\ssg}}} \leq \E{f(\load_{\ssg})} \leq \UB$.
Observe that $\E{\topl(\load_{\ssg})}/{\ell}$ does not increase as $\ell$ increases.
It follows that $\frac{2\,\UB}{m^2}\leq 2\,\E{\topl(\load_{\ssg})}/{\ell}\leq 2\,\UB$ for
all $\ell\in\POS$. 
Define
\begin{equation}
\T := \curly{ \vt \in \Rp^\POS:\  
\forall\ell\in\POS,\quad \frac{2 \, \UB}{m^2} \leq t_\ell < 4 \, \UB,\ \ 
t_\ell \text{ is a  power of $2$},\ \ 
\frac{t_\ell}{t_{2\ell}}\in \{1,2\}\text{ whenever $2\ell\in\POS$}}
\label{setofguesses}
\end{equation}

\begin{claim} \label{clm:needleinapolyhaystack}
The vector $\vts$ belongs to $\T$, and $|\T|=O(m\log m)$. 
\end{claim}

\begin{proof}
Note that each vector $\vt\in\T$ is completely defined by specifying $t_1$, and the set of
``breakpoint'' indices $\ell\in\POS$ for which $\frac{t_\ell}{t_{2\ell}}=2$.
There are $O(\log m)$ choices for $t_1$, and at most $2^{|\POS|}\leq m$ choices for the
set of breakpoint indices; hence, $|\T|=O(m\log m)$.
To see that $\vts\in\T$, by definition, each $t^*_\ell$ is a power of $2$.
The bounds shown above on $2\,\E{\topl(\load_\ssg)}/\ell$
show that $t^*_\ell$ lies within the stated bounds. 
The only nontrivial condition to check is $t^*_\ell/t^*_{2\ell} \in \{1,2\}$ whenever
$2\ell \in \POS$. 
We have $\frac{2\ell \, t^*_{2\ell}}{2} \geq \E{\Topl[2\ell]{\load_{\ssg}}} \geq \E{\Topl{\load_{\ssg}}} > \frac{\ell t^*_\ell}{4}$, which implies $t^*_\ell < 4 \, t^*_{2\ell}$.
Next, by $\Topl[2\ell]{\cdot} \leq 2 \Topl{\cdot}$ we have $\frac{\ell t^*_{\ell}}{2} \geq \E{\Topl{\load_{\ssg}}} \geq \frac12 \E{\Topl[2\ell]{\load_{\ssg}}} > \frac{2\ell t^*_{2\ell}}{8}$, which implies $t^*_{\ell} > \frac12 \, t^*_{2\ell}$.
Since $t^*_\ell$s are powers of $2$, this implies $t^*_\ell/t^*_{2\ell} \in \{1,2\}$.
\end{proof}

We enumerate over all guess vectors $\vt \in \T$ and check if $\LPv$ is feasible.
For each $\vt \in \T$, and $\ell\in\POS$, define $\budg_\ell(\vt\,) := \ell t_\ell$. 
Let $\bvec(\vt\,)\in\Rp^m$ be the vector given by Lemma~\ref{bvecest} for 
the sequence $\paren{\budg_\ell(\vt\,) }_{\ell \in \POS}$.
Among all feasible $\vt$ vectors, let $\vtb$ be a vector that minimizes
$f\bigl(\bvec(\vtb\,)\bigr)$. Let $\overline{\sg}$ be the assignment computed by our
algorithm by rounding a feasible solution to $\LPv[\vtb]$. 

Then, we have $f\bigl(\bvec(\vtb\,)\bigr)\leq f\bigl(\bvec(\vts\,)\bigr)$. Since
$\budg_\ell(\vts\,)=\ell t^*_\ell\leq 4\E{\topl(\load_{\ssg})}$ for all $\ell\in\POS$,
Corollary~\ref{detestimate} (b) shows that 
$f\bigl(\bvec(\vts\,)\bigr)\leq 12\,\E{f(\load_{\ssg})}$; hence, we have
$f\bigl(\bvec(\vtb\,)\bigr)\leq 12\,\OPT$.

Our rounding algorithm ensures that 
$\E{\Topl{\load_{\overline{\sg}}}}\leq 406\,\budg_\ell(\vtb\,)$ for all $\ell \in \POS$
(Lemma~\ref{lem:boundedgen}). Hence, by Corollary~\ref{detestimate} (a), we have
$\E{f(\load_{\overline{\sg}})}\leq 2\cdot 28\cdot 406\cdot\bigl(f(\bvec(\vtb\,))\bigr)=O(1)\cdot\OPT$. 
\hfill \qed

\subsubsection{Rounding algorithm and analysis: general distributions} \label{sec:loadbalgen}

As in Section~\ref{sec:loadbalber}, we assume that $\vt \in \Rp^\POS$ is such that $\LPv$ is feasible,
$t_\ell$ is a power of $2$ (possibly smaller than $1$) for all $\ell \in \POS$, and
$t_\ell / t_{2\ell} \in \{1,2\}$ for all $\ell$ such that $2\ell \in \POS$. 
Let $(\bz,\beload[])$ be a feasible fractional solution to this LP.
We show how to round $\bz$ to obtain an assignment $\sg$ satisfying $\E{\Topl{\load_\sg}} = \bigo{\log m/\log \log m} \ell t_\ell$ for all $\ell \in \POS$.

Our strategy is simpler than the one used in Section~\ref{sec:loadbalber}. 
We set up an auxiliary LP which is a specialization of \eqref{berauxlp}: for each machine $i$ and index $\ell \in \POSl$ we let $\ld_{i,\ell} := 1$, or equivalently, we let $\M^\ell_\low = \emptyset$ and $\M^\ell_\hi = [m]$.
For each $\ell \in \POS$, the constraint~\eqref{eq:excepl} bounding the total contribution from exceptional jobs is retained as is.
Formally, for the truncated jobs at any index $\ell \in \POSl$, we are only bounding their expected contribution to the load vector.
The fractional assignment $\bz$ yields a feasible solution to this auxiliary LP.
We invoke Theorem~\ref{iterrndthm} with this simpler auxiliary LP to round $\bz$ and obtain an assignment $\sg$.

\paragraph{Analysis.}
We first show that the constraint matrix of the above auxiliary LP has $O(1)$-bounded column sums for the support of $\bz$.
Recall from Claim~\ref{clm:lowcolsumexcep} that for any job size distributions the budget constraints for the exceptional jobs have $\bigo{1}$-bounded column sums.
Claim~\ref{clm:lowcolsumexpct} shows $O(1)$ column sum for
constraints \eqref{eq:budgethi} in \eqref{berauxlp}. 
Note that there are no constraints of the type \eqref{eq:budgetlow} since $\M^\ell_\low = \emptyset$.

\begin{claim} \label{clm:lowcolsumexpct}
Let $X_{ij}$s be arbitrary random variables. 
Then, $\sum_{\ell \in \POSl} \E{\trv{X_{ij}}{t_\ell}/4 t_\ell} \leq 1$ for any $i,j$.
\end{claim}

\begin{proof}
Fix some machine $i$ and job $j$.
Recall that $\POSl = \{ \ell \in \POS : \ell = 1 \text{ or } t_{\ell} = \frac{t_{\ell/2}}{2} \}$.
So, $t_\ell$ drops by a factor $2$ as $\ell$ steps over the indices in $\POSl$.
For the sake of simplicity, suppose that $X_{ij}$ is a discrete random variable; the same argument can be extended to continuous distributions.
For $s \in \supp(X_{ij})$, let $q_s=\Pr[X_{ij} = s]$. Then,
\begin{equation*}
\sum_{\ell \in \POSl} \E{\trv{X_{ij}}{t_\ell}/4 t_\ell} = 
\sum_{s \in \supp(X_{ij})} \frac{q_s}{4}\cdot\biggl(\sum_{\ell \in \POSl:t_\ell>s} \frac{s}{t_\ell}\biggr) 
\leq \sum_{s \in \supp(X_{ij})} \frac{q_s}{4} \paren{
  1 + 1/2 + 1/4 + \dots } \leq \frac12 \ . 
\qedhere
\end{equation*}
\end{proof}

Claims~\ref{clm:lowcolsumexcep}~and~\ref{clm:lowcolsumexpct} imply that the column-sums across the budget constraints are bounded by $5$.

\begin{theorem} \label{roundinglogmloglogm}
Let $\sg$ denote the assignment obtained by invoking Theorem~\ref{iterrndthm} for rounding the solution $\bz$ with the parameter $\viol = 5$.
Then,
\begin{enumerate}[(i), topsep=0.5ex, itemsep=0.25ex, leftmargin=*]
	\item for any index $\ell \in \POS$, $\sum_{j} \E{\erv{X_{\sg(j),j}}{t_\ell}} \leq 6 \, \ell t_\ell$;
	\item for any machine $i$ and index $\ell \in \POS$, $\sum_\jtoi \E{\trv{X_{ij}}{t_\ell}} \leq (16\beload[i,\ell] + 44)t_\ell$.
	\item for any index $\ell\in \POS$, $\E{\Topl{\load_\sg}} = \bigo{\log m/\log \log m} \ell t_\ell$.
\end{enumerate}
\end{theorem}
\begin{proof}
By Theorem~\ref{iterrndthm} (b), part (i) holds for any $\ell \in \POSl \cap \POSe$, and
repeating the argument from Lemma~\ref{lem:boundedgenexcep} implies the bound for
remaining indices in $\POS$. 
Again, Theorem~\ref{iterrndthm} (b) yields part (ii) for any index $\ell \in \POSl$; the
bound also holds for remaining $\ell \in \POS$ since $t_\ell = t_{\ell/2}$. 

We prove part (iii) by separately bounding the $\topl$ norm of the load vectors induced by
the truncated jobs and the exceptional jobs for any $\ell \in \POS$. 
Let $Y^{\bd,\ell}$ and $Y^{\excep,\ell}$ be defined as follows: for each $i$, $Y^{\bd,\ell}_i := \sum_\jtoi \trv{X_{ij}}{t_\ell}$; and $Y^{\excep,\ell}_i := \sum_\jtoi \erv{X_{ij}}{t_\ell}$. 
By construction, $\load_\sg = Y^{\bd,\ell} + Y^{\excep,\ell}$.
As before, the treatment for the load induced by exceptional jobs is easy to handle. 
The first conclusion implies a bound on the $\topl[m]$ norm of $Y^{\excep,\ell}$: $\E{\Topl[m]{Y^{\excep,\ell}}} = \sum_j \E{\erv{X_{\sg(j),j}}{t_\ell}} \leq 6 \ell t_\ell$.
So, the same bound extends to the (possibly smaller) $\topl$ norm of $Y^{\excep,\ell}$.

Since each $Y^{\bd,\ell}_i$ is a sum of independent $[0,t_\ell)$-bounded r.v.s, we can use standard Chernoff-style concentration inequalities to bound $\E{\Topl{Y^{\bd,\ell}}}$ by $\bigo{\log m/\log \log m} \ell t_\ell$; this is sufficient since triangle inequality implies the bound given in (iii).
To this end, fix $\ell \in \POS$.
The following fact will be useful.
Let $\scal > 0$ be a scalar and $y \in \Rp^m$ be some vector.
Suppose that for each $i \in [m]$, $y_i \leq \scal (16 \beload[i,\ell] + 44)t_\ell$ holds.
Then, 
\[
\Topl{y} = \max_{S \subseteq [m], |S| = \ell} \sum_{i \in S} y_i \leq \scal \paren{16 t_\ell \sum_i \beload[i,\ell] + 44 \ell t_\ell} \leq 60 \scal \ell t_\ell \ ,
\] 
where $\sum_i \beload[i,\ell] \leq \ell$ follows from the feasibility of $(\bz,\beload[])$ to $\LPv$.

For any $\scal \in \Rp$, let $q(\scal)$ denote the probability of the event:
$\bigl\{\exists i\ \text{s.t.}\ Y^{\bd,\ell}_i > \alpha (16 \beload[i,\ell] + 44)t_\ell\bigr\}$. 
By the above, we have that $\pr{\Topl{Y^{\bd,\ell}}>60\scal\ell t_\ell}$ is at most $q(\scal)$.
Therefore
\begin{equation*}
\begin{split}
\E{\topl(Y^{\bd,\ell})} & = \int_0^\infty \pr{\topl(Y^{\bd,\ell})>\al}d\al \\
& = \bigl(60\ell t_\ell\bigr)\cdot\int_0^\infty\pr{\topl(Y^{\bd,\ell})>60\scal\ell t_\ell}d\scal
\leq \paren{\int_{0}^{\infty} q(\scal) d \scal} \bigo{\ell t_\ell}.
\end{split}
\end{equation*}
The second equality follows by the change of variable, $\al=60\scal\ell t_\ell$. 
We finish the proof by arguing that $\int_0^\infty q(\scal)d\scal=\bigo{\log m/\log \log m}$.

By Chernoff bounds (see Lemma~\ref{chernoff}), for any machine $i$ and any $\scal\geq e^2$ we have:
\[
\pr{ Y^{\bd,\ell}_i > \scal \paren{16 \, \beload[i,\ell] + 44} t_\ell \, } \leq 
\paren{\frac{e^{\scal-1}}{\scal^\scal}}^{16 \, \beload[i,\ell] + 44} 
\leq e^{-\frac{1}{2}\scal \ln \scal(16 \, \beload[i,\ell] + 44)}
\leq e^{-22\,\scal\ln\scal}\,.
\]
For $\scal=\al\ln m/\ln\ln m$, where $\al\geq 1$, we therefore obtain that 
$\Pr\bigl[Y^{\bd,\ell}_i>\scal(16\,\beload[i,\ell]+44)\bigr]$ is at most $e^{-11\al\ln m}$
(since $\scal\ln\scal\geq\al\ln m/2$).
By the union bound, it follows that $q(\al\ln m/\ln\ln m)\leq e^{-10\al\ln m}$ for all
$\al\geq 1$. 
Thus,
\begin{equation*}
\begin{split}
\int_0^{\infty} q(\scal) d \scal & 
\leq \frac{\ln m}{\ln \ln m}
+\int_{\frac{\ln m}{\ln\ln m}}^{\infty} q(\scal) d \scal  
= \frac{\ln m}{\ln \ln m}\biggl[1+\int_1^\infty q(\al\ln m/\ln\ln m)d\al\biggr] \\
& \leq \frac{\ln m}{\ln \ln m}\biggl[1+\int_1^\infty e^{-10\al\ln m}d\al\biggr]  
= \frac{\ln m}{\ln \ln m}+\frac{m^{-10}}{10\ln \ln m}
= \bigo{\log m/\log \log m} \ . \qedhere
\end{split}
\end{equation*}
\end{proof}

\begin{proof}[{\bf Finishing up the proof of Theorem~\ref{thm:loadbalgen}}]
We finish up the proof by finding a suitable $\vt$ vector, as in the proof of
Theorem~\ref{thm:loadbalber}, and returning the solution found for this vector. 
As before, for each $\vt \in \T$ (see \eqref{setofguesses} for the definition of $\T$),
and $\ell \in \POS$, define $\budg_\ell(\vt\,) := \ell t_\ell$. 
Let $\bvec(\vt\,) \in \Rp^m$ be the vector given by Lemma~\ref{bvecest} for 
the sequence $\bigl(\budg_\ell(\vt\,)\bigr)_{\ell \in \POS}$.
Among all feasible $\vt$ vectors, let $\vtb$ be a vector that minimizes
$f\bigl(\bvec(\vt\,)\bigr)$; recall that 
$f\bigl(\bvec(\vtb\,)\bigr) = \bigo{\OPT}$. 
Let $\overline{\sg}$ be the assignment computed by our algorithm by rounding a feasible
solution to $\LPv[\vtb]$. 
By Theorem~\ref{roundinglogmloglogm}, for each $\ell \in \POS$, $\overline{\sg}$ satisfies
$\E{\Topl{\load_{\overline{\sg}}}} = \bigo{\log m/\log \log m} \budg_\ell(\vtb\,)$. 
Hence, by Corollary~\ref{detestimate} (a), we have 
$\E{f(\load_{\overline{\sg}})} = \bigo{\frac{\log m}{\log \log m}}\cdot f\paren{\bvec(\vtb\,)} 
= \bigo{\frac{\log m}{\log \log m}}\cdot\OPT$. 
\end{proof}

\section{Spanning Trees} \label{sec:tree}

In this section, we apply our framework to devise an approximation algorithm for \emph{stochastic minimum norm spanning tree}.
We are given an undirected graph $G = (V,E)$ with stochastic edge-weights and we are interested in connecting the vertices of this graph to each other.
For an edge $e \in E$, we have a nonnegative r.v. $X_e$ that denotes its weight.
Edge weights are independent.
A feasible solution is a spanning tree $T \subseteq E$ of $G$.
This induces a random weight vector $\treerv^T = (X_e)_{e \in T}$; 
note that $\treerv^T$ follows a product distribution on $\Rp^{\n}$ where $n := |V|$.
The goal is to find a spanning tree $T$ that minimizes $\E{f(\treerv^T)}$ for a given
monotone symmetric norm $f$. 
(As an aside, we note that the deterministic version of this problem, wherein each edge
$e$ has a fixed weight $w_e\geq 0$, can be solved optimally: an MST $T'$ 
simultaneously minimizes $\sum_{e \in T} \paren{w_e-\scal}^+$ among all spanning trees $T$
of $G$, for all $\scal\geq 0$.
Since $\Topl{x} = \min_{\scal \in \Rp} \curly{\ell \scal + \sum_i \paren{x_i-\scal}^+}$,
it follows that $T'$ minimizes every $\topl$-norm, and so by Theorem~\ref{monnormthm}, it
is simultaneously optimal for all monotone, symmetric norms.) 

\begin{theorem} \label{thm:tree}
There is an $\bigo{1}$-approximation algorithm for stochastic $f$-norm spanning tree with arbitrary edge-weight distributions.
\end{theorem}

\paragraph{Overview.}
We drop the superscript $T$ in $\treerv^T$ when the tree $T$ is clear from the context.
We use $\treerv_e$ to denote the coordinate in $\treerv$ corresponding to edge $e$;
we will ensure that $\treerv_e$ is used only when $e \in T$.
Let $T^*$ denote an optimal solution. 
Let $\treerv^* := \treerv^{T^*}$ and $\OPT := \E{f(\treerv^*)}$ denote the optimal value.

The cost vector $\treerv^T$ is inherently less complex than the load vector in load
balancing, in that each coordinate $\treerv^T_e$ is an ``atomic'' random variable whose
distribution we can directly access. Thus, our approach is guided by
Lemma~\ref{tauexpr} and Theorem~\ref{apxstatsthm}, which show that 
to obtain an approximation guarantee for stochastic $f$-norm spanning tree, it suffices to
find a spanning tree $T$ such that the $\taul$ statistics of $\treerv^T$ are
``comparable'' to those of $\treerv^*$. 

We write an LP relaxation \ref{treelp} that works with a non-increasing vector $\vt \in
\Rp^{\POS_{\n}}$ that is intended to be a guess of 
$\bigl(\problth(\treerv^*)\bigr)_{\ell\in\POS}$.
By Lemma~\ref{tauexpr}, 
$\E{f(\treerv)} = \tht{\sum_e\E{\paren{\treerv_e-\taul[1](Y)}^+} + f(\vec{\taul[]}(Y))}$.  
So our LP seeks a (fractional) spanning tree whose cost vector $\treerv$ 
minimizes $\sum_e\E{(X_e-t_1)}^+$ subject to the constraint that
$\taul(\treerv) \leq t_\ell$ for all $\ell \in \POS$.
We round an optimal solution to this LP using Theorem~\ref{iterrndthm}, and argue that
this solution has expected cost at most $O\bigl(\LPOPTv + f(\vtp)\bigr)$,
where $\LPOPTv$ is the optimal value of \ref{treelp} and $\vtp\in\Rp^{n-1}$ is the expansion of $\vt$. 
(Recall that $\vtp$ is defined by setting $\vtp_i := \vt_{2^{\floor{\log_2 i}}}$ for all
$i \in [\n]$.)  
Finally, we show that we can find a suitable vector $\vt$ for which 
$\bigl(\LPOPTv+f(\vtp)\bigr)$ is $O(\OPT)$.

\subsection{LP relaxation} \label{sec:LPtree}
Let $\vt \in \Rp^{\POS_{\n}}$ be a non-increasing vector.
We have $z_e$ variables indicating if an edge $e$ belongs to the spanning tree. So $z$
belongs to the polytope 
$$
\Q^\mst :=\ \bigl\{z \in \Rp^E: \quad z(E) = |V| - 1, \qquad 
z(A) \leq n -\comp(A) \quad \forall A \subseteq E\bigr\},
$$
where $\comp(A)$ denotes the number of connected
components of $(V,A)$, and $z(A)$ denotes $\sum_{e \in A} z_e$.  
It is well-known that the above polytope is the spanning-tree polytope of $G$, i.e., it is
the convex hull of indicator vectors of spanning trees of $G$.
We use the matroid base-polytope characterization for two reasons: (i) our arguments can
be generalized verbatim to the more general setting with an arbitrary
matroid; 
and (ii) we can directly invoke Theorem~\ref{iterrndthm} to round fractional LP solutions
to integral solutions. 
The constraints of our LP encode that $z\in\Q^\mst$, and that the $\taul$-statistics of the
cost vector induced by $z$ are bounded by $\vt$.
The objective function accounts for the contribution from the coordinates larger than
$\taul[1]$ to the expected $f$-norm. 
\begin{alignat}{2}
\min & \quad & \sum_{e \in E} \E{\paren{X_e - t_1}^+} z_e &
\tag*{\LPTv} \label{treelp} \\
\text{s.t.} && \sum_{e \in E} \pr{X_e > t_\ell} z_e & \leq \ell \qquad \forall \ell \in \POS \label{eq:taul} \\
&& z \in \Q^\mst. & \label{eq:fractree}
\end{alignat}
The following claim is immediate. 

\begin{claim} \label{clm:LPfeastree}
$\LPTv$ is feasible whenever $t_\ell \geq \Taul{\treerv^*}$ for all $\ell \in \POS$.
\end{claim}

\begin{proof}
Consider the solution $\sz$ induced by the optimal solution $T^*$ i.e., $\sz$ is the indicator vector of $T^*$.
For each $\ell \in \POS$ constraint~\eqref{eq:taul} is satisfied since $t_\ell \geq \taul$.
\end{proof}

\subsection{Rounding algorithm and analysis} \label{sec:roundingtree}
Assume that $\vt \in \Rp^\POS$ is a non-increasing vector such that $\LPTv$ is feasible.
Let $\bz$ be an optimal fractional spanning tree solution to $\LPTv$, and let $\LPOPTv$ be
the optimal value of $\LPTv$.
Define $\val(\vt) := \LPOPTv + f(\vtp)$, where $\vtp \in\Rp^{\n}$ denotes the expansion of
$\vt$. 
We show how to round $\bz$ to obtain a spanning tree $T$ satisfying 
$\E{f(\treerv)}=\bigo{ \val(\vt) }$. 
The expression for $\val(\vt)$ is motivated by Lemma~\ref{tauexpr}, which shows that 
$\E{f(\treerv)} = \tht{ \sum_e \E{\paren{\treerv_e-\taul[1]}^+} + f(\vec{\taul[]})}$.

Note that scaling constraint \eqref{eq:taul} by $\ell$ yields $O(1)$-bounded column sums
in the resulting constraint matrix, and an additive $O(1)$ violation of the scaled constraint
translates to an additive $O(\ell)$ violation of constraint~\eqref{eq:taul} for index $\ell$.

We invoke Theorem~\ref{iterrndthm} with the parameter $\viol := \sum_{\ell \in \POS}
\frac{1}{\ell}<2$ to round $\bz$ to an integral spanning tree $T$. 

\begin{theorem} \label{thm:boundedtree}
The weight vector $\treerv=\treerv^T$ satisfies:
\begin{enumerate}[(i), topsep=0.25ex, itemsep=0.25ex, leftmargin=*]
\item $\sum_{e \in T} \E{\paren{\treerv_e - t_1}^+} \leq \LPOPTv$; 
\item for each $\ell \in \POS$, $\Taul[3\ell]{\treerv} \leq t_\ell$; and
\item $\E{f(\treerv)}\leq 126\cdot\val(\vt)$.
\end{enumerate}
\end{theorem}

\begin{proof}
Parts (i) and (ii) follow from parts (a) and (b) of Theorem~\ref{iterrndthm} respectively.
Part (i) is immediate from part (a).
For part (ii), since we invoke Theorem~\ref{iterrndthm}
with $\viol < 2$, by part (b), we have that \linebreak
$\sum_{e \in T} \pr{\treerv_e > t_\ell} < 3 \ell$, and hence,
$\Taul[3\ell]{\treerv} \leq t_\ell$. 
For part (iii), observe that $\treerv, \vt$ satisfy the assumptions of
Theorem~\ref{apxstatsthm}(a) with $\alpha = 1$ and $\beta = 3$, so we obtain that
$\E{f(\treerv)} \leq 126 \, \bigl( \sum_{e \in T} \E{\paren{\treerv_e - t_1}^+}+f(\vtp)\bigr)$, 
and the RHS expression is at most $126\cdot\val(\vt)$ by part (i).
\end{proof}

\subsubsection*{Finishing up the proof of Theorem~\ref{thm:tree}}
For $\ell \in \POS$, let $t^*_\ell$ be the largest number of the form $2^k$, where
$k \in \Z$ (possibly negative) such that $t^*_\ell \leq 2 \Taul{\treerv^*} + \delta$ and
$\delta > 0$ is a small constant that we will set later. 
Given Claim~\ref{clm:LPfeastree} and Theorem~\ref{thm:boundedtree}, we know that the
solution $T$ computed by rounding an optimal fractional solution $\sz$ to $\LPTv[\vts]$
satisfies \linebreak $\E{f(\treerv^T)} = \bigo{ \val(\vts) }$. 
We argue below that we can identify a $\poly(m)$-sized set $\T \subseteq \Rp^\POS$
containing $\vts$.
While we cannot necessarily identify $\vts$, we use Theorem~\ref{apxstatsthm} (b) to infer
that if $\LPTv$ is feasible, then 
$\val(\vt)$ acts as a proxy for the expected $f$-norm of the solution computed for $\vt$.  
This will imply that the ``best'' vector $\vt \in \T$ yields an $\bigo{1}$-approximate solution.

As in the case of load balancing, Claim~\ref{sandwich} immediately yields lower and
upper bounds on $\OPT$
that differ by a multiplicative factor of $\n$. 

\begin{lemma} \label{binsearchtree}
Define $\UB$ to be the weight of a minimum-weight spanning tree in $G$ with edge-weights
given by (the deterministic quantity) $\E{X_e}$.
Then, $\frac{\UB}{\n} \leq \E{\Topl[1]{\treerv^*}} \leq \E{f(\treerv^*)} \leq \UB$.
\end{lemma}

\begin{claim} \label{smalltau1}
$\Taul[1]{\treerv^*} \leq 4 \, \UB$.
\end{claim}
\begin{proof}
Using Lemma~\ref{eproxlexpr} and Theorem~\ref{lem:gammaproxytopl} we get: $\Taul[1]{\treerv^*} \leq \Gaml[1]{\treerv^*} \leq 4 \, \E{\Topl[1]{\treerv^*}} \leq 4 \, \UB$.
\end{proof}

Define $\delta := \UB/n^2$.
Let 
\begin{equation*} \label{setofguessestree}
\T := \curly{ \vt \in \Rp^\POS:\ 
\forall \ell \in \POS,\quad\delta/2 < t_\ell \leq 8 \, \UB + \delta, \ \ 
t_\ell \text{ is a power of $2$} }
\end{equation*}

\begin{claim} \label{clm:needleinapolyhaystacktree}
The vector $\vts$ belongs to $\T$, and $|\T|=\poly(m)$.
\end{claim}
\begin{proof}
The polynomial bound follows from Claim~\ref{monseq} since each $\vt \in \T$ is a
non-increasing vector, and $\log_2 \vt_\ell$ can take only $\bigo{\log m}$ values.
To see that $\vts \in \T$, by definition, $t^*_\ell$ is the highest power of $2$ satisfying $t^*_\ell \leq 2 \Taul{\treerv^*} + \delta$.
Since $t^*_\ell$ is the largest such power of $2$, we get $t^*_\ell > \delta/2$.
And by Claim~\ref{smalltau1}, we get that $t^*_\ell \leq 8 \, \UB + \delta$.
Thus, $\vts \in \T$.
\end{proof}

We enumerate over all guess vectors $\vt \in \T$ and check if $\LPTv$ is feasible.
Among all feasible $\vt$ vectors, let $\vtb$ be a vector that minimizes $\val(\vt)$;
recall that $\val(\vt) = \LPOPTv + f(\vtp)$, where $\LPOPTv$ is the optimal value of
$\LPTv$ and $\vtp\in\Rp^{\n}$ is the expansion of $\vt$. 
Let $\overline{T}$ be the spanning tree computed by our algorithm by rounding an optimal
fractional solution to $\LPv[\vtb]$. 
By Theorem~\ref{thm:boundedtree} (iii), we have 
$\E{f(\treerv^{\overline{T}})} = \bigo{\val(\vtb)}$, and by
Claim~\ref{clm:needleinapolyhaystacktree}, we know that $\val(\vtb)\leq\val(\vts)$.

To finish the proof, we argue that
$\val(\vts) = \bigo{\OPT}$.
Recall that 
for each $\ell \in \POS$, $t^*_\ell$ is the highest power of $2$ satisfying 
$t^*_\ell \leq 2 \Taul{\treerv^*} + \delta$. It follows that  
$\Taul{\treerv^*}\leq t^*_\ell\leq 2 \Taul{\treerv^*} + \delta$ for all $\ell \in \POS$.
Thus, by Claim~\ref{clm:LPfeastree}, $\LPTv[\vts]$ is feasible; in particular, the
indicator vector of $T^*$ is a feasible solution to $\LPTv[\vts]$.
It follows that 
\[
\val(\vts)\leq \sum_{e \in T^*} \E{\paren{\treerv^*_e - t^*_1}^+}+f\bigl(\vec{t^*}'\bigr)
\leq 32\cdot\E{f(\treerv^*)}+(\n)\delta=O(\OPT),
\]
where the second inequality is due to Theorem~\ref{apxstatsthm} (b).
\hfill \qed

\subsection{Extension: stochastic minimum-norm matroid basis} \label{matbasis}
Our results extend quite seamlessly to the {\em stochastic minimum-norm matroid basis}
problem, which is the generalization of stochastic minimum-norm spanning tree, where we
replace spanning trees by bases of an arbitrary given matroid.  
More precisely, 
the setup is that we are given a
matroid $\M=(U,\indset)$ specified via its rank function $\rank$, and a monotone symmetric
norm $f$ on $\R^{\rank(U)}$. Each element $e\in U$ has a random weight $X_e$. The goal is
to find a basis of $M$ whose induced random weight vector has minimum expected $f$-norm.

The only change in the algorithm is that we replace the spanning-tree polytope $\Q^\mst$ in
\eqref{eq:fractree} by the base polytope of $\M$, and we of course invoke
Theorem~\ref{iterrndthm} with the same base polytope. The analysis is {\em identical} to
the spanning-tree setting.

\section{Proof of Lemma~\ref{lem:largebetalargetail}} \label{append-largebeta}
We restate the lemma below for convenient reference.

\begin{duplicate}[Lemma~\ref{lem:largebetalargetail} (restated)]
Let $\sumrv = \sum_{j \in [k]} \scrv_j$, where the $\scrv_j$s are independent
$[0,\qt]$-bounded random variables. Let $\ld\geq 1$ be an integer. 
Then, 
\[
\E{\sumrv^{\geq \qt}} \geq \qt \cdot \frac{\sum_{j \in [k]}
\beta_\lambda\paren{\scrv_j/4\qt} - 6}{4 \ld} 
\]
\end{duplicate}

\medskip
It suffices to prove the result for $\qt=1$. This is because if we set
$\scrv'_j=\scrv_j/\qt$ for all $j\in[k]$, then the $\scrv'_j$s are $[0,1]$-bounded random
variables, and $\sumrv'=\sumrv/\qt$; so applying the 
result for $\qt=1$ yields the inequality
$\E{\sumrv'^{\geq 1}}\geq\frac{\sum_{j\in[k]}\beta_\ld(\scrv'/4)-6}{4\ld}$, and we have 
$\E{\sumrv'^{\geq 1}}=\E{\sumrv^{\geq\qt}}/\qt$.
So we work with $\qt=1$ from now on.

To keep notation simple, we always reserve $j$ to index over the set $[k]$.
We will actually prove the following ``volume'' inequality:
\begin{equation}
\sum_{j \in [k]} \eff{\scrv_j/4} \leq \paren{3\ld + 2} \Eb{\erv{\sumrv}{1}} + 6 \ ,
\label{volineq}
\end{equation}
holds for all integer $\ld \geq 1$.
The above inequality holds trivially for $\ld = 1$ since $\sum_j \eff{\scrv_j/4} = \frac14 \,
\E{\sumrv} \leq \frac14 \, \paren{\E{ \erv{\sumrv}{1} } + 1}$.  
(In fact, the inequality, and the lemma is very weak for $\ld = 1$.) 
In the sequel, $\ld \geq 2$, 
and note that for $3\ld+2\leq 4\ld$ for $\ld\geq 2$, so \eqref{volineq} implies the lemma.

At a high level, we combine and adapt the proofs of Lemmas 3.2 and 3.4
from~\cite{KleinbergRT00} to obtain our result. 
We say that a Bernoulli trial is of type $(q,s)$ if it takes size $s$ with probability $q$,
and size $0$ with probability $1-q$. 
For a Bernoulli trial $B$ of type $(q,s)$, Kleinberg et al.~\cite{KleinbergRT00} define a
modified notion of effective size: $\effm{B} := \min(s,sq\ld^s)$. 
The following claim will be useful.

\begin{claim}[Proposition 2.5 from \cite{KleinbergRT00}] \label{modeff}
$\eff{B} \leq \effm{B}$.
\end{claim}

Roughly speaking, inequality \eqref{volineq} states that each unit of
$\ld$-effective size contributes ``$\ld^{-1}$-units'' towards $\E{\sumrv^{\geq 1}}$. This is
indeed what we show, but we first reduce to the setting of Bernoulli trials.

The proof will involve various transformations of the $\scrv_j$ random variables, and the
notion of stochastic dominance will be convenient to compare the various random
variables so obtained. A random variable $R$ {\em stochastically dominates}
another random variable $B$, denoted $B\sdleq R$, if $\pr{R\geq t}\geq\pr{B\geq t}$ for all
$t\in\R$. We will use the following well-known facts about stochastic dominance.
\begin{enumerate}[label=(F\arabic*), topsep=0.25ex, itemsep=0ex, leftmargin=*]
\item \label{sdmon}
If $B\sdleq R$, then for any non-decreasing function $u:\R\mapsto\R$, we have
$\E{u(B)}\leq\E{u(R)}$. 
\item \label{sdsum}
If $B_i\sdleq R_i$ for all $i=1,\ldots,k$, then 
$\bigl(\sum_{i=1}^k B_i\bigr)\sdleq\bigl(\sum_{i=1}^k R_i\bigr)$.
\end{enumerate}
All the random variables encountered in the proof will be nonnegative, and we will often omit
stating this explicitly.

\subsection{Bernoulli decomposition: reducing to Bernoulli trials}
We utilize a result of~\cite{KleinbergRT00} that shows how to replace an arbitrary
random variable $R$ with a sum of Bernoulli trials that is ``close'' to $R$ in terms of
stochastic dominance. This allows us to reduce to the case where all 
random variables are Bernoulli trials (Lemma~\ref{berredn}). 

We use $\supp(R)$ to denote the support of a random variable $R$; if $R$ is a
discrete random variable, this is the set $\{x\in\R: \pr{R=x}>0\}$.
Following~\cite{KleinbergRT00}, we say that a random variable is {\em geometric} if its
support is contained in $\{0\}\cup\{2^r: r\in\Z\}$.

\begin{lemma}[Lemma~3.10 from \cite{KleinbergRT00}] \label{decomp}
Let $R$ be a geometric random variable. 
Then there exists a set of independent Bernoulli trials $B_1,\dots,B_p$ such that 
$B = B_1+ \dots + B_p$ satisfies 
$\pr{R = t} = \pr{t \leq B < 2t}$ for all $t\in\supp(R)$. 
Furthermore, the support of each $B_i$ is contained in the support of $R$. 
\end{lemma}

\begin{corollary} \label{decompcor}
Let $R$ be a geometric random variable, and $B$ be the corresponding sum of Bernoulli
trials given by Lemma~\ref{decomp}. Then, we have $B/2\sdleq R\sdleq B$.
\end{corollary} 

\begin{proof}
To see that $R\sdleq B$, consider any $t\geq 0$. We have 
$$
\pr{R\geq t}=\sum_{t'\in\supp(R): t'\geq t}\pr{R=t'}
=\sum_{t'\in\supp(R): t'\geq t}\pr{t'\leq B<2t'}\leq\pr{B\geq t}
$$
where the second equality is due to Lemma~\ref{decomp}, and the
last inequality follows since the intervals $[t',2t')$ are disjoint for
$t'\in\supp(R)$, as $R$ is a geometric random variable.

To show that $B/2\sdleq R$, we argue that $\pr{R<t}\leq\pr{B<2t}$ for all $t\in\R$.
This follows from a very similar argument as above. We have
$$
\pr{R<t}=\sum_{t'\in\supp(R): t'<t}\pr{R=t'}
=\sum_{t'\in\supp(R): t'<t}\pr{t'\leq B<2t'}\leq\pr{B<2t}.
$$
The last inequality again follows because the intervals $[t',2t')$ are disjoint for
$t'\in\supp(R)$, as $R$ is a geometric random variable.
\end{proof}

We intend to apply Lemma~\ref{decomp} to the $\scrv_j$s to obtain a collection of
Bernoulli trials, but first we need to convert them to geometric random variables.
For technical reasons that will be clear soon, 
we first scale our random variables by a factor of $4$; then, we convert each scaled random
variable to a geometric random variable by rounding up the non-zero values in its support 
to the closest power of $2$.
Formally, for each $j$, let $\scrv^\rd_j$ denote the {geometric} random variable obtained
by \emph{rounding up} each non-zero value in the support of $\scrv_j/4$ to the
closest power of $2$.
(So, for instance, $0.22$ would get rounded up to $0.25$, whereas $1/8$ would stay the
same.) 
Note that $\scrv^\rd_j$ is a geometric $[0,1/4]$-bounded random variable. 

We now apply Lemma~\ref{decomp} to each $\scrv^\rd_j$ to obtain a collection 
$\{ \scrv^\ber_{\ell} \}_{\ell\in F_j}$ of independent Bernoulli trials.
Let $F$ be the disjoint union of the $F_j$s, so $F$ is the collection of all the Bernoulli
variables so obtained. 
Define $\sumrv^\ber := \sum_{\ell\in F} \scrv^\ber_{\ell}$. 
For $\ell\in F$, let $\scrv^\ber_{\ell}$ be a Bernoulli trial of type $(q_{\ell},s_{\ell})$; 
from Lemma~\ref{decomp}, we have that $s_{\ell} \in [0,1/4]$ and is a (inverse) power of $2$.
We now argue that it suffices to prove \eqref{volineq} for the
Bernoulli trials $\{\scrv^\ber_\ell\}_{\ell\in F}$.

\begin{lemma} \label{berredn}
Inequality \eqref{volineq} follows from 
the inequality: 
$\sum_{\ell\in F}\eff{\scrv^\ber_\ell}\leq (3\ld+2)\E{(\sumrv^\ber)^{\geq 1}}+6$. 
By Claim~\ref{modeff}, this in turn is implied by the inequality
\begin{equation}
\sum_{\ell\in F}\effm{\scrv^\ber_\ell}\leq (3\ld+2)\E{(\sumrv^\ber)^{\geq 1}}+6.
\label{berineq}
\end{equation}
\end{lemma}

\begin{proof}
Fix any $j\in[k]$.
By Corollary~\ref{decompcor}, we have $\scrv^\rd_j\sdleq\sum_{\ell\in F_j}\scrv^\ber_\ell$. 
We also have that $\scrv_j/4\leq\scrv^\rd_j$, and so 
$\scrv_j/4\sdleq\sum_{\ell\in F_j}\scrv^\ber_\ell$. 
Therefore, since $\ld^x$ is an increasing function of $x$, we have 
$\eff{\scrv_j/4}\leq\eff{\sum_{\ell\in F_j}\scrv^\ber_\ell}=\sum_{\ell\in F_j}\eff{\scrv^\ber_\ell}$. 
Summing over all $j$, we obtain that 
$\sum_{j\in[k]}\eff{\scrv_j/4}\leq\sum_{\ell\in F}\eff{\scrv^\ber_\ell}$. 

Corollary~\ref{decompcor} also yields that 
$\sum_{\ell\in F_j}\scrv^\ber_\ell\sdleq 2\cdot\scrv^\rd_j$ for all $j$. 
We also have $\scrv^\rd_j\leq 2(\scrv_j/4)$, and therefore, we have
$\sum_{\ell\in F_j}\scrv^\ber_\ell\sdleq\scrv_j$ for all $j$.
Using Fact~\ref{sdsum}, this implies that $\sumrv^\ber\sdleq\sumrv$, and hence, by
Fact~\ref{sdmon}, we have $\E{(\sumrv^\ber)^{\geq 1}}\leq\E{\sumrv^{\geq 1}}$.

To summarize, we have shown that $\sum_{j\in[k]}\eff{\scrv_j/4}$ is at most 
$\sum_{\ell\in F}\eff{\scrv^\ber_\ell}$, which is at most the LHS of
\eqref{berineq} (by Claim~\ref{modeff}), and 
the RHS of \eqref{volineq} is at least the RHS of \eqref{berineq}. 
\end{proof}

\subsection{Proof of inequality \eqref{berineq}}
We now focus on proving inequality \eqref{berineq}.
Let $F^\sml := \{\ell\in F : \ld^{s_\ell}\leq 2\}$
index the set of \emph{small} Bernoulli trials, and let $F^\lar := F \setminus F^\sml$
index the remaining {\em large} Bernoulli trials.

It is easy to show that the total modified effective size of small Bernoulli trials is at
most $2\E{\sumrv^\ber}\leq 2\E{(\sumrv^\ber)^{\geq 1}}+2$ (see Claim~\ref{clm:smalleffsize} 
and inequality \eqref{smallvars}). Bounding the modified effective size of the large
Bernoulli trials is more involved. Roughly speaking, we first consolidate these random
variables by replacing them with ``simpler'' Bernoulli trials, and then show that
each unit of total modified effective size of these simpler Bernoulli trials makes a
contribution of $\ld^{-1}$ towards $\E{(\sumrv^\ber)^{\geq 1}}$. The constant $6$ in
\eqref{berineq} 
arises due to two reasons:
(i) because we bound the modified effective size of the small Bernoulli trials
by $O\bigl(\E{\sumrv^\ber}\bigr)$ (as opposed to $O\bigl(\E{(\sumrv^\ber)^{\geq 1}}\bigr)$);
and (ii) because we lose some modified effective size in the consolidation 
or large Bernoulli trials.%
\footnote{Note that some constant must {\em unavoidably} appear on the RHS of
inequalities \eqref{volineq} and \eqref{berineq}; that is, we cannot bound the total
effective size (or modified effective size) by a purely multiplicative factor of
$\E{\sumrv^{\geq 1}}$, even for Bernoulli trials. This is because if $\ld=1$, and say we
have only one (Bernoulli) random variable $\scrv$ that is strictly less than $1$, then its
effective size (as also its modified effective size) is simply $\E{\scrv}$, whereas 
$\sumrv^{\geq 1}=\scrv^{\geq 1}=0$.} 

\begin{claim}[Shown in Lemma~3.4 in \cite{KleinbergRT00}] \label{clm:smalleffsize}
Let $B$ be a Bernoulli trial of type $(q,s)$ with $\ld^s\leq 2$. 
Then, $\effm{B} \leq 2 \, \E{B}$.
\end{claim}

\begin{proof}
By definition, $\effm{B} = \Min{s,sq\ld^s} \leq 2 qs = 2 \, \E{B}$. 
\end{proof}

By the above claim we get the following volume inequality for the small Bernoulli trials.

\begin{equation} \tag{vol-small} \label{smallvars}
\sum_{\ell \in F^\sml} \effm{\scrv^\ber_{\ell}} \leq 2 \, \sum_{\ell \in F^\sml}
  \E{\scrv^\ber_{\ell}} \leq 2 \, \E{\sumrv^\ber}\leq 2 \, \E{(\sumrv^\ber)^{\geq 1}}+2\ . 
\end{equation}

We now handle the Bernoulli trials in $F^\lar$.
For each $\ell\in F^\lar$, we set $q'_{\ell} = \min(q_{\ell},\ld^{-s_{\ell}})$. 
Observe that this operation 
does not change the modified effective size, and we have that 
$\effm{(q_\ell,s_\ell)}=\effm{(q'_\ell,s_\ell)}=s_\ell q'_\ell\ld^{s_\ell}$.
The following claim from \cite{KleinbergRT00} is useful in consolidating Bernoulli trials
of the same size. 

\begin{claim}[Claim~3.1 from \cite{KleinbergRT00}] \label{grouping}
Let $\mce_1,\dots,\mce_p$ be independent events, with $\pr{\mce_i} = p_i$.
Let $\mce'$ be the event that at least one of these events occurs. 
Then $\pr{\mce'} \geq \frac{1}{2}\min(1,\sum_i p_i)$.
\end{claim}

\paragraph{Consolidation for a fixed size.}
For each $s$ that is an inverse power of 2, we define 
$F^s := \{\ell\in F^\lar : s_{\ell}=s \}$, so that $\bigcup_{s} F^s$ is a partition of
$F^\lar$. 
Next, we further partition $F^s$ into sets $P^s_1,\dots,P^s_{n_s}$ such that for all 
$i =1,\dots,n_s-1$, we have $2 \ld^{-s} \leq \sum_{\ell \in P^s_i} q'_{\ell} < 3 \ld^{-s}$ and
$\sum_{\ell\in P^s_{n_s}} q'_{\ell} < 2 \ld^{-s}$.  
Such a partitioning always exists since for each $\ell \in F^s$ we have
$q'_{\ell}\leq\lambda^{-s}]$ by definition. 
We now apply Claim~\ref{grouping} to ``consolidate'' each $P^s_i$: by this, we mean that
for $i=1,\ldots,n_s-1$, we think of representing $P^s_i$ by the
``simpler'' Bernoulli trial $B^s_i$ of type $(\ld^{-s},s)$ and using this to replace the
individual random variables in $P^s_i$.

By Claim~\ref{grouping}, for any $i=1,\ldots,n_s-1$, 
we have $\pr{\sum_{\ell\in P^s_i}\scrv^\ber_\ell\geq s}\geq\ld^{-s}=\pr{B^s_i=s}$ (note
that this only works for large Bernoulli trials since $2\ld^{-s}\leq 1$);
hence, it follows that $B^s_i\sdleq\sum_{\ell\in P^s_i}\scrv^\ber_\ell$.

Note that 
$$
\sum_{\ell\in P^s_i}\effm{\scrv^\ber_\ell}=s\ld^s\sum_{\ell\in P^s_i}q'_\ell
<
\begin{cases} 
3s; & \text{if $i\in\{1,\ldots,n_s-1\}$} \\
2s; & \text{if $i=n_s$}. 
\end{cases}
$$
Also, $\effm{B^s_i} = \min(s, s \, \lambda^{-s} \, \lambda^s) = s$.
Putting everything together,
\[
(n_s-1) s = \sum_{i=1}^{n_s-1} \effm{B^s_i} \geq \frac{\sum_{\ell \in F^s} \effm{\scrv^\ber_{\ell}} - 2s}{3} \ .
\]

\paragraph{Consolidation across different sizes.}
Summing the above inequality for all $s$ (note that each $s$ is an inverse power of $2$
and at most $1/4$), we obtain that 
\begin{equation}
\sum_s \sum_{i=1}^{n_s-1} \effm{B^s_i} \geq \frac{ \sum_{\ell \in F^\lar} \effm{\scrv^\ber_{\ell}} - 1}{3}.
\label{berlrgineq}
\end{equation}
Let $\vol$ denote the RHS of \eqref{berlrgineq}, and 
let $m := \floor{\vol}$. 
(Note that $m$ could be $0$; but if $m=0$, then $\vol<1$, and \eqref{berineq} trivially holds.)
Since each $B^s_i$ is a Bernoulli trial of type $(\ld^{-s},s)$, where $s$ is an inverse
power of $2$, 
we can obtain $m$ disjoint subsets $A_1,\dots,A_m$ of $(s,i)$ pairs from the entire
collection $\{B^s_i\}_{s,i}$ of Bernoulli trials, such that 
$\sum_{(s,i) \in A_u} \effm{B^s_i} = \sum_{(s,i) \in A_u} s = 1$ for each $u\in [m]$.%
\footnote{To justify this statement, it suffices to show the following. Suppose we have
created some $r$ sets $A_1,\ldots,A_r$, where $r<m$, and let $I$ be the set of $(s,i)$
pairs 
indexing the Bernoulli trials that are not in $A_1,\ldots,A_r$; then, we can find a subset 
$I'\sse I$ such that $\sum_{(s,i)\in I'}s=1$.  
To see this, first since $r<m$, we have $\sum_{(s,i)\in I}s\geq 1$. We sort the $(s,i)$
pairs in $I$ in non-increasing order of $s$; to avoid excessive notation, let $I$ denote
this sorted list. Now since each $s$ is an inverse power of $2$, it is easy to see by
induction that if $J$ is a prefix of $I$ such that $\sum_{(s,i)\in J}s<1$, then
$1-\sum_{(s,i)\in J}s$ is at least as large as the $s$-value of the pair in $I$
appearing immediately after $J$. Coupled with the fact that $\sum_{(s,i)\in I}s\geq 1$,
this implies that there is a prefix $I'$ such that $\sum_{(s,i)\in I'}s=1$.}
For each subset $A_u$, 
\[
{\boldsymbol\Pr}\biggl[\sum_{(s,i) \in A_u} B^s_i = 1\biggr] 
= \prod_{(s,i) \in A_u} \pr{B^s_i = s} = \prod_{(s,i) \in A_u} \ld^{-s} = \ld^{-1} \ .
\]

\paragraph{Finishing up the proof of inequality \eqref{berineq}.}
For any nonnegative random variables $R_1,R_2$, 
we have {$\E{\erv{\paren{R_1+R_2}}{1}} \geq \E{\erv{R_1}{1}} + \E{\erv{R_2}{1}}$}.
So, 
\[
{\boldsymbol\Exp}\biggl[\erv{\Bigl(\sum_s\sum_{i=1}^{n_s-1} B^s_i\Bigr)}{1}\biggr] 
\geq \sum_{u = 1}^{m} {\boldsymbol\Exp}\biggl[\erv{\Bigl(\sum_{(s,i) \in A_u} B^s_i\Bigr)}{1}\biggr] 
= \frac{m}{\ld} \geq \frac{ \sum_{\ell \in F^\lar} \effm{\scrv^\ber_{\ell}} - 4}{3 \ld} \ .
\]
As noted earlier, we have that $B^s_i\sdleq\sum_{\ell\in P^s_i}\scrv^\ber_\ell$ for all
$s$, and all $i=1,\ldots,n_s-1$.
By Fact~\ref{sdsum}, it follows that 
$\bigl(\sum_s\sum_{i=1}^{n_s-1}B^s_i\bigr)
\sdleq\bigl(\sum_s\sum_{i=1}^{n_s-1}\sum_{\ell\in P^s_i}\scrv^\ber_\ell\bigr)$.
Also,
$$
\sum_s\sum_{i=1}^{n_s-1}\sum_{\ell\in P^s_i}\scrv^\ber_\ell
\leq\sum_{\ell\in F^\lar}\scrv^\ber_\ell\leq\sum_{\ell\in F}\scrv^\ber_\ell=\sumrv^\ber, 
$$
and combining the above with Fact~\ref{sdmon}, we obtain that  
$\E{\bigl(\sum_s\sum_{i=1}^{n_s-1}B^s_i\bigr)^{\geq 1}}\leq\E{(\sumrv^\ber)^{\geq 1}}$.
Thus, we have shown that 
\begin{equation} \tag{vol-large} \label{largevars}
\sum_{\ell \in F^\lar} \effm{\scrv^\ber_{\ell}} \leq 3 \ld \E{\erv{(\sumrv^\ber)}{1}} + 4 \ .
\end{equation}

Adding \eqref{smallvars} and \eqref{largevars} gives
\[
\sum_{\ell \in F} \effm{\scrv^\ber_{\ell}} \leq \paren{3\ld + 2} \E{\erv{(\sumrv^\ber)}{1}} + 6.
\]
This completes the proof of inequality \eqref{berineq}, and hence the lemma. 
\hfill \qed

\section{Conclusions and discussion} \label{concl}
We introduce stochastic minimum-norm optimization, and present a framework for designing
algorithms for stochastic minimum-norm optimization problems. A key
component of our framework is a structural result showing that 
if $f$ is a monotone symmetric norm, and $Y\in\Rp^m$ is a nonnegative random vector
with independent coordinates, then $\E{f(Y)}=\Theta\bigl(f(\E{Y^\down})\bigr)$;
in particular, this shows that $\E{f(Y)}$ can be controlled by controlling $\E{\topl(Y)}$
for all $\ell\in[m]$ (or all $\ell\in\{1,2,4,\ldots,2^{\floor{\log_2 m}}\}$). 
En route to proving this result, we develop various deterministic proxies to reason about
expected $\topl$-norms, which also yield a deterministic proxy for $\E{f(Y)}$.
We utilize our framework to develop approximation algorithms for stochastic minimum-norm
load balancing on unrelated machines and stochastic minimum-norm spanning tree (and
stochastic minimum-norm matroid basis). 
We obtain $O(1)$-approximation algorithms for spanning tree, and load balancing with 
(i) arbitrary monotone symmetric norms and Bernoulli job sizes, and 
(ii) $\topl$ norms and arbitrary job-size distributions.

The most pressing question left open by our work is developing a constant-factor
approximation algorithm for the general case of stochastic min-norm load balancing, where
both the monotone symmetric norm and the job-size distributions are arbitrary; currently,
we only have an $O\bigl(\frac{\log m}{\log\log m}\bigr)$-approximation.

Another interesting question is to obtain significantly-improved approximation factors. As
mentioned in the Introduction, we have nowhere sought to optimize constants, and 
it is indeed possible to identify some places where one could tighten constants. However,
even with such optimizations, the approximation factor that we obtain for load balancing
(in settings (i) and (ii) above) is likely to be in the thousands (this is true even of
prior work on minimizing expected makespan~\cite{GuptaKNS18}), and the approximation
factor for spanning tree is likely to be in the hundreds. It would be very interesting to
obtain substantial improvements in these approximation factors, and in particular, obtain
approximation factors that are close to those known for the deterministic problem. This
will most likely require new insights. 
It would also be interesting to know if the stochastic problems are strictly harder than 
their deterministic counterparts.

In this context, we highlight one particularly mathematically appealing question, namely,
that of proving tight(er) bounds on the ratio $\E{f(Y)}/f\bigl(\E{Y^\down}\bigr)$. 
We feel that this question, which is an {\em analysis} question bereft of any
computational concerns, is a fundamental question about monotone symmetric norms and
independent random variables that is of independent interest. 
We prove an upper bound of $\cexpnorm$ (which can be improved by a factor of roughly $2$),
but we do not know if a much smaller constant upper bound---say, even $2$---is
possible. Of course, any improvement in the upper bound 
would also translate to improved approximation factors. Proving lower bounds on the
above ratio would also be illuminating, especially, if one is looking to establish tight
bounds.

\bibliographystyle{abbrv}
\bibliography{stochminnorm-short}

\begin{thebibliography}{10}

\bibitem{AlamdariS17}
S.~Alamdari and D.~B. Shmoys.
\newblock A bicriteria approximation algorithm for the {$k$}-center and
  {$k$}-median problems.
\newblock In {\em Proceedings of the 15th {WAOA}}, pages 66--75, 2017.

\bibitem{AlaouiCRWJ16}
A.~E. Alaoui, X.~Cheng, A.~Ramdas, M.~J. Wainwright, and M.~I. Jordan.
\newblock Asymptotic behavior of {$\ell_p$}-based {L}aplacian regularization in
  semi-supervised learning.
\newblock In {\em Proceedings of the 29th {COLT}}, pages 879--906, 2016.

\bibitem{AouadS19}
A.~Aouad and D.~Segev.
\newblock The ordered {$k$}-median problem: Surrogate models and approximation
  algorithms.
\newblock {\em Mathematical Programming}, 177:55–--83, 2019.

\bibitem{AwerbuchAGKKV95}
B.~Awerbuch, Y.~Azar, E.~F. Grove, M.~Kao, P.~Krishnan, and J.~S. Vitter.
\newblock Load balancing in the {$L_p$} norm.
\newblock In {\em Proceedings of the 36th {FOCS}}, pages 383--391, 1995.

\bibitem{AzarE05}
Y.~Azar and A.~Epstein.
\newblock Convex programming for scheduling unrelated parallel machines.
\newblock In {\em Proceedings of the 37th {STOC}}, pages 331--337, 2005.

\bibitem{ByrkaSS18}
J.~Byrka, K.~Sornat, and J.~Spoerhase.
\newblock Constant-factor approximation for ordered {$k$}-median.
\newblock In {\em Proceedings of the 50th {STOC}}, pages 620--631, 2018.

\bibitem{ChakrabartyS18}
D.~Chakrabarty and C.~Swamy.
\newblock Interpolating between {$k$}-median and {$k$}-center: Approximation
  algorithms for ordered {$k$}-median.
\newblock In {\em Proceedings of the 45th {ICALP}}, pages 29:1--29:14, 2018.

\bibitem{ChakrabartyS19a}
D.~Chakrabarty and C.~Swamy.
\newblock Approximation algorithms for minimum norm and ordered optimization
  problems.
\newblock In {\em Proceedings of the 51st {STOC}}, pages 126--137, 2019.

\bibitem{ChakrabartyS19b}
D.~Chakrabarty and C.~Swamy.
\newblock Simpler and better algorithms for minimum-norm load balancing.
\newblock In {\em Proceedings of the 27th {ESA}}, pages 27:1--27:12, 2019.

\bibitem{GoelI99}
A.~Goel and P.~Indyk.
\newblock Stochastic load balancing and related problems.
\newblock In {\em Proceedings of the 40th {FOCS}}, pages 579--586, 1999.

\bibitem{GuptaKNS18}
A.~Gupta, A.~Kumar, V.~Nagarajan, and X.~Shen.
\newblock Stochastic load balancing on unrelated machines.
\newblock In {\em Proceedings of the 29th {SODA}}, pages 1274--1285, 2018.

\bibitem{GuptaT08}
A.~Gupta and K.~Tangwongsan.
\newblock Simpler analyses of local search algorithms for facility location.
\newblock {\em CoRR}, abs/0809.2554, 2008.

\bibitem{GuptaMUX20}
V.~Gupta, B.~Moseley, M.~Uetz, and Q.~Xie.
\newblock Greed works—online algorithms for unrelated machine stochastic
  scheduling.
\newblock {\em Mathematics of Operations Research}, 45(2):497--516, 2020.

\bibitem{Hui88}
J.~Y. Hui.
\newblock Resource allocation for broadband networks.
\newblock {\em {IEEE} Journal on Selected Areas in Communications},
  6(9):1598--1608, 1988.

\bibitem{ImMP15}
S.~Im, B.~Moseley, and K.~Pruhs.
\newblock Stochastic scheduling of heavy-tailed jobs.
\newblock In {\em Proceedings of the 32nd {STACS}}, pages 474--486, 2015.

\bibitem{KleinbergRT00}
J.~M. Kleinberg, Y.~Rabani, and {\'{E}}.~Tardos.
\newblock Allocating bandwidth for bursty connections.
\newblock {\em {SIAM} J. Comput.}, 30(1):191--217, 2000.

\bibitem{LaporteNG19}
G.~Laporte, S.~Nickel, and F.~S. da~Gama.
\newblock {\em Location science}.
\newblock Springer, 2019.

\bibitem{Latala97}
R.~Lata\l{}a.
\newblock Estimation of moments of sums of independent real random variables.
\newblock {\em The Annals of Probability}, 25:1502--1513, 1997.

\bibitem{LiD19}
J.~Li and A.~Deshpande.
\newblock Maximizing expected utility for stochastic combinatorial optimization
  problems.
\newblock {\em Mathematics of Operations Research}, 44(1):354--375, 2019.

\bibitem{LiL16}
J.~Li and Y.~Liu.
\newblock Approximation algorithms for stochastic combinatorial optimization
  problems.
\newblock {\em Journal of the Operations Research Society of China},
  4(1):1--47, 2016.

\bibitem{LiY13}
J.~Li and W.~Yuan.
\newblock Stochastic combinatorial optimization via {P}oisson approximation.
\newblock In {\em Proceedings of the 45th {STOC}}, pages 971--980, 2013.

\bibitem{LinharesOSZ20}
A.~Linhares, N.~Olver, C.~Swamy, and R.~Zenklusen.
\newblock Approximate multi-matroid intersection via iterative refinement.
\newblock {\em Math. Program.}, 183(1):397--418, 2020.

\bibitem{MakarychevS18}
K.~Makarychev and M.~Sviridenko.
\newblock Solving optimization problems with diseconomies of scale via
  decoupling.
\newblock {\em J. {ACM}}, 65(6):42:1--42:27, 2018.

\bibitem{MoehringSU99}
R.~H. M{\"{o}}hring, A.~S. Schulz, and M.~Uetz.
\newblock Approximation in stochastic scheduling: the power of {LP}-based
  priority policies.
\newblock {\em J. {ACM}}, 46(6):924--942, 1999.

\bibitem{Molinaro19}
M.~Molinaro.
\newblock Stochastic {$\ell_p$} load balancing and moment problems via the
  {$L$}-function method.
\newblock In {\em Proceedings of the 30th {SODA}}, pages 343--354, 2019.

\bibitem{NickelP05}
S.~Nickel and J.~Puerto.
\newblock {\em Location Theory: A Unified Approach}.
\newblock Springer Berlin Heidelberg, 2006.

\bibitem{Pinedo04}
M.~Pinedo.
\newblock Offline deterministic scheduling, stochastic scheduling, and online
  deterministic scheduling.
\newblock In {\em Handbook of Scheduling}. Chapman and Hall/CRC, 2004.

\bibitem{ShmoysT93}
D.~B. Shmoys and {\'{E}}.~Tardos.
\newblock An approximation algorithm for the generalized assignment problem.
\newblock {\em Math. Program.}, 62:461--474, 1993.

\end{thebibliography}

\end{document}